\documentclass[12pt]{article}
\usepackage[authoryear]{natbib}
\usepackage{graphicx,bm,multirow}
\usepackage{amsthm,amsfonts,amssymb,color,comment}
\usepackage{amsmath}
\usepackage{comment}
\usepackage[hidelinks]{hyperref} 
\usepackage{booktabs, siunitx}
\usepackage[a4paper, total={6.5in, 9in}]{geometry}
\usepackage{setspace}
\usepackage{algorithm}
\usepackage{algpseudocode}
\usepackage{multirow}
\usepackage{enumitem} 
\usepackage{caption}
\usepackage{xcolor}
\usepackage{rotating}
\usepackage[affil-it]{authblk}
\usepackage{adjustbox}

\newtheorem{Theorem}{Theorem}
\newtheorem{Lemma}{Lemma}

\newtheorem{Corollary}{Corollary}

\newtheorem{Assumption}{Assumption}
\newtheorem{Remark}{Remark}

\usepackage{tikz}

\usepackage{subcaption}

\numberwithin{equation}{section}

\usepackage{chngcntr}

\definecolor{Red}{rgb}{.9,0,0}

\begin{document}
\title{A Bayesian Approach to Joint Estimation of Multiple Graphical Models}

\author{Peyman Jalali, Kshitij Khare\footnote{Kshitij Khare was partially supported by NSF grant DMS-119545} \ and George Michailidis}
\affil{Department of Statistics and Informatics Institute\\
University of Florida}  \date{}

\maketitle

\onehalfspacing

\begin{abstract}
The problem of joint estimation of multiple graphical models from high dimensional data has been studied in the statistics and machine learning literature, due to its importance in diverse fields including molecular biology, neuroscience and the social sciences. 
This work develops a Bayesian approach that decomposes the model parameters across the
multiple graphical models into shared components across subsets of models and edges, and idiosyncratic ones. Further, it leverages a novel multivariate 
prior distribution, coupled with a pseudo-likelihood that enables fast computations through a robust and efficient Gibbs sampling scheme.
We establish strong posterior consistency for model selection, as well as estimation of model parameters under high dimensional scaling with the number of variables 
growing exponentially with the sample size. The efficacy of the proposed approach is illustrated on both synthetic and real data.

 \vspace{0.2cm}
\end{abstract}

\section{Introduction}
The problem of {\em joint estimation} of multiple {\em related} Gaussian graphical models has attracted a lot of interest in the statistics and machine learning due to its wide application in biomedical studies involving Omics data -e.g. \cite{pierson2015sharing} and \cite{kling2015efficient}-, as well as text mining and roll call voting (\cite{guo2011joint}). The key idea is to ``borrow strength" across the related models and thus enhance the ``effective" sample size used for estimation of the model parameters, which is achieved primarily by the use of various penalty functions. Specifically, \cite{guo2011joint}, who first formulated the
problem, modeled the elements of each inverse covariance matrix as a product of a {\em common} across all models component and an {\em idiosyncratic} (model specific) component and imposed an $\ell_1$ penalty on each one; thus, when the penalty sets the common component to zero, the corresponding edge is absent
across all models, whereas if the common component is not zero, edges can be absent because the penalty sets the idiosyncratic one to zero for selected models. Another set of approaches aims to achieve a certain amount of ``fusing" across all models under consideration, thus focusing both of the presence
of common edges, as well as their absence across {\em all} models simultaneously. Examples of such approaches include \cite{danaher2014joint} that employed a group lasso and/or a fused lasso penalty on each edge parameter across all models and \cite{cai2016joint} that used a mixed $\ell_1/\ell_{\infty}$
norm for the same task. Variants of the above approaches with modifications to the penalties have also been explored (\cite{zhu2014structural}, \cite{majumdar2018joint}).

However, in many application settings, shared connectivity patterns across models occurs only for a subset of edges, while the remaining ones exhibit different connectivity patterns in each model. In other settings, subsets of edges share common connectivity patterns across only a subset of models. In both
instances, the previously mentioned approaches will exhibit a rather poor performance in discovering these more complex patterns. To address this issue, \cite{ma2016joint} proposed a {\em supervised} approach based on fusing through a group lasso penalty, wherein the various connectivity patterns across 
subsets of edges and subsets of models are {\em a priori} known. An alternative supervised approach \cite{saegusa2016joint} employed a similarity graph penalty for fusing across models, coupled with an $\ell_1$ penalty for obtaining sparse model estimates. The similarity graph is assumed to be a priori
known. A Bayesian variant of the latter approach was introduced in  \cite{peterson2015bayesian}, wherein  a Markov random field prior distribution was used to capture model similarity, followed by a spike-and-slab prior distribution on the edge model parameters. 
Another Bayesian approach was recently developed in \cite{tan2017bayesian} which, similar to \cite{peterson2015bayesian}, uses $G$-Wishart prior distributions on the group-wise precision matrices given the sparsity patterns/networks in each group,  and then employs a multiplicative model based hierarchical prior distribution on these networks to induce similarity/dependence.

Most of the frequentist approaches reviewed above come with performance guarantees in the form of high probability errors bounds for the model parameters as a function of the number of variables (nodes), sparsity level and sample size. Some recent work has focused on constructing confidence intervals
for the difference in the magnitude of the edge parameter across two models that are estimated {\em separately} using an $\ell_1$ penalization. 
On the other hand, theoretical guarantees based on high-dimensional posterior consistency results are not available for the Bayesian approaches mentioned above. Also, these approaches can suffer from 
computational scalability/efficiency issues in moderate/high dimensional settings, say in the presence of $p > 30$ nodes/variables. 

The key objective of this paper is to develop a {\it scalable} approach to jointly estimate multiple related Gaussian graphical models that exhibit complex edge connectivity patterns across models for different {\em subsets of edges}. 
To that end, we introduce a novel {\em Subset Specific ($\mathcal{S}^2$) prior} that for each edge aims to select the 
subset of models it is common to. We couple it with the Gaussian pseudo-likelihood used in \cite{khare2015convex} for estimating a single Gaussian graphical model, that leads to an easy to implement and scalable Gibbs sampling scheme for exploring the posterior distribution. Finally, we establish strong posterior model selection
consistency results that can be leveraged for construction of credible intervals for the edge parameters. Intuitively, the proposed framework achieves the objectives set forth in the \cite{ma2016joint} work, {\em without} requiring a priori specification of the shared edge connectivity patterns; thus, the approach
can be considered as fully unsupervised.

The remainder of the paper is organized as follows. Section 2 formulates the problem, while Section 3 introduces the 
$\mathcal{S}^2$ prior. In Section 4, we combine this prior with a Gaussian pseudo-likelihood to obtain a (pseudo) 
posterior distribution of the model parameters, and discuss how to sample from the posterior distribution. Section 5 
establishes high-dimensional consistency properties for this posterior distribution. Section 6 presents extensive numerical 
results based on synthetic data for the framework's performance and comparisons with existing approaches in the literature. 
Section 7 presents an application to metabolomics data from a case-control study on Inflammatory Bowel Disease.

\section{Framework for Joint Estimation}

\noindent
Suppose we have data from $K$ {\it a priori} defined groups. For each group $k$ ($k=1,2,...,K$), let $\mathcal{Y}_k := 
\left\{ \mathbf{y}_{i:}^k\right\}_{i=1}^{n_k}$ denote $p$-dimensional {\it i.i.d} observations from a multivariate normal 
distribution, with mean $\boldsymbol{0}$ and covariance matrix $\left(\mathbf{\Omega}^k\right)^{-1}$, which is specific to 
group $k$. Based on the discussion in the introductory section, the $K$ precision matrices 
$\left\{ \mathbf{\Omega}^k \right\}_{k=1}^K$ can share common patterns across subsets of the $K$ models, as delineated 
next. Our goal is to account for these shared structures. 

Let $\mathcal{P}\left(K\right)$ denote the power set of $\{1,\cdots,K$\} and for $k=1, ..., K$, define $\vartheta_k$ as follows:
\begin{equation}
\vartheta_k = \left\{ r \in \mathcal{P}\left(K\right)\setminus \{0\}: k\in r \right\}, \quad \quad k = 1, ..., K.
\end{equation}
It is easy to check that each $\vartheta_k$ is the collection of subsets which contain $k$, and has 
$\sum_{i=0}^{K-1}
\begin{pmatrix}
K-1 \\ i
\end{pmatrix} 
= 2^{K-1}$ 
members. 
Denote by $\mathbf{\Psi}^{r}$ the matrix that contains common patterns amongst precision matrices 
$\left\{ \mathbf{\Omega}^j \right\}_{j \in r}$. Specifically, for any singleton set $r =\{k\}$, the matrix $\mathbf{\Psi}^{r}$ contains a pattern that is unique to  group $k$, while for any other set $r$ containing more than a single
element, $\mathbf{\Psi}^{r}$ captures edge connectivity patterns (and their magnitudes) that are common across all members in $r$. For example, $\mathbf{\Psi}^{123} := \mathbf{\Psi}^{\left\{1,2,3\right\}}$ contains shared structures in $\mathbf{\Omega}^1$, $\mathbf{\Omega}^2$, and $\mathbf{\Omega}^3$.

Therefore, each precision matrix $\mathbf{\Omega}^k$ can be decomposed as
\begin{equation}\label{eq21}
\mathbf{\Omega}^{k} = \sum\limits_{r \in \mathop \vartheta \nolimits_k } { \mathbf{\Psi}^{r}}, \quad \quad \quad k=1, ..., K,
\end{equation}
where $\sum\limits_{r \in \mathop \vartheta \nolimits_k } { \mathbf{\Psi}^{r}}$ accounts for all the structures in $\mathbf{\Omega}^{k}$ which are either unique to group $k$ (i.e. $\mathbf{\Psi}^k$) or are shared exclusively between group $k$ and some combination of other groups (i.e. $\sum\limits_{r \in \vartheta_k\setminus \{k\}  } { \mathbf{\Psi}^{r}}$).
We further assume that $\mathbf{\Psi}^k \in \mathbb{M}_{p}^{+}$ for $k=1,2,...,K$, where $\mathbb{M}_{p}^{+}$ denotes the space of all $p \times p$ matrices with positive diagonal entries. Finally, the diagonal entries of every joint matrix $\mathbf{\Psi}^r$, with $r \in \cup_{k=1}^{K}\left(\vartheta_k\setminus\{k\}\right)$ are set to zero; in other words, the diagonals entries of $\mathbf{\Omega}^k$ are contained in the corresponding $\mathbf{\Psi}^{k}$.

To illustrate the notation, consider the case of $K=3$ groups. Then, each precision matrix is decomposed as 
\begin{equation*}
\begin{split}
\mathbf{\Omega}^1 &= \mathbf{\Psi}^1 + \mathbf{\Psi}^{12} + \mathbf{\Psi}^{13} + \mathbf{\Psi}^{123}\\
\mathbf{\Omega}^2 &= \mathbf{\Psi}^2 + \mathbf{\Psi}^{12} + \mathbf{\Psi}^{23} + \mathbf{\Psi}^{123}\\
\mathbf{\Omega}^3 &= \mathbf{\Psi}^3 + \mathbf{\Psi}^{13} + \mathbf{\Psi}^{23} + \mathbf{\Psi}^{123}
\end{split}
\end{equation*}
where the $\mathbf{\Psi}^1$, $\mathbf{\Psi}^2$, and $\mathbf{\Psi}^3$ matrices contain group specific patterns, the $\mathbf{\Psi}^{12}$, $\mathbf{\Psi}^{13}$, $\mathbf{\Psi}^{23}$ matrices contain patterns shared across pairs of models (for subsets of the edges) and finally matrix $\mathbf{\Psi}^{123}$ contains patterns
shared across all models/groups.

\subsection{Identifiability Considerations} \label{Identifiability}

\noindent
A moment of reflection shows that the model decomposition  (\ref{eq21}) is not unique.  For example, for any arbitrary matrix $\mathbf{X}$, the model (\ref{eq21}) is equivalent to $\mathbf{\Omega}^{k} = \sum\limits_{r \in \mathop \vartheta \nolimits_k } {\mathop \mathbf{\Phi}^r }$ with $\mathbf{\Phi}^r = \mathbf{\Psi}^r + \mathbf{X}$ and $\mathbf{\Phi}^k = \mathbf{\Psi}^k - \frac{1}{2^{K-1} - 1}\mathbf{X}$. Hence, without imposing appropriate identifiability constraints, meaningful estimation of the model parameters is not feasible.

In order to address this issue, we first rewrite the element-wise representation of model (\ref{eq21}):
\begin{equation}\label{eq22}
{\omega}^{k}_{ij} = \sum\limits_{r \in \vartheta_k } {{\psi}^{r}_{ij} }, \quad \quad 1 \le i < j \le p, 1 \leq k \leq K,
\end{equation}
where ${\omega}^{k}_{ij}$, and ${\psi}^r_{ij}$ are the $ij^{\text{th}}$ coordinates of the matrices $\mathbf{\Omega}^k$ and $\mathbf{\Psi}^r$, respectively. We only consider the upper off-diagonal entries due to the symmetry of the precision matrix and
thus define vectors $\boldsymbol{\theta}_{ij}$ for every $1\leq i < j \leq p$, as follows
\begin{equation}\label{theta_ij}
\theta_{ij} = \left\{ \psi_{ij}^r \right\}_{r \in \mathcal{P}\left(K\right) \setminus\{0\}},
\end{equation}
where each $\boldsymbol{\theta}_{ij}$ has $2^{K}-1$ distinct parameters. For {\em identifiability purposes} we require that each vector $\boldsymbol{\theta}_{ij}$ has at most {\em one non-zero element}. Note that under this constraint, if an edge $(i,j)$ is shared amongst
many groups, the non-zero element will be allocated to the maximal set $s \in \cup_{k=1}^{K}\left(\vartheta_k\setminus\{k\}\right)$, while {\em all subsets} of $s$ will be allocated a zero value. Further, the magnitude of {\em all} $\{\omega_{ij}^k\}_{k\in s}$ will be the same. As an example, consider again the case of $K=3$ groups and an edge $(i,j)$ shared amongst all three groups. In this case, the edge will be allocated to the $\mathbf{\Psi}^{123}$ component and not to any other components, such as 
$\mathbf{\Psi}^{12}$ or $\mathbf{\Psi}^{13}$. Hence $\mathbf{\Psi}^{123}_{ij}$ will be non-zero, but 
$$
\mathbf{\Psi}^{12}_{ij} = \mathbf{\Psi}^{13}_{ij} = \mathbf{\Psi}^{23}_{ij} = \mathbf{\Psi}^{1}_{ij} = \mathbf{\Psi}^{2}_{ij} = 
\mathbf{\Psi}^{3}_{ij} = 0. 
$$

\noindent
Next, we construct a novel prior distribution that respects the introduced identifiability constraint.

\section{Subset Specific ($\mathcal{S}^2$) Prior Distribution}

\noindent
For any generic symmetric $p\times p$ matrix $\mathbf{A}$, define
\begin{equation*}
\begin{split}
\underline{\boldsymbol{a}} = \left( a_{12}, a_{13}, ..., a_{p-1p}\right), 
\boldsymbol{\delta}_{\mathbf{A}} = \left(a_{11}, ..., a_{pp}\right),
\end{split}
\end{equation*}
where due to the symmetric nature of $\mathbf{A}$, the vector $\underline{\boldsymbol{a}}$ contains all the off-diagonal elements, while $\boldsymbol{\delta}_{\mathbf{A}}$ the diagonal ones. In particular, ${\underline{\boldsymbol{\psi}}}^r$ is the vectorized version of the off-diagonal elements of $\mathbf{\Psi}^r$. 

Using the above notation, define $\mathbf{\Theta}$ to be the vector obtained by combining the vectors $\left\{ {\underline{\boldsymbol{\psi}}}^r, r\in \mathcal{P}\left(K\right) \setminus \{0\}\right\}$. To illustrate, for $K=3$ groups, $\mathbf{\Theta}$ is given by 
\begin{equation}\label{theta}
{\mathbf{\Theta}} = ( {{\underline{\boldsymbol{\psi}}}^{123}}, {{\underline{\boldsymbol{\psi}}}^{23}}, {{\underline{\boldsymbol{\psi}}}^{13}}, {{\underline{\boldsymbol{\psi}}}^{12}},{{\underline{\boldsymbol{\psi}}}^{3}},{{\underline{\boldsymbol{\psi}}}^{2}},{{\underline{\boldsymbol{\psi}}}^{1}})'.
\end{equation}
In view of (\ref{theta_ij}), it can be easily seen that that $\mathbf{\Theta}$ is a rearrangement of the vector $(\boldsymbol{\theta}_{12}, \boldsymbol{\theta}_{13}, ...,$ $ \boldsymbol{\theta}_{p-1p})'$. Thus, according to the location of the zero coordinates in $\boldsymbol{\theta}_{ij}$ ($2^{K}$ possibilities), 
there are $2^{\frac{Kp(p-1)}{2}}$ possible sparsity patterns across the $K$ groups that $\mathbf{\Theta}$ can have. Let $\mathbf{\ell}$ be a generic sparsity pattern for $\mathbf{\Theta}$ and denote the set of all the $2^{\frac{Kp(p-1)}{2}}$ sparsity patterns by $\mathcal{L}$. To illustrate, consider $K=2$ groups and $p=3$ variables. In this case, each matrix has 3 off-diagonal edges ($\{ij: 1\leq i < j \leq p\} = \{12, 13, 23\}$). Now, assume edge $12$ is shared between the two groups, edge 13 is unique to group 2, and edge 23 is absent in both groups. In this case, $\mathbf{\Theta}$ is given by 
$
\mathbf{\Theta} = \left( \left(\psi_{12}^{12}, 0,0\right), \left(0,\psi^2_{13},0\right), \left(0, 0,0\right) \right)';
$
and the sparsity pattern extracted from $\mathbf{\Theta}$ is as follows: 
\begin{equation*}
\boldsymbol{\ell} = \left( \left( 1, 0, 0\right), \left(0, 1, 0\right), \left(0, 0, 0\right)\right)'.
\end{equation*}
For every sparsity pattern $\boldsymbol{\ell}$, let $d_{\boldsymbol{\ell}}$ be the density (number of non-zero entries) of 
$\boldsymbol{\ell}$, and $\mathcal{M}_{\boldsymbol{\ell}}$ be the space where $\mathbf{\Theta}$ varies, when 
restricted to follow the sparsity pattern $\boldsymbol{\ell}$. We specify the hierarchical prior distribution ${S}^2$ as follows:
\begin{equation}\label{hierarchical prior 1}
\pi \left( \mathbf{\Theta} | \boldsymbol{\ell}\right) = \frac{|\mathbf{\Lambda}_{\boldsymbol{\ell}\boldsymbol{\ell}}|^{\frac{1}{2}}}{\left(2\pi \right)^{\frac{d_{\boldsymbol{\ell}}}{2}}} \exp \left(-\frac{\mathbf{\Theta}'\mathbf{\Lambda}\mathbf{\Theta}}{2} \right)I_{\left( \mathbf{\Theta} \in \mathcal{M}_{\ell}\right)},
\end{equation}

\begin{equation}\label{hierarchical prior 2}
\pi(\boldsymbol{\ell}) \propto \begin{cases}
q_1^{d_{\boldsymbol{\ell}}}(1-q_1)^{\binom{p}{2} - d_{\boldsymbol{\ell}}} & d_{\boldsymbol{\ell}} \leq \tau, \\
q_2^{d_{\boldsymbol{\ell}}}(1-q_2)^{\binom{p}{2} - d_{\boldsymbol{\ell}}} & d_{\boldsymbol{\ell}} > \tau,
\end{cases}
\end{equation}
where $\mathbf{\Lambda}$ is a diagonal matrix whose entries determine the amount of shrinkage imposed on the 
corresponding elements in $\mathbf{\Theta}$, $\mathbf{\Lambda}_{\boldsymbol{\ell}\boldsymbol{\ell}}$ is a sub-matrix of $
\mathbf{\Lambda}$ obtained after removing the rows and columns corresponding to the zeros in $\mathbf{\Theta}\in 
\mathcal{M}_{\boldsymbol{\ell}}$, and $q_1$ and $q_2$ are edge inclusion probabilities, respectively, for the case of sparse 
($d_{\boldsymbol{\ell}} \leq \tau$) and dense ($d_{\boldsymbol{\ell}} > \tau$) $\mathbf{\Theta}$. Later in our theoretical 
analysis, we specify values for $q_1$, $q_2$, and the threshold $\tau$. Let $\mathbf{\Theta}_{\boldsymbol{\ell}}$ be the 
vector containing the non-zero coordinates of $\mathbf{\Theta}\in \mathcal{M}_{\boldsymbol{\ell}}$. Then, the prior in 
(\ref{hierarchical prior 1}) corresponds to putting an independent normal prior on each entry of 
$\mathbf{\Theta}_{\boldsymbol{\ell}}$. 

Using the prior distribution posited in (\ref{hierarchical prior 1}) and (\ref{hierarchical prior 2}), we derive the marginal prior distribution on $\mathbf{\Theta}$, as follows
\begin{equation}\label{prior on theta}
\begin{split}
\pi\left( \mathbf{\Theta} \right)  = \sum\limits_{ \boldsymbol{\ell} \in \mathcal{L}}\pi\left(\mathbf{\Theta} | \boldsymbol{\ell}\right)\pi\left(\boldsymbol{\ell}\right)
&=\sum\limits_{ \boldsymbol{\ell} \in \mathcal{L}} \left\{ \frac{|\mathbf{\Lambda}_{\boldsymbol{\ell}\boldsymbol{\ell}}|^{\frac{1}{2}}}{\left(2\pi \right)^{\frac{d_{\boldsymbol{\ell}}}{2}}} \exp \left(-\frac{\mathbf{\Theta}'\mathbf{\Lambda}\mathbf{\Theta}}{2} \right)I_{\left( \mathbf{\Theta} \in \mathcal{M}_{\ell}\right)} \right\} \pi \left( \boldsymbol{\ell} \right)\\
&= \sum\limits_{ \boldsymbol{\ell} \in \mathcal{L}} u_\ell \frac{|\mathbf{\Lambda}_{\boldsymbol{\ell}\boldsymbol{\ell}}|^{\frac{1}{2}}}{\left(2\pi \right)^{\frac{d_{\boldsymbol{\ell}}}{2}}}  \exp \left(-\frac{\mathbf{\Theta}_\ell'\mathbf{\Lambda}_{\ell\ell}\mathbf{\Theta}_\ell}{2} \right) I_{\left( \mathbf{\Theta}_\ell \in \mathbb{R}^{d_\ell}\right)}, 
\end{split}
\end{equation}

\noindent
where 
\begin{equation}\label{mixing probs}
u_{\ell} =  \frac{ q_1^{d_{\boldsymbol{\ell}}}(1-q_1)^{\binom{p}{2} - d_{\boldsymbol{\ell}}} I_{\left\{d_{\boldsymbol{\ell}}\leq \tau\right\}} + q_1^{d_{\boldsymbol{\ell}}}(1-q_1)^{\binom{p}{2} - d_{\boldsymbol{\ell}}} I_{\left\{d_{\boldsymbol{\ell}}\leq \tau\right\}}}{\sum\limits_{\ell \in \mathcal{L}} \left[ q_1^{d_{\boldsymbol{\ell}}}(1-q_1)^{\binom{p}{2} - d_{\boldsymbol{\ell}}} I_{\left\{d_{\boldsymbol{\ell}}\leq \tau\right\}} + q_1^{d_{\boldsymbol{\ell}}}(1-q_1)^{\binom{p}{2} - d_{\boldsymbol{\ell}}} I_{\left\{d_{\boldsymbol{\ell}}\leq \tau\right\}}\right]}. 
\end{equation}

\noindent
In other words, $\pi (\mathbf{\Theta})$ can be regarded as a mixture of $2^{\frac{Kp(p-1)}{2}}$ multivariate normal 
distributions of dimensions $d_\ell$ that is obtained after projecting a larger dimension ($\frac{p(p-1)}{2}(2^{K}-1)$) multivariate normal 
distribution  into the union of all the subspaces $\mathcal{M}_\ell$; namely, 
$\bigcup\limits_{ \boldsymbol{\ell} \in \mathcal{L}}\mathcal{M}_{\boldsymbol{\ell}}$. Note that the $\mathcal{S}^2$ prior 
induces sparsity on $\mathbf{\Theta}$, which will be helpful for model selection purposes. Further, the prior respects the 
identifiability constraint by forcing {\it at least} $\frac{p(p-1)}{2}(2^{K}-2)$  parameters to be exactly equal to zero. In addition to 
forcing sparsity, the diagonal entries of $\mathbf{\Lambda}_{\ell\ell}$ enforce shrinkage to the corresponding elements in 
$\mathbf{\Theta}_{\ell}$. We shall discuss the selection of these shrinkage parameters later in 
Section \ref{sec:select}. 

Note that the vector $\mathbf{\Theta}$ only incorporates the off-diagonal entries of $\mathbf{\Psi}$ matrices. Regarding the 
diagonal entries, for every $k \in \{1, ...,K\}$, we let $\boldsymbol{\delta}_{\mathbf{\Psi}^k}$ be the vector 
comprising of the diagonal elements of the matrix ${\mathbf{\Psi}^k}$ and define $\mathbf{\Delta}$ to be the vector of all 
diagonal vectors $\boldsymbol{\delta}_{\mathbf{\Psi}^K}$, i.e.
\begin{equation}\label{delta}
\mathbf{\Delta} = \left( \boldsymbol{\delta}_{\mathbf{\Psi}^1}, ..., \boldsymbol{\delta}_{\mathbf{\Psi}^K}\right).
\end{equation}

\noindent
We assign an independent Exponential($\gamma$) prior on each coordinate of $\mathbf{\Delta}$ (diagonal element of the matrices $\mathbf{\Psi}_{k}$, $k=1, ...,K$), i.e,
\begin{equation}\label{prior on delta}
\pi\left( \mathbf{\Delta} \right) \propto \exp \left( -\gamma \boldsymbol{1}'\mathbf{\Delta} \right) I_{\mathbb{R}_{+}^{Kp}}\left( \mathbf{\Delta} \right).
\end{equation}

\noindent
The selection of the hyperparameter $\gamma$ is also discussed in Section \ref{sec:select}. Since the diagonal entries of 
every joint matrix $\mathbf{\Psi}^r$, with $r \in \cup_{k=1}^{K}\left(\vartheta_k\setminus\{k\}\right)$ are set to zero, the 
specification of the prior is now complete.

\section{The Bayesian Joint Network Selector (BJNS)}

\noindent
Estimation of the model parameters $(\mathbf{\Theta}, \mathbf{\Delta})$ is based on a pseudo-likelihood approach. The 
pseudo-likelihood, which is based on the regression interpretation of the entries of $\Omega$, can be regarded as a weight 
function and as long as the product of the pseudo-likelihood and the prior density is integrable over the parameter space, 
one can construct a (pseudo) posterior distribution and carry out  Bayesian inference. The main advantage of using a 
pseudo-likelihood, as opposed to a full Gaussian likelihood, is that it allows for an easy to implement sampling scheme 
from the posterior distribution and in addition provides more robust results under deviations from the Gaussian assumption, 
as illustrated by work in the frequentist domain \cite{khare2015convex,peng2009partial}. Note that the pseudo-likelihood 
approach does not respect the positive definite constraint on the precision matrices under consideration, but since our 
primary goal is estimating the skeleton of the underlying graphs this mitigates this issue. Further, accurate estimation of 
the magnitude of the estimated edges can be accomplished through a refitting step of the model parameters restricted to 
the skeleton, as shown in \cite{ma2016joint}. We will also establish high-dimensional sparsity selection and estimation 
consistency for our procedure later in Section \ref{sec:consistency}. 

Let $\mathbf{S}^k$ denote the sample covariance matrix of the observations in the $k^{th}$ group. Based on the above 
discussion, we employ the CONCORD pseudo-likelihood introduced in \cite{khare2015convex}, 
\begin{equation*} 
\begin{split}
 \exp \left\{ n\sum\limits_{j=1}^{p}{\text{log}\omega^k_{jj}} - \frac{n}{2} \text{tr} \left[  \left( \mathbf{\Omega}^k\right)^2 \mathbf{S}^k\right]\right\}, \quad \quad k = 1, ..., K,
\end{split}
\end{equation*}
and the model specification (\ref{eq21}) to construct the joint pseudo-likelihood function for $K$ precision matrices, as follows,
\begin{equation}\label{joint pseudo likelihood}
\prod\limits_{k=1}^{K}{\exp \left\{n\sum\limits_{j=1}^{p}{\text{log}\psi^{k}_{jj}} - \frac{n}{2}\text{tr}\left[ \left(\sum\limits_{r \in \mathop \vartheta \nolimits_k } { \mathbf{\Psi}^{r} }\right)^2 \mathbf{S}^k \right]\right\}}.
\end{equation}

\noindent
Since we have parametrized the $\mathcal{S}^2$ prior in terms of $\left( \mathbf{\Theta}, \mathbf{\Delta}\right)$, we will rewrite the above pseudo-likelihood function in terms of $\left( \mathbf{\Theta}, \mathbf{\Delta}\right)$, as well. Some straightforward algebra shows that

\begin{equation}\label{Theta'UpsilonTheta}
\begin{split}
\text{tr}\left[ \left( \sum\limits_{r \in \mathop \vartheta \nolimits_k } { \mathbf{\Psi}^{r} }\right)^2 \mathbf{S}^k \right]  =\left( {\begin{array}{*{20}{c}}
{\mathbf{\Theta}'} &
{\mathbf{\Delta}'} \\
\end{array} } \right) \left( {\begin{array}{*{20}{c}}
{\mathbf{\Upsilon}} & {\mathbf{A}} \\
{\mathbf{A}'} & {\mathbf{D}} \\
\end{array} } \right) \left( {\begin{array}{*{20}{c}}
{\mathbf{\Theta}} \\
{\mathbf{\Delta}} \\
\end{array} } \right),
\end{split}
\end{equation}

where, $\mathbf{\Upsilon}$ is a $\frac{p(p-1)(2^K-1)}{2}\times\frac{p(p-1)(2^K-1)}{2}$ symmetric matrix whose entries are either zero or a linear combination of $\left\{ s_{ij}^k \right\}_{1 \leq i < j \leq p}^{1 \leq k \leq K}$; $\mathbf{D}$ is a $Kp \times Kp$ diagonal matrix with entries $\left\{ s_{ii}^k \right\}_{1 \leq i \leq p}^{1 \leq k \leq K}$; $\boldsymbol{a}$ is a $\frac{p(p-1)(2^K-1)}{2}\times 1$ vector whose entries are depend on $\mathbf{\Delta}$ and $\left\{ s_{ij}^k \right\}_{1 \leq i < j \leq p}^{1 \leq k \leq K}$; and finally $\mathbf{A}$ is a $\frac{p(p-1)(2^K-1)}{2}\times Kp$ matrix such that $\mathbf{A}\mathbf{\Delta} = \boldsymbol{a}$. To see the algebraic details of the equality in (\ref{Theta'UpsilonTheta}), structures of $\mathbf{\Upsilon}$ and $\boldsymbol{a}$ we refer the reader to Supplemental section S1. 

Now, letting $\mathcal{Y} :=\left(\{\mathbf{y}_{i:}^1\}_{i=1}^{n_1}, ..., \{\mathbf{y}_{i:}^K\}_{i=1}^{n_K}\right)$ and by applying Bayes' rule, the posterior distribution of $\left(\mathbf{\Theta}, \mathbf{\Delta} \right)$ is given by

\begin{equation}\label{posterior theta and delta}
\begin{split}
\pi \left\{ \left(\mathbf{\Theta}, \mathbf{\Delta} \right)| \mathcal{Y}\right\} &\propto \exp \left\{ n \boldsymbol{1}'\log \left( \mathbf{\Delta}\right) - \frac{n}{2}\left[ \left( {\begin{array}{*{20}{c}}
{\mathbf{\Theta}'} &
{\mathbf{\Delta}'} \\
\end{array} } \right) \left( {\begin{array}{*{20}{c}}
{\mathbf{\Upsilon}} & {\mathbf{A}} \\
{\mathbf{A}'} & {\mathbf{D}} \\
\end{array} } \right) \left( {\begin{array}{*{20}{c}}
{\mathbf{\Theta}} \\
{\mathbf{\Delta}} \\
\end{array} } \right) \right] \right\} \\
 &\times \exp \left(-\frac{\mathbf{\Theta}'\mathbf{\Lambda}\mathbf{\Theta}}{2} \right) \sum\limits_{ \boldsymbol{\ell} \in \mathcal{L}} \left\{ \frac{|\mathbf{\Lambda}_{\boldsymbol{\ell}\boldsymbol{\ell}}|^{\frac{1}{2}}}{\left(2\pi \right)^{\frac{d_{\boldsymbol{\ell}}}{2}}} I_{\left( \mathbf{\Theta} \in \mathcal{M}_{\ell}\right)}\right. \\
&\times \left. \left[ q_1^{d_{\boldsymbol{\ell}}}(1-q_1)^{\binom{p}{2} - d_{\boldsymbol{\ell}}} I_{\left\{d_{\boldsymbol{\ell}}\leq \tau\right\}} + q_1^{d_{\boldsymbol{\ell}}}(1-q_1)^{\binom{p}{2} - d_{\boldsymbol{\ell}}} I_{\left\{d_{\boldsymbol{\ell}}\leq \tau\right\}}\right]\right\}\\
&\times \exp \left( -\gamma \boldsymbol{1}'\mathbf{\Delta} \right).
\end{split}
\end{equation}
Moreover, the conditional posterior distribution of $\mathbf{\Theta}$ given $\mathbf{\Delta}$ is given by
\begin{equation}\label{posterior theta given delta}
\begin{split}
&\pi \left\{ \mathbf{\Theta} | \mathbf{\Delta}, \mathcal{Y}\right\} \propto \exp \left\{ - \frac{1}{2}\left[\mathbf{\Theta}'\left( n\mathbf{\Upsilon} + \mathbf{\Lambda}\right) \mathbf{\Theta} + 2n\mathbf{\Theta}'\boldsymbol{a}\right] \right\} \times\\
&\sum\limits_{ \boldsymbol{\ell} \in \mathcal{L}} \left\{ \frac{|\mathbf{\Lambda}_{\boldsymbol{\ell}\boldsymbol{\ell}}|^{\frac{1}{2}}}{\left(2\pi \right)^{\frac{d_{\boldsymbol{\ell}}}{2}}} I_{\left( \mathbf{\Theta} \in \mathcal{M}_{\boldsymbol{\ell}}\right)}  \left[ q_1^{d_{\boldsymbol{{\boldsymbol{\ell}}}}}(1-q_1)^{\binom{p}{2} - d_{\boldsymbol{\ell}}} I_{\left\{d_{\boldsymbol{\ell}}\leq \tau\right\}} + q_1^{d_{\boldsymbol{\ell}}}(1-q_1)^{\binom{p}{2} - d_{\boldsymbol{\ell}}} I_{\left\{d_{\boldsymbol{\ell}}\leq \tau\right\}}\right]\right\},
\end{split}
\end{equation}
while that of $\mathbf{\Delta}$ given $\mathbf{\Theta}$ by
\begin{equation}\label{posterior delta given theta}
\begin{split}
\pi \left\{ \mathbf{\Delta} | \mathbf{\Theta}, \mathcal{Y}\right\} \propto \prod\limits_{k=1}^K \prod\limits_{i=1}^p \left( \psi^k_{ii}\right)^n \exp \left\{ -\frac{n}{2} s^k_{ii}\left( \psi^k_{ii}\right)^2 - \left(\gamma + n\sum_{j \ne i}{\omega^k_{ij}s^k_{ij}}\right) \psi^k_{ii}\right\},
\end{split}
\end{equation}
where $\omega^k_{ij} = \sum\limits_{r\in\vartheta_k} \psi^r_{ij}$, for $1 \leq i < j \leq p$ and $1 \leq k \leq K$.

Next, we discuss a Gibbs sampling algorithm to generate approximate samples from the above posited posterior distribution.

\subsection{Gibbs Sampling Scheme for BJNS}

Generating exact samples from the multivariate distribution in (\ref{posterior theta and delta}) is not computationally feasible. 
We instead generate approximate samples from (\ref{posterior theta and delta}), by computing the full 
conditional distribution of the vectors $\boldsymbol{\theta}_{ij}$s, ($1 \le i < j \le p$) and that of the diagonal entries 
$\psi^{k}_{ii}$ ($1 \leq i \leq p$, $1 \leq k \leq K$). 

As discussed earlier, each $\boldsymbol{\theta}_{ij}$ contains $2^{K}-1$ elements of which at most one is non-zero. For ease of exposition, let $\theta_{l,ij}$ denote the $l^{\text{th}}$ element of $\boldsymbol{\theta}_{ij}$ for $l=1, ..., 2^{K}-1$ (based on (\ref{theta_ij}), every $\theta_{l,ij}$ represents a $\psi^{r}_{ij}$, for some $r \in \mathcal{P}(K)$). Using the same notation for the shrinkage parameters (diagonal elements of $\mathbf{\Lambda}$), let $\lambda_{l,ij}$ be the shrinkage parameter corresponding to $\theta_{l,ij}$. Since there are $2^{K}$ possibilities for the location of the zeros in each $\boldsymbol{\theta}_{ij}$, one can see $\boldsymbol{\theta}_{ij}$ as an element in one of the disjoint spaces $\mathbb{M}_0$, $\mathbb{M}_1$, ..., $\mathbb{M}_{2^{K}-1}$, where $\mathbb{M}_0$ is the singleton set consisting of the zero vector of length $2^{K}-1$ and $\mathbb{M}_{l}$ ($l=1, ..., 2^{K}-1$) is the space spanned by the $l^{\text{th}}$ unit vector of length $2^{K}-1$. 

Denote by $\mathbf{\Theta}_{-(ij)}$ the sub vector of $\mathbf{\Theta}$ obtained by removing $\boldsymbol{\theta}_{ij}$. Define $\mathbf{\Upsilon}_{(ij)(ij)}$ to be the sub-matrix of $\mathbf{\Upsilon}$ obtained by removing the rows and columns with indices corresponding to the zero elements of $\boldsymbol{\theta}_{ij}$ inside $\mathbf{\Theta}$. Similarly, let $\mathbf{\Upsilon}_{(ij)(-(ij))}$ be the sub-matrix of $\mathbf{\Upsilon}$ obtained by removing the rows with indices corresponding to zero elements of $\boldsymbol{\theta}_{ij}$ inside $\mathbf{\Theta}$ and columns with indices corresponding to zero elements of $\mathbf{\Theta}_{-(ij)}$ inside $\mathbf{\Theta}$.
Further,  let $\boldsymbol{a}_{ij}$ be the sub-vector of $\boldsymbol{a}$ whose indices match those of $\boldsymbol{\theta}_{ij}$ inside $\mathbf{\Theta}$. Then, the conditional posterior distribution of $\boldsymbol{\theta}_{ij}$ given $ \mathbf{\Theta}_{-(ij)}$ and $\mathbf{\Delta}$ is given by
\begin{equation*}
\begin{split}
\pi &\left\{ \boldsymbol{\theta}_{ij}| \mathbf{\Theta}_{-(ij)}, \mathbf{\Delta} , \mathcal{Y} \right\} \\
&\propto \exp\left\{ -\frac{1}{2} \left[ \boldsymbol{\theta}_{ij}' \left(n\mathbf{\Upsilon} + \mathbf{\Lambda}\right)_{(ij) (ij)} \boldsymbol{\theta}_{ij} + 
2 n \boldsymbol{\theta}_{ij}' \left( \boldsymbol{a}_{ij} + \mathbf{\Upsilon}_{(ij) (-(ij))}\mathbf{\Theta}_{-(ij)} \right)\right]\right\} I_{\bigcup\limits_{{l} = 0}^{2^K - 1} \mathbb{M}_{l} }( \boldsymbol{\theta}_{ij}).
\end{split}
\end{equation*}
Note that since $\boldsymbol{\theta}_{ij}$ has at most one non-zero element, all cross products in $\boldsymbol{\theta}_{ij}' \left(n\mathbf{\Upsilon} + \mathbf{\Lambda}\right)_{(ij) (ij)} \boldsymbol{\theta}_{ij}$ are equal to zero, i.e.
\begin{equation*}
\boldsymbol{\theta}_{ij}' \left(n\mathbf{\Upsilon} + \mathbf{\Lambda}\right)_{(ij) (ij)} \boldsymbol{\theta}_{ij} = \sum_{{l} = 1}^{2^K - 1}\theta_{l,ij}^2 \left[ \left(n\mathbf{\Upsilon} + \mathbf{\Lambda}\right)_{(ij) (ij)}\right]_{{l}{l}} = 
\sum_{{l} = 1}^{2^K - 1} \theta_{l,ij}^2\{n \left[ \mathbf{\Upsilon}_{(ij) (ij)}\right]_{{l}{l}} + \lambda_{l,ij}\}
\end{equation*}
where $\left[ \mathbf{\Upsilon}_{(ij) (ij)}\right]_{{l}{l}}$  is the $l^{\text{th}}$ diagonal element of matrix $\mathbf{\Upsilon}_{(ij) (ij)}$, for $l=1, ..., 2^K-1$. Hence, denoting the univariate normal probability density function by $\phi$, we get
\begin{equation*}
\begin{split}
\pi &\left\{ \boldsymbol{\theta}_{ij}| \mathbf{\Theta}_{-(ij)}, \mathbf{\Delta} , \mathcal{Y} \right\} \\
&\propto \exp\left\{ -\frac{1}{2} \sum_{{l} = 1}^{2^K - 1}\left(\theta_{l,ij}^2 \{n \left[ \mathbf{\Upsilon}_{(ij) (ij)}\right]_{{l}{l}} + \lambda_{l,ij}\} + 
2 n \theta_{l,ij}\left[ \boldsymbol{a}_{ij} + \mathbf{\Upsilon}_{(ij) (-(ij))}\mathbf{\Theta}_{-(ij)} \right]_{{l}}\right)\right\} I_{\bigcup\limits_{{l} = 0}^{2^K - 1} \mathbb{M}_{l} }( \boldsymbol{\theta}_{ij})\\
& = I_{{\boldsymbol \theta}_{ij} \in \mathbf{M}_0} + \sum_{{l} =1}^{2^K-1} c_{{l}, ij} \phi\{\theta_{l,ij} ;\left( \mu_{{l}, ij}, \nu_{{l}, ij}^{2}\right)\}I_{{\boldsymbol \theta}_{ij} \in \mathbf{M}_l}, \quad 1 \le i < j \le p,
\end{split}
\end{equation*}
with
\begin{equation}\label{mixture probs}
\begin{split}
\mu_{{l},ij}&= - \frac{n\left( \boldsymbol{a}_{ij} + \mathbf{\Upsilon}_{(ij) (-(ij))}\mathbf{\Theta}_{-(ij)} \right)_{{l}}}{n \left[ \mathbf{\Upsilon}_{(ij) (ij)}\right]_{{l}{l}} + \lambda_{l,ij}}, \\
\nu_{{l},ij}^{2}&= \frac{1}{n \left[ \mathbf{\Upsilon}_{(ij) (ij)}\right]_{{l}{l}} + \lambda_{l,ij}},\\
c_{{l}, ij} &= \sqrt{2\pi\nu_{{l} ,ij}^{2}} \exp\left\{\frac{\mu_{{l},ij}^{2}}{2\nu_{{l},ij}^{2}} \right\}, 
\end{split}
\end{equation}
for ${l} = 0, ..., 2^K-1$, and $1 \le i < j \le p$. Hence, for $1 \le i < j \le p$ we can write,
\begin{equation}\label{posterior density of theta_ij}
\begin{split}
\pi &\left\{ \boldsymbol{\theta}_{ij}| \mathbf{\Theta}_{-(ij)}, \mathbf{\Delta}, \mathcal{Y} \right\} =\frac{ I_{{\boldsymbol \theta}_{ij} \in \mathbf{M}_0} + \sum_{{l} =1}^{2^K-1} c_{{l}, ij} \phi\{\theta_{l,ij} ;\left( \mu_{{l}, ij}, \nu_{{l}, ij}^{2}\right)\}I_{{\boldsymbol \theta}_{ij} \in \mathbf{M}_l}}{ 1 + \sum_{{l} =1}^{2^K-1} c_{{l}, ij} },
\end{split}
\end{equation}
The above density is a mixture of univariate normal densities and the cost of generating samples from this density is comparable to that of generating from a univariate normal distribution. 

In view of (\ref{posterior delta given theta}), one can also easily see that conditional on $\mathbf{\Theta}$, the diagonal entries $\psi^k_{ii}$ ($1 \le k \le K$, $1 \le i \le p$) are a posteriori independent and their conditional posterior density given $\mathbf{\Theta}$ is as follows,
\begin{equation}\label{posterior density of psi_ii}
\pi \left\{ \psi^k_{ii} | \mathbf{\Theta} , \mathcal{Y} \right\} \propto \exp \left\{n \log\left(  \psi^k_{ii}\right) -\frac{n}{2} s_{ii}^k \left(  \psi^k_{ii} \right)^2 - b_i^k  \psi^k_{ii}\right\},
\end{equation}
where $b_i^k = \gamma+ n\sum_{j \ne i}{\left( \sum\limits_{r\in\vartheta_k} \psi^r_{ij} \right)s^k_{ij}}$, for $1 \leq i < j \leq p$ and $1 \leq k \leq K$. 

Note that the density in (\ref{posterior density of psi_ii}) is not a standard density but an efficient algorithm to sample from this density is provided in Supplemental section S2.

\subsection{Selection of Shrinkage Parameters} \label{sec:select}

\noindent
Let $\theta$ and $\delta$ be generic elements of $\mathbf{\Theta}$ and $\mathbf{\Delta}$ and let $\lambda$ and $\gamma$ be their corresponding shrinkage parameters. Selecting appropriate values for the latter is an important task. In other Bayesian analysis of high dimensional models, 
shrinkage parameters are usually generated based on an appropriate prior distribution; see \cite{park2008bayesian,kyung2010penalized,hans2009bayesian}) for regression analysis and \cite{wang2012bayesian} for graphical models. We assign independent gamma priors on each shrinkage parameter $\lambda$ or $\gamma$;
specifically, $\lambda, \gamma \sim \text{Gamma}(r, s)$, for some hyper-parameters $r$ and $s$. The amount of shrinkage imposed on each element $\theta$ and $\delta$ can be calculated by considering the posterior distribution of $\lambda$ and $\gamma$, given $\left(\mathbf{\Theta}, \mathbf{\Delta}\right)$. Straightforward algebra shows that 
\begin{equation*}
\begin{split}
\lambda| \left(\mathbf{\Theta}, \mathbf{\Delta}\right) &\sim \text{Gamma}(r + 0.5, 0.5\theta^2+ s),\\
\gamma| \left(\mathbf{\Theta}, \mathbf{\Delta}\right) &\sim \text{Gamma}(r + 1, |\delta | + s).
\end{split}
\end{equation*}
Thus, we shrink $\theta$ and $\gamma$ on average by $\mathbb{E}\{\lambda| \left(\mathbf{\Theta}, \mathbf{\Delta}\right)\} = \frac{r + 0.5}{0.5\ \theta^2+ s}$ and $\mathbb{E}\{\gamma| \left(\mathbf{\Theta}, \mathbf{\Delta}\right)\} = \frac{r + 1}{|\delta | + s}$, respectively. That is, our approach selects the shrinkage parameters $\lambda$ and $\gamma$ based on the current values of $\theta$ and $\delta$ in a way that larger (smaller) entries are regularized more (less). The selection of the hyper-parameters $r$ and $s$ is also an important task and can significantly affect performance.
Based on numerical evidence from synthetic data, we set $r=10^{-2}$ and $s=10^{-6}$; a similar suggestion is also made in  \cite{wang2012bayesian}.

Finally, the construction of the Gibbs sampler proceeds as follows:  matrices $\{\mathbf{\Psi}^k\}_{k=1}^K$ are initialized as the identity matrix, while $\{\mathbf{\Psi}^r\}_{\{r: r \in \mathcal{P}\left(K\right), \&|r|>1\}}$ at zero. Then, in each iteration of the Markov Chain Monte Carlo chain, we update the vectors $\boldsymbol{\theta}_{ij}$ in (\ref{theta_ij}) and the diagonal entries $\psi^k_{ii}$, one at a time, using their full conditional posterior densities given in (\ref{posterior density of theta_ij}) and (\ref{posterior density of psi_ii}), respectively. Algorithm \ref{BlockGibbsBCONCORD} describes one iteration of the resulting Gibbs Sampler.
\begin{algorithm}[htp]
\caption{Gibbs Sampler for BJNS}\label{BlockGibbsBCONCORD}
\begin{algorithmic}[1]
\Procedure{BJNS}{}\Comment{Input $\mathcal{Y}, \mathbf{\Theta}, \mathbf{\Delta}$}
\For{$i=1, ..., p-1$}
\For{$j = i+1, ..., p$}
\For{$l = 1, ..., 2^K-1$}
\State $\lambda \gets \text{Gamma}(r + 0.5, 0.5(\theta_{l,ij}^2 + s))$ 
\State $\mu \gets - \frac{n\left( \boldsymbol{a}_{ij} + \mathbf{\Upsilon}_{(ij) (-(ij))}\mathbf{\Theta}_{-(ij)} \right)_{{l}}}{n \left[ \mathbf{\Upsilon}_{(ij) (ij)}\right]_{{l}{l}} + \lambda}$
\State $\nu^2 \gets \frac{1}{n \left[ \mathbf{\Upsilon}_{(ij) (ij)}\right]_{{l}{l}} + \lambda}$
\State $c_{l} \gets \sqrt{2\pi\nu^{2}} \exp\left\{\frac{\mu^{2}}{2\nu^{2}} \right\}$
\EndFor
\State $\boldsymbol{\theta}_{ij} \gets 0_{(2^K-1)\times1}$
\State ${l} \gets \text{sample}\left(1, \left\{0, 1, ..., 2^K-1\right\}, \text{probs} \propto \left\{1, c_{1,ij} ..., c_{2^K-1, ij}\right\} \right)$ 
\If{${l} \ne 0$}
\State $\boldsymbol{\theta}_{{l}, ij} \gets N\left(\mu, \nu^{2} \right)$
\EndIf
\EndFor
\For{$k=1, ..., K$}
\State $\gamma \gets \text{Gamma}(r + 1, |\psi^k_{ii} | + s)$
\State $b \gets \gamma+ n\sum_{j \ne i}{\left( \sum\limits_{r\in\vartheta_k} \psi^r_{ij} \right)s^k_{ij}}$
\State Update $\psi^k_{ii}$ using Algorithm 1 in Supplemental section S2 
\EndFor
\EndFor
\For{$k=1, ..., K$}
\State $\gamma \gets \text{Gamma}(r + 1, |\psi^k_{pp} | + s)$
\State $b \gets \gamma+ n\sum_{j \ne p}{\left( \sum\limits_{r\in\vartheta_k} \psi^r_{pj} \right)s^k_{pj}}$
\State Update $\psi^k_{pp}$ using Algorithm 1 in Supplemental section S2 
\EndFor
\State \textbf{return} $ \mathbf{\Theta}, \mathbf{\Delta}$\Comment{Return the set of updated parameters}
\EndProcedure
\end{algorithmic}
\end{algorithm}

\begin{Remark}
Note that although BJNS is a completely unsupervised approach, available prior knowledge on shared patterns across the groups can be easily incorporated by removing redundnant components $\mathbf{\Psi}^r$.
Further, prior information could be incorporated through appropriate specification of  the edge selection probabilities $q_1, q_2$.
\end{Remark}
\begin{Remark}
The Gibbs sampler described in Algorithm \ref{BlockGibbsBCONCORD} does not involve any matrix inversion which is critical for its computational efficiency.
\end{Remark}

\subsection{Procedure for Sparsity Selection}

\noindent
Note that the conditional posterior probability density of the off-diagonal elements of $\boldsymbol{\theta}_{ij}$ is a mixture 
density that puts all of its mass on the events $\left\{\boldsymbol{\theta}_{ij}: | \boldsymbol{\theta}_{ij} | \leq 1 \right\}$, where 
$| \boldsymbol{\theta}_{ij} |$ is the number of non-zero coordinates of $\boldsymbol{\theta}_{ij}$. This property of BJNS 
allows for model selection, in the sense that in every iteration of the Gibbs sampler one can check whether $
\boldsymbol{\theta}_{ij} = \boldsymbol{0}$ or which element of $\boldsymbol{\theta}_{ij}$ (there could be at most one non-
zero element) is non-zero. Finally, in the end of the procedure, we choose the event with the highest frequency during 
sampling. In addition, credible intervals can be constructed, by using the empirical quantiles of the values corresponding to 
the frequency distribution during sampling. 

\section{High Dimensional Sparsity Selection Consistency and Convergence Rates} \label{sec:consistency}

Next, we establish theoretical properties for BJNS.
Let $\left\{ \mathbf{\Psi}^{r,0}, r\in \bigcup_{k=1}^{K}\vartheta_k\right\}$ be the true values of the matrices in (\ref{eq21}), so that $\mathbf{\Omega}^{k,0} = \sum\limits_{r\in \vartheta_k}\mathbf{\Psi}^{r,0}$ corresponds to the true decomposition of each precision matrix $k=1, ..., K$.
Using a similar ordering as in (\ref{theta}), we define $\mathbf{\Theta}^0$ to be the vectorized version of the off-diagonal elements of the true matrices $\left\{ \mathbf{\Psi}^{r,0}, r\in \bigcup_{k=1}^{K}\vartheta_k\right\}$. 

The following assumptions are made to obtain the results.

\begin{Assumption}\label{Assumption1: Accurate Diagonal estimates}
(Accurate Diagonal estimates) There exist estimates $\left\{ \hat{\psi}^k_{ii}\right\}_{1 \leq i \leq p}^{1 \leq k \leq K}$, such that for any $\eta>0$, there exists a constant $C>0$, such that 
\begin{equation} \label{accurate-estimation}
\max_{\substack{1 \leq i \leq p \\ 1 \leq k \leq K}} \| \hat{\psi}^k_{ii} - \psi^k_{ii}\| \leq C \left( \sqrt{\frac{\log p}{n}}\right),
\end{equation}
with probability at least $1-O(n^{-\eta})$.
\end{Assumption}

\noindent
Note that our main objective is to accurately estimate the off-diagonal entries of all the matrices present in the model 
(\ref{eq21}). Hence, as commonly done for pseudo-Likelihood based high-dimensional consistency proofs -see 
\cite{khare2015convex,peng2009partial}- we assume the existence of accurate estimates for the diagonal elements through 
Assumption \ref{accurate-estimation}. One way to get the accurate estimates of the diagonal entries is discussed in 
Lemma 4 of \cite{khare2015convex}. Denote the resulting estimates of the vectors $\mathbf{\Delta}$ and $\boldsymbol{a}$  by $\hat{\mathbf{\Delta}}$ and  $\hat{\boldsymbol{a}}$, respectively. 
We now consider the accuracy of the estimates of the off-diagonal entries obtained after running the BJNS procedure with 
the diagonal entries fixed at $\hat{\mathbf{\Delta}}$. 

\begin{Assumption}\label{Assumption2: d_t} 
$d_{\boldsymbol{t}} \sqrt{\frac{\log p}{n}} \to 0, \quad \quad \text{as} \quad n \to \infty.
$
\end{Assumption}

\noindent
This assumption essentially states that the number of variables $p$ has to grow slower than 
$e^{(\frac{n}{d_{\boldsymbol{t}}^2})}$. Similar assumptions have been made in other high dimensional covariance estimation 
methods e.g. \cite{banerjee2014posterior}, \cite{banerjee2015bayesian}, \cite{bickel2008regularized}, and \cite{xiang2015high}.

\begin{Assumption}\label{Assumption4: Sub-Gaussianity} There exists $c>0$, independent of $n$ and $K$ such that 
\begin{equation*}
\mathbb{E}\left[ \exp \left( {\boldsymbol{\alpha}'\mathbf{y_{i:}^k}} \right) \right] \leq \exp \left(c \boldsymbol{\alpha}'\boldsymbol{\alpha} \right).
\end{equation*}
\end{Assumption}

\noindent
While building our method, we assumed that the data in each group comes from a multivariate normal distribution. The 
above assumption allows for deviations from Gaussianity. Hence, Theorem \ref{theorem (strong consistency)} below 
will show that the BJNS procedure is robust (in terms of consistency) under misspecification of the data generating distribution, as long as its tails 
are sub-Gaussian. 

\begin{Assumption}\label{Assumption3: Bounded eigenvalues}
(Bounded eigenvalues). There exists $\tilde{\varepsilon_{0}} > 0$, independent of $n$ and $K$, such that for all $1\leq k \leq K$,
\begin{equation*}
\tilde{\varepsilon_{0}} \leq \text{eig}_{\min}\left( {\mathbf{\Sigma}}^k\right) \leq \text{eig}_{\max}\left( {\mathbf{\Sigma}}^k\right) \leq \frac{1}{\tilde{\varepsilon_{0}}}.
\end{equation*}
\end{Assumption}

\noindent
This is a standard assumption in high dimensional analysis to obtain consistency results, see for example  
\cite{buhlmann2011statistics}.

Henceforth, we let $\varepsilon_{0} = \frac{c^2\tilde{\varepsilon_{0}}}{K}$, $a_1 = \frac{\varepsilon_0^3}{768 K}$, $a_2 = \frac{8c^2}{K\varepsilon_0}$, $a_3 = \frac{16c^2}{\lambda}$.

\begin{Assumption}\label{Assumption5: Signal Strength}
(Signal Strength). Let $s_n$ be the smallest non-zero entry (in magnitude) in the vector $\mathbf{\Theta}_0$. We assume 
$\frac{\frac{1}{2}\log n + a_2d_{\boldsymbol{t}}\log p}{a_1 n s_n^2} \to 0$.
\end{Assumption}

\noindent
This is again a standard assumption. Similar assumptions on the appropriate signal size can be found in 
\cite{khare2015convex,peng2009partial}.

\begin{Assumption}\label{Assumption6: Decay rate} (Decay rate of the edge probabilities).
Let $q_{1} = p^{-a_2 d_{\boldsymbol{t}}} $, $q_{2} = p^{-a_3 n}$, and $\tau_n = \frac{\varepsilon_0}{4c}\sqrt{\frac{n}{\log p}}$, for some constant $c$.
\end{Assumption}

This can be interpreted as a priori penalizing matrices with too many non-zero entries. We have faster rate $q_{2}$ for the 
case of super dense matrices. Similar assumptions are common in the literature, see for example 
\cite{narisetty2014bayesian} and \cite{cao2016posterior}.

We now establish the main posterior consistency result. In particular, we show that the
posterior mass assigned to the true model converges to one in probability (under the true model).

\begin{Theorem}\label{theorem (strong consistency)}
(Strong Selection Consistency) Based on the joint posterior distribution given in (\ref{posterior theta and delta}), and under Assumptions 1-6, the following holds,
\begin{equation}\label{main result}
\pi \left\{ \mathbf{\Theta}\in\mathcal{S}_{\boldsymbol{t}} | \hat{\mathbf{\Delta}},\mathcal{Y}\right\} \xrightarrow{\text{$\mathbb{P}_0$}} 1, \quad \quad \text{as} \quad n \to \infty.
\end{equation}
\end{Theorem}

\noindent
Our next result establishes estimation consistency of the BJNS procedure for $\mathbf{\Theta}$ and also provides 
a corresponding rate of convergence. 
\begin{Theorem}\label{theorem (Estimation Consistency and Convergence Rate)}
(Estimation Consistency Rate) Let $R_n$ be the maximum value (in magnitude) in the vector $\mathbf{\Theta}_0$ and 
assume that $R_n$ can not grow at a rate faster than $\sqrt{n \log p}$. Then, based on the joint posterior distribution given in (\ref{posterior theta 
and delta}), and under Assumptions 1-6, there exists a large enough constant $G$ (not depending on $n$), such that the 
following holds,
\begin{equation}
\mathbb{E}_{0} \left[P \left(\| \mathbf{\Theta} - \mathbf{\Theta}^0\|_2 > G \sqrt{\frac{d_{\boldsymbol{t}} \log p}{n}} \mid 
\hat{\mathbf{\Delta}}, \mathcal{Y}\right) \right] \quad \to 0 \quad \text{as} \quad n \to \infty.
\end{equation}
\end{Theorem}

\noindent
The proofs of the above results are provided in the Supplement, section S3.

\section{Simulation Studies}
In this section, we present three simulation studies to evaluate the performance of BJNS. In the first study, we illustrate the performance of BJNS in two scenarios with four precision matrices, each. In the second simulation, we compare the performance of BJNS with other methodologies, such as Glasso, where the Graphical lasso by \cite{friedman2008sparse} is applied to each graphical model separately, joint estimation by \cite{guo2011joint}, denoted by JEM-G, the Group Graphical Lasso denoted by GGL by \cite{danaher2014joint}, and the Joint Structural Estimation Method denoted by JSEM, by \cite{ma2016joint}. In the third simulation, we demonstrate the numerical scalability of BJNS, when the number of precision matrices $K$ is relatively large.

Throughout, for any $K$ precision matrices $\{ \mathbf{\Omega}^k\}_{k=1}^K$ of dimensions $p \times p$, we generate data as follows,
\begin{equation*}
\mathbf{y}_{i:}^k \sim \mathcal{N}_p\left( \mathbf{0}, \left( \mathbf{\Omega}^k \right)^{-1}\right), \quad \quad i=1,...,n_k, k=1, ..., K.
\end{equation*}

We assess the model performance using three well known accuracy measures, specificity (SP), sensitivity (SE) and Matthews Correlation Coefficients (MCC) defined as:
\begin{equation}\label{accuracy measures}
\begin{split}
\text{SP} &= \frac{\text{TN}}{\text{TN} + \text{FP}}, \quad \quad \quad \text{SE} = \frac{\text{TP}}{\text{TP} + \text{FN}}\\
\text{MCC} &= \frac{\text{TP} \times \text{TN} - \text{FP} \times \text{FN}}{\sqrt{(\text{TP} + \text{FP})(\text{TP} + \text{FN})(\text{TN} + \text{FP})(\text{TN} + \text{FN})}}
\end{split}
\end{equation}
where, TP, TN, FP and FN represent the number of true positives, true negatives, false positives
and false negatives, respectively. Larger values of any of the above metrics indicates a better sparsity selection produced by the underlying method.

Recall that BJNS estimates the vectors $\boldsymbol{\theta}_{ij}$ so that only a single coordinate is non-zero; i.e. regardless of the number of networks $K$, BJNS estimates at most $p(p-1)/2$ off-diagonal parameters and enforces at least $(2^{K-1}-1)p(p-1)$ of the remaining off-diagonal parameters to be zero. The final decision for the sparsity of each vector $\boldsymbol{\theta}_{ij}$ in (\ref{theta_ij}), is made based on majority voting. For every experiment, we base our inference on 2000 samples that are generated from the MCMC chain after removing 2000 burn-in samples. We ensure the robustness of the presented results, by repeating each experiment over 100 replicate data sets and taking the average of the accuracy measures across the replicates.

Finally, the implementation of BJNS using the Rcpp and Rcpp Armadillo libraries (\cite{RcppArmadillo}, \cite{Rcpp}) together with avoidance of matrix inversion operations as previously discussed contributes to its computational efficiency.

\subsection{Simulation 1: Four groups ($K=4$)} \label{Simulation 1: Four groups ($K=4$)}

We consider two challenging scenarios to examine the performance of BJNS for the case of simultaneously estimating four inverse covariance matrices ($K=4$).
\begin{enumerate}[label=(\roman*)]
\item Four precision matrices $\mathbf{\Omega}^1$, $\mathbf{\Omega}^2$, $\mathbf{\Omega}^3$, and $\mathbf{\Omega}^4$, with different degrees of shared structures. We let the first matrix $\mathbf{\Omega}^1$ to be an AR(2) model with $\omega^1_{jj} = 1$, for $j=1,...,p$; $\omega^1_{jj+1} = \omega^1_{j+1j} = 0.5$, for $j=1,...,p-1$; and $\omega^1_{jj+2} = \omega^1_{j+2j} = 0.25$, for $j=1,...,p-2$. To construct $\mathbf{\Omega}^2$ we randomly replace $\frac{p}{4}$ non zero edges from $\mathbf{\Omega}^1$ with zeros and replace $\frac{p}{4}$ zero edges, at random, with numbers generated from $[-0.6, -0.4] \cup [0.4, 0.6]$. We construct $\mathbf{\Omega}^3$ by randomly removing $\frac{p}{2}$ edges shared between $\mathbf{\Omega}^1$ and $\mathbf{\Omega}^2$ and then using $[-0.6, -0.4] \cup [0.4, 0.6]$, we randomly add $\frac{p}{2}$ other edges that are present in neither $\mathbf{\Omega}^1$ nor $\mathbf{\Omega}^2$. Finally, we construct $\mathbf{\Omega}^4$ by removing the remaining $2p-3 - \frac{3p}{4}$ edges that are common in $\mathbf{\Omega}^1$ and $\mathbf{\Omega}^2$ and randomly add $2p-3 - \frac{3p}{4}$ edges that are not present in any of the other graphs. The resulting matrix $\mathbf{\Omega}^4$ has nothing in common with the other precision matrices. Therefore, the true relation between the four networks is as follows,
\begin{equation}\label{sim1.1}
\begin{split}
\mathbf{\Omega}^1 &= \mathbf{\Psi}^1 + \mathbf{\Psi}^{12} + \mathbf{\Psi}^{123} \quad \quad \mathbf{\Omega}^3 = \mathbf{\Psi}^3 + \mathbf{\Psi}^{123} \\
\mathbf{\Omega}^2 &= \mathbf{\Psi}^2 + \mathbf{\Psi}^{12} + \mathbf{\Psi}^{123} \quad \quad \mathbf{\Omega}^4 = \mathbf{\Psi}^4, 
\end{split}
\end{equation}
where, $\mathbf{\Psi}^1$, $\mathbf{\Psi}^2$, $\mathbf{\Psi}^3$, and $\mathbf{\Psi}^4$ account for the edges that are unique to their corresponding groups and $\mathbf{\Psi}^{12}$ and $\mathbf{\Psi}^{123}$ contains the edges that are common between  the four groups.
A heat map plot of the true precision matrices with $p=50$ is given figure \ref{4AR_truth}.
\begin{figure}[h]
\centering
\caption{Heat maps of the precision matrices of the four groups}
\label{4AR_truth}
\includegraphics[height=1.1in,width=4in,angle=0]{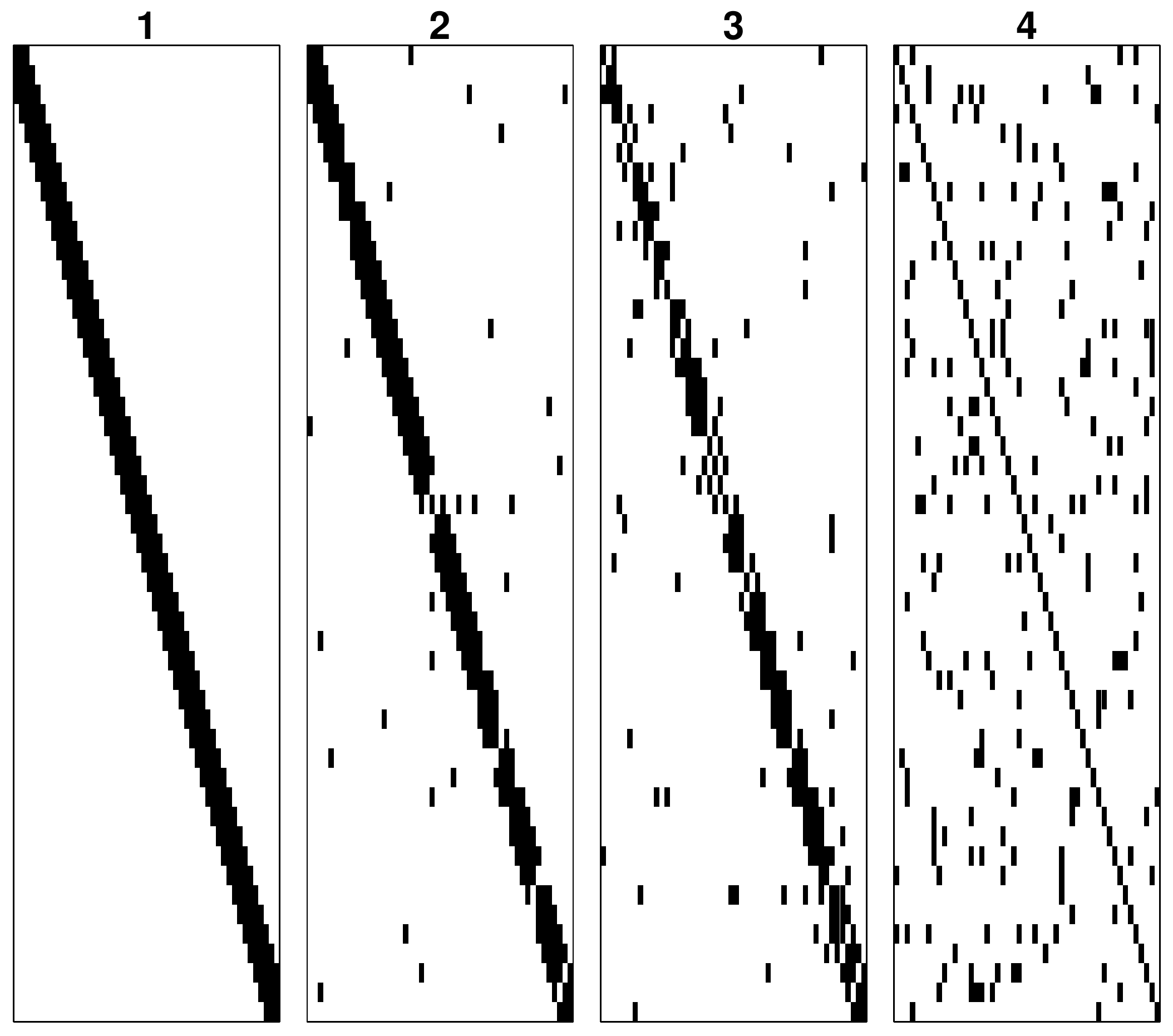} 
\end{figure}
\item In most real data applications, the underlying networks tend to follow sparsity patterns that are completely random and do not necessarily follow a certain structure. For this reason, we would like to examine the performance of BJNS under unstructured settings. We will specifically consider estimating four precision matrices with complete random sparsity patterns with signals generated from $[-0.6, -0.4] \cup [0.4, 0.6]$. We take the level of sparsity in each matrix to be 95\%, half of which is unique to the matrices and the other half is shared between all four of them. The true relationship between the networks is as follows,
\begin{equation}\label{sim1.2}
\begin{split}
\mathbf{\Omega}^k &= \mathbf{\Psi}^k + \mathbf{\Psi}^{1234}, \quad k = 1, 2, 3, 4.
\end{split}
\end{equation}
where, $\mathbf{\Psi}^1$, $\mathbf{\Psi}^2$, $\mathbf{\Psi}^3$, and $\mathbf{\Psi}^4$ account for the edges that are unique to their corresponding groups and $\mathbf{\Psi}^{1234}$ contains the edges that are common between all of the four groups. 
\end{enumerate}

For each of the above settings, we run the following combinations of $p$ and $n$:
\begin{itemize}
\item $p=200$, and $n=50, 100, 150, 200, 250, \text{and}, 300$,
\item $p=500$, and $n=300, 400, \text{and}, 500$.
\end{itemize}
The estimations are based on the full decomposition model, which for the case of $K=4$ is given as,
\begin{equation}\label{full model for k=4}
\begin{split}
\mathbf{\Omega}^1 = \mathbf{\Psi}^1 + \mathbf{\Psi}^{12} + \mathbf{\Psi}^{13} + \mathbf{\Psi}^{14} + \mathbf{\Psi}^{123} + \mathbf{\Psi}^{124} + \mathbf{\Psi}^{134} + \mathbf{\Psi}^{1234},\\
\mathbf{\Omega}^2 = \mathbf{\Psi}^2 + \mathbf{\Psi}^{12} + \mathbf{\Psi}^{23} + \mathbf{\Psi}^{24} + \mathbf{\Psi}^{123} + \mathbf{\Psi}^{124} + \mathbf{\Psi}^{234} + \mathbf{\Psi}^{1234},\\
\mathbf{\Omega}^3 = \mathbf{\Psi}^3 + \mathbf{\Psi}^{13} + \mathbf{\Psi}^{23} + \mathbf{\Psi}^{34} + \mathbf{\Psi}^{123} + \mathbf{\Psi}^{134} + \mathbf{\Psi}^{234} + \mathbf{\Psi}^{1234},\\
\mathbf{\Omega}^4 = \mathbf{\Psi}^4 + \mathbf{\Psi}^{14} + \mathbf{\Psi}^{24} + \mathbf{\Psi}^{34} + \mathbf{\Psi}^{124} + \mathbf{\Psi}^{134} + \mathbf{\Psi}^{234} + \mathbf{\Psi}^{1234}.
\end{split}
\end{equation}
Tables \ref{table_sim1.1} and \ref{table_sim1.2} summarize the average of the accuracy measures for both settings, across 100 repetitions. As can be seen, the values of the accuracy measures tend to be much higher for the joint effects (namely $\mathbf{\Psi}^{12}$ and $\mathbf{\Psi}^{123}$ in table \ref{table_sim1.1} and $\mathbf{\Psi}^{1234}$ in table \ref{table_sim1.2}), which implies that BJNS is borrowing strength across the distinct samples, to provide more robust estimates of the joint edges. In addition, high values of specificity, regardless of the sample size, shows a surprisingly low false positive rate. 

Another key strength of BJNS is its ability to select the correct underlying model from the above full representation. As described in the Gibbs sampler in section 4.1, model selection takes place at the level of $\boldsymbol{\theta}_{ij}$s (for $1 \leq i < j \leq p$) in (\ref{theta_ij}) which in the case of $K=4$, are vectors of length $2^K-1=15$,  
\begin{equation}\label{sim1.3}
\begin{split}
\boldsymbol{\theta}_{ij} = \left( \psi^1_{ij}, \psi^2_{ij} , \psi^3_{ij}, \psi^4_{ij}, \psi^{12}_{ij}, \psi^{13}_{ij}, \psi^{14}_{ij}, \psi^{23}_{ij} ,\psi^{24}_{ij}, \psi^{123}_{ij}, \psi^{124}_{ij}, \psi^{134}_{ij}, \psi^{234}_{ij}, \psi^{1234}_{ij}\right)'.
\end{split}
\end{equation}
for every $1 \leq i < j \leq p$, BJNS aims at detecting whether $\boldsymbol{\theta}_{ij}$ is a zero vector or which of it's components is non-zero (recall that $\boldsymbol{\theta}_{ij}$ has at most one non-zero element). For example, in the case of the first simulation, if an edge $(ij)$ is common across all four networks, then ${\psi}^{1234}_{ij}$ must be the non-zero element in the corresponding ${\boldsymbol{\theta}}_{ij}$. That is, we must have that
\begin{equation*}
\boldsymbol{\theta}_{ij} = {\left(\begin{array}{*{20}{c}}
 \boldsymbol{0}_{14} \\ \psi^{1234}_{ij}
\end{array}
\right)}.
\end{equation*}
Thus, based on the way we estimate $\boldsymbol{\theta}_{ij}$s, high values of the accuracy measures for the matrices that are present in the true models (\ref{sim1.1}) and (\ref{sim1.2}), automatically indicates that the model is performing well in selecting the true sparsity patterns of all $\boldsymbol{\theta}_{ij}$s. In view of the full model (\ref{full model for k=4}), high values of the accuracy measures for the inverse covariance matrices $\{\mathbf{\Omega}^k\}_{k=1}^4$ is another indication of strong selection capability of BJNS. 

\begin{sidewaystable}[htp]
\centering 
\caption{Summary results for simulation i}
\begin{tabular}{
l
S[table-format = 3]
S[table-format = 3]
S[table-format = 3]
S[table-format = 3]
S[table-format = 3]
S[table-format = 3]
S[table-format = 3]
S[table-format = 3]
S[table-format = 3]
S[table-format = 3]
S[table-format = 3]
}
        \toprule
        \multicolumn{2}{c}{} & \multicolumn{6}{c}{$p=200$} & 
        \multicolumn{3}{c}{$p=500$}  \\
        \cmidrule(lr){3-8}  
        \cmidrule(lr){9-11} 
        && {$n=50$}&{$n=100$}&{$n=150$}&{$n=200$}&{$n=250$}&{$n=300$}&{$n=300$}
        &{$n=400$}&{$n=500$}\\
         \cmidrule(lr){1-2}  
        \cmidrule(lr){3-8}  
        \cmidrule(lr){9-11} 
{\multirow{3}{*}{$\mathbf{\Omega}_1$} } 
&{MC\%}&  {60 (0.020)} & {79 (0.016)} & {86 (0.012)} & {88 (0.011)} & {89 (0.010)} & {89 (0.009)} & {95 (0.004)} & {96 (0.005)}& {96 (0.004)} \\
& {SP\%}& {99 (0.001)} & {99 (0.001)} & {99 (0.001)} & {99 (0.001)} & {99 (0.001)} & {99 (0.001)} & {100 (0.000)} & {100 (0.000)}& {100 (0.000)} \\
& {SE\%}& {67 (0.021)} & {90 (0.014)} & {97 (0.009)} & {99 (0.004)} & {100 (0.002)}& {100 (0.001)} & {100 (0.001)} & {100 (0.000)}& {100 (0.000)} \\
\cmidrule(lr){1-2}  
\cmidrule(lr){3-8}  
\cmidrule(lr){9-11}
{\multirow{3}{*}{$\mathbf{\Omega}_2$} } 
&{MC\%}&  {51 (0.021)} & {69 (0.016)} & {78 (0.013)} & {81 (0.014)} & {80 (0.011)} & {82 (0.012)} & {86 (0.007)} & {87 (0.006)}& {88 (0.006)} \\
& {SP\%}& {99 (0.001)} & {99 (0.001)} & {99 (0.001)} & {99 (0.001)} & {99 (0.001)} & {99 (0.001)} & {100 (0.000)} & {100 (0.000)}& {100 (0.000)} \\
& {SE\%}& {59 (0.024)} & {82 (0.018)} & {92 (0.014)} & {95 (0.011)} & {95 (0.010)} & {97 (0.010)} & {97 (0.005)} & {98 (0.004)}& {99 (0.003)} \\
\cmidrule(lr){1-2}  
\cmidrule(lr){3-8}  
\cmidrule(lr){9-11}
{\multirow{3}{*}{$\mathbf{\Omega}_3$} } 
&{MC\%}&  {43 (0.021)} & {64 (0.019)} & {73 (0.015)} & {79 (0.015)} & {80 (0.014)} & {80 (0.014)} & {84 (0.007)} & {86 (0.007)}& {86 (0.007)} \\
& {SP\%}& {98 (0.001)} & {99 (0.001)} & {99 (0.001)} & {99 (0.001)} & {99 (0.001)} & {99 (0.001)} & {100 (0.000)} & {100 (0.000)}& {100 (0.000)} \\
& {SE\%}& {51 (0.024)} & {78 (0.022)} & {89 (0.016)} & {95 (0.015)} & {96 (0.011)} & {96 (0.010)} & {97 (0.006)} & {98 (0.005)}& {99 (0.004)} \\
\cmidrule(lr){1-2}  
\cmidrule(lr){3-8}  
\cmidrule(lr){9-11}
{\multirow{3}{*}{$\mathbf{\Omega}_4$} } 
&{MC\%}&  {19 (0.020)} & {47 (0.020)} & {54 (0.020)} & {72 (0.015)} & {75 (0.016)} & {81 (0.012)} & {80 (0.009)} & {85 (0.006)}& {88 (0.007)} \\
& {SP\%}& {98 (0.001)} & {99 (0.001)} & {99 (0.001)} & {99 (0.001)} & {99 (0.001)} & {99 (0.001)} & {100 (0.000)} & {100 (0.000)}& {100 (0.000)} \\
& {SE\%}& {21 (0.020)} & {53 (0.024)} & {61 (0.023)} & {84 (0.016)} & {87 (0.014)} & {95 (0.010)} & {89 (0.009)} & {96 (0.005)}& {99 (0.003)} \\
\cmidrule(lr){1-2}  
\cmidrule(lr){3-8}  
\cmidrule(lr){9-11}
{\multirow{3}{*}{$\mathbf{\Psi}_1$} } 
&{MC\%}&  {30 (0.048)} & {52 (0.047)} & {66 (0.051)} & {75 (0.043)} & {76 (0.040)} & {81 (0.037} & {87 (0.020)} & {91 (0.017)}& {93 (0.015)} \\
& {SP\%}& {100 (0.000)} & {100 (0.000)} & {100 (0.000)} & {100 (0.000)} & {100 (0.000)} & {100 (0.000)} & {100 (0.000)} & {100 (0.000)}& {100 (0.000)} \\
& {SE\%}& {33 (0.056)} & {57 (0.058)} & {71 (0.067)} & {82 (0.047)} & {85 (0.044)} & {89 (0.040)} & {86 (0.028)} & {91 (0.021)}& {94 (0.020)} \\
\cmidrule(lr){1-2}  
\cmidrule(lr){3-8}  
\cmidrule(lr){9-11}
{\multirow{3}{*}{$\mathbf{\Psi}_2$} } 
&{MC\%}&  {31 (0.051)} & {56 (0.043)} & {71 (0.035)} & {75 (0.037)} & {74 (0.033)} & {76 (0.035)} & {78 (0.021)} & {79 (0.018)}& {79 (0.019)} \\
& {SP\%}& {100 (0.000)} & {100 (0.000)} & {100 (0.000)} & {100 (0.000)} & {100 (0.000)} & {100 (0.000)} & {100 (0.000)} & {100 (0.000)}& {100 (0.000)} \\
& {SE\%}& {37 (0.067)} & {69 (0.067)} & {87 (0.041)} & {92 (0.040)} & {94 (0.038)} & {97 (0.024)} & {95 (0.020)} & {98 (0.013)}& {99 (0.009)} \\
\cmidrule(lr){1-2}  
\cmidrule(lr){3-8}  
\cmidrule(lr){9-11}
{\multirow{3}{*}{$\mathbf{\Psi}_3$} } 
&{MC\%}&  {25 (0.041)} & {58 (0.040)} & {73 (0.033)} & {82 (0.025)} & {84 (0.023)} & {84 (0.022)} & {85 (0.013)} & {87 (0.013)}& {88 (0.012)} \\
& {SP\%}& {100 (0.000)} & {100 (0.000)} & {100 (0.000)} & {100 (0.000)} & {100 (0.000)} & {100 (0.000)} & {100 (0.000)} & {100 (0.000)}& {100 (0.000)} \\
& {SE\%}& {23 (0.038)} & {60 (0.046)} & {79 (0.039)} & {92 (0.028)} & {94 (0.024)} & {96 (0.022)} & {94 (0.013)} & {98 (0.009)}& {99 (0.007)} \\
\cmidrule(lr){1-2}  
\cmidrule(lr){3-8}  
\cmidrule(lr){9-11}
{\multirow{3}{*}{$\mathbf{\Psi}_4$} } 
&{MC\%}&  {24 (0.027)} & {56 (0.023)} & {64 (0.021)} & {81 (0.017)} & {84 (0.013)} & {90 (0.011)} & {87 (0.008)} & {93 (0.006)}& {95 (0.005)} \\
& {SP\%}& {100 (0.000)} & {100 (0.000)} & {100 (0.000)} & {100 (0.000)} & {100 (0.000)} & {100 (0.000)} & {100 (0.000)} & {100 (0.000)}& {100 (0.000)} \\
& {SE\%}& {14 (0.017)} & {42 (0.024)} & {52 (0.022)} & {76 (0.020)} & {80 (0.017)} & {90 (0.015)} & {84 (0.010)} & {93 (0.008)}& {97 (0.005)} \\
\cmidrule(lr){1-2}  
\cmidrule(lr){3-8}  
\cmidrule(lr){9-11}
{\multirow{3}{*}{$\mathbf{\Psi}_{12}$} } 
&{MC\%}&  {36 (0.041)} & {57 (0.034)} & {70 (0.021)} & {78 (0.029)} & {80 (0.030)} & {83 (0.026)} & {86 (0.015)} & {90 (0.014)}& {93 (0.012)} \\
& {SP\%}& {100 (0.000)} & {100 (0.000)} & {100 (0.000)} & {100 (0.000)} & {100 (0.000)} & {100 (0.000)} & {100 (0.000)} & {100 (0.000)}& {100 (0.000)} \\
& {SE\%}& {34 (0.043)} & {59 (0.040)} & {74 (0.038)} & {82 (0.036)} & {83 (0.039)} & {88 (0.029)} & {87 (0.019)} & {92 (0.014)}& {95 (0.014)} \\
\cmidrule(lr){1-2}  
\cmidrule(lr){3-8}  
\cmidrule(lr){9-11}
{\multirow{3}{*}{$\mathbf{\Psi}_{123}$} } 
&{MC\%}&  {50 (0.029)} & {71 (0.022)} & {80 (0.020)} & {88 (0.018)} & {89 (0.014)} & {90 (0.013)} & {92 (0.008)} & {95 (0.006)}& {96 (0.005)} \\
& {SP\%}& {100 (0.000)} & {100 (0.000)} & {100 (0.000)} & {100 (0.000)} & {100 (0.000)} & {100 (0.000)} & {100 (0.000)} & {100 (0.000)}& {100 (0.000)} \\
& {SE\%}& {34 (0.028)} & {60 (0.028)} & {72 (0.029)} & {83 (0.027)} & {85 (0.020)} & {88 (0.022)} & {88 (0.013)} & {93 (0.009)}& {95 (0.009)} \\
\bottomrule
\end{tabular}
\label{table_sim1.1}
\end{sidewaystable}

\begin{sidewaystable}[htp]
\centering 
\caption{Summary results for simulation ii}
\begin{tabular}{
l
S[table-format = 3]
S[table-format = 3]
S[table-format = 3]
S[table-format = 3]
S[table-format = 3]
S[table-format = 3]
S[table-format = 3]
S[table-format = 3]
S[table-format = 3]
S[table-format = 3]
S[table-format = 3]
}
        \toprule
        \multicolumn{2}{c}{} & \multicolumn{6}{c}{$p=200$} & 
        \multicolumn{3}{c}{$p=500$}  \\
        \cmidrule(lr){3-8}  
        \cmidrule(lr){9-11} 
        && {$n=50$}&{$n=100$}&{$n=150$}&{$n=200$}&{$n=250$}&{$n=300$}&{$n=300$}
        &{$n=400$}&{$n=500$}\\
         \cmidrule(lr){1-2}  
        \cmidrule(lr){3-8}  
        \cmidrule(lr){9-11} 
{\multirow{3}{*}{$\mathbf{\Omega}^1$} } 
& {MC\%}&{17 (0.009)}&{36 (0.012)}&{46 (0.012)}&{56 (0.011)}&{60 (0.011)}&{68 (0.009)} &{61 (0.004)}&{67 (0.004)}&{72 (0.004)} \\
& {SP\%}&{88 (0.004)}&{94 (0.003)}&{95 (0.002)}&{95 (0.002)}&{96 (0.002)}&{96 (0.002} &{99 (0.000)}&{99 (0.000)}&{99 (0.000)} \\
& {SE\%}&{39 (0.015)}&{53 (0.013)}&{63 (0.013)}&{75 (0.013)}&{81 (0.011)}&{90 (0.009)} &{47 (0.004)}&{56 (0.004)}&{63 (0.004)} \\
\cmidrule(lr){1-2}  
\cmidrule(lr){3-8}  
\cmidrule(lr){9-11}
{\multirow{3}{*}{$\mathbf{\Omega}^2$} } 
& {MC\%}&{18 (0.011)}&{36 (0.011)}&{49 (0.013)}&{55 (0.010)}&{59 (0.013)}&{66 (0.011)} &{59 (0.004)}&{65 (0.004)}&{70 (0.004)} \\
& {SP\%}&{88 (0.004)}&{94 (0.003)}&{95 (0.002)}&{95 (0.002)}&{96 (0.002)}&{96 (0.002} &{99 (0.000)}&{99 (0.000)}&{99 (0.000)} \\
& {SE\%}&{39 (0.016)}&{53 (0.013)}&{67 (0.013)}&{74 (0.011)}&{78 (0.014)}&{85 (0.011)} &{46 (0.004)}&{54 (0.004)}&{61 (0.004)} \\
\cmidrule(lr){1-2}  
\cmidrule(lr){3-8}  
\cmidrule(lr){9-11}
{\multirow{3}{*}{$\mathbf{\Omega}^3$} } 
& {MC\%}&{18 (0.009)}&{39 (0.011)}&{48 (0.012)}&{56 (0.011)}&{63 (0.010)}&{64 ()0.010} &{61 (0.005)}&{67 (0.005)}&{72 (0.004)} \\
& {SP\%}&{88 (0.004)}&{94 (0.003)}&{95 (0.002)}&{95 (0.002)}&{96 (0.002)}&{96 (0.002} &{99 (0.000)}&{99 (0.000)}&{99 (0.000)} \\
& {SE\%}&{39 (0.014)}&{57 (0.013)}&{66 (0.014)}&{75 (0.013)}&{84 (0.011)}&{85 (0.010)} &{48 (0.005)}&{56 (0.004)}&{63 (0.004)} \\
\cmidrule(lr){1-2}  
\cmidrule(lr){3-8}  
\cmidrule(lr){9-11}
{\multirow{3}{*}{$\mathbf{\Omega}^4$} } 
& {MC\%}&{17 (0.009)}&{36 (0.011)}&{50 (0.013)}&{58 (0.010)}&{61 (0.012)}&{67 (0.009)} &{61 (0.005)}&{68 (0.008)}&{72 (0.004)} \\
& {SP\%}&{88 (0.004)}&{94 (0.002)}&{95 (0.002)}&{95 (0.002)}&{96 (0.002)}&{96 (0.001} &{99 (0.000)}&{99 (0.000)}&{99 (0.000)} \\
& {SE\%}&{38 (0.013)}&{53 (0.014)}&{69 (0.013)}&{78 (0.010)}&{82 (0.012)}&{89 (0.010)} &{47 (0.005)}&{56 (0.004)}&{63 (0.004)} \\
\cmidrule(lr){1-2}  
\cmidrule(lr){3-8}  
\cmidrule(lr){9-11}
{\multirow{3}{*}{$\mathbf{\Psi}^1$} } 
& {MC\%}&{7 (0.013)}&{20 (0.020)}&{32 (0.021)}&{47 (0.023)}&{55 (0.017)}&{67 (0.015)} &{30 (0.008)}&{40 (0.008)}&{48 (0.008)} \\
& {SP\%}&{98 (0.001)}&{99 (0.001)}&{99 (0.001)}&{99 (0.001)}&{99 (0.001)}&{99 (0.001)} &{100 (0.000)}&{100 (0.000)}&{100 (0.000)} \\
& {SE\%}&{9 (0.012)}&{17 (0.018)}&{26 (0.019)}&{41 (0.022)}&{49 (0.019)}&{64 (0.019)} &{11 (0.005)}&{18 (0.006)}&{26 (0.007)} \\
\cmidrule(lr){1-2}  
\cmidrule(lr){3-8}  
\cmidrule(lr){9-11}
{\multirow{3}{*}{$\mathbf{\Psi}^2$} } 
& {MC\%}&{7 (0.012)}&{20 (0.018)}&{37 (0.021)}&{46 (0.019)}&{52 (0.019)}&{65 (0.020)} &{29 (0.008)}&{38 (0.008)}&{46 (0.008)} \\
& {SP\%}&{98 (0.001)}&{99 (0.001)}&{99 (0.001)}&{99 (0.001)}&{99 (0.001)}&{99 (0.001)} &{100 (0.000)}&{100 (0.000)}&{100 (0.000)} \\
& {SE\%}&{9 (0.011)}&{17 (0.015)}&{31 (0.019)}&{40 (0.020)}&{46 (0.020)}&{62 (0.021)} &{11 (0.004)}&{17 (0.006)}&{24 (0.007)} \\
\cmidrule(lr){1-2}  
\cmidrule(lr){3-8}  
\cmidrule(lr){9-11}
{\multirow{3}{*}{$\mathbf{\Psi}^3$} } 
& {MC\%}&{8 (0.013)}&{24 (0.021)}&{37 (0.021)}&{47 (0.021)}&{58 (0.019)}&{62 (0.019)} &{30 (0.008)}&{40 (0.009)}&{49 (0.006)} \\
& {SP\%}&{98 (0.001)}&{99 (0.001)}&{99 (0.001)}&{99 (0.001)}&{99 (0.001)}&{99 (0.001)} &{100 (0.000)}&{100 (0.000)}&{100 (0.000)} \\
& {SE\%}&{9 (0.013)}&{20 (0.019)}&{29 (0.021)}&{40 (0.020)}&{53 (0.023)}&{59 (0.022)} &{11 (0.005)}&{18 (0.006)}&{26 (0.006)} \\
\cmidrule(lr){1-2}  
\cmidrule(lr){3-8}  
\cmidrule(lr){9-11}
{\multirow{3}{*}{$\mathbf{\Psi}^4$} } 
& {MC\%}&{7 (0.014)}&{21 (0.019)}&{38 (0.024)}&{50 (0.022)}&{56 (0.022)}&{66 (0.017)} &{30 (0.009)}&{40 (0.008)}&{49 (0.008)} \\
& {SP\%}&{98 (0.001)}&{99 (0.001)}&{99 (0.001)}&{99 (0.001)}&{99 (0.001)}&{99 (0.001)} &{100 (0.000)}&{100 (0.000)}&{100 (0.000)} \\
& {SE\%}&{9 (0.013)}&{18 (0.017)}&{32 (0.020)}&{43 (0.023)}&{50 (0.022)}&{63 (0.021)} &{11 (0.005)}&{18 (0.006)}&{26 (0.007)} \\
\cmidrule(lr){1-2}  
\cmidrule(lr){3-8}  
\cmidrule(lr){9-11}
{\multirow{3}{*}{$\mathbf{\Psi}^{1234}$} } 
& {MC\%}&{21 (0.020)}&{48 (0.019)}&{63 (0.022)}&{73 (0.015)}&{76 (0.014)}&{85 (0.011)} &{68 (0.007)}&{76 (0.006)}&{82 (0.005)} \\
& {SP\%}&{99 (0.001)}&{100 (0.000)}&{100 (0.000)}&{100 (0.000)}&{100 (0.000)}&{100 (0.000)} &{100 (0.000)}&{100 (0.000)}&{100 (0.000)} \\
& {SE\%}&{16 (0.015)}&{36 (0.019)}&{52 (0.026)}&{66 (0.019)}&{73 (0.019)}&{82 (0.016)} &{49 (0.009)}&{61 (0.009)}&{69 (0.008)} \\
\bottomrule
\end{tabular}
\label{table_sim1.2}
\end{sidewaystable}

Further, note that the model selection properties of BJNS to check convergence of the MCMC algorithm. This is accomplished by tracking the ratio of correctly selected $\boldsymbol{\theta}_{ij}$s, denoted by $\kappa$ and defines as 
\begin{equation*}
\kappa= \frac{ \# \boldsymbol{\theta}_{ij}\text{s} \ \text{that are correctly selected}}{\left( \frac{p(p-1)}{2}\right)}
\end{equation*}
In addition to accuracy assessment, $\kappa$ helps studying the number of iterations that on average it takes for the Gibbs sampler to converge. Figure \ref{MCMC and Selection plot} describes the trace plot and the histogram of $\kappa$ for 4000 iterations for the above two simulations for the $p=200$, $n = 300$ settings. The MCMC trace plots of $\kappa$ show that the Gibbs  sampler converges fairly quickly. Moreover, the histograms of $\kappa$ indicate a high proportion of correctly selected $\boldsymbol{\theta}_{ij}$s.
\begin{figure}[htp] 
\centering
\caption{Convergence and the distribution of $\kappa$ during the MCMC samplings}
\includegraphics[width=0.4\textwidth]{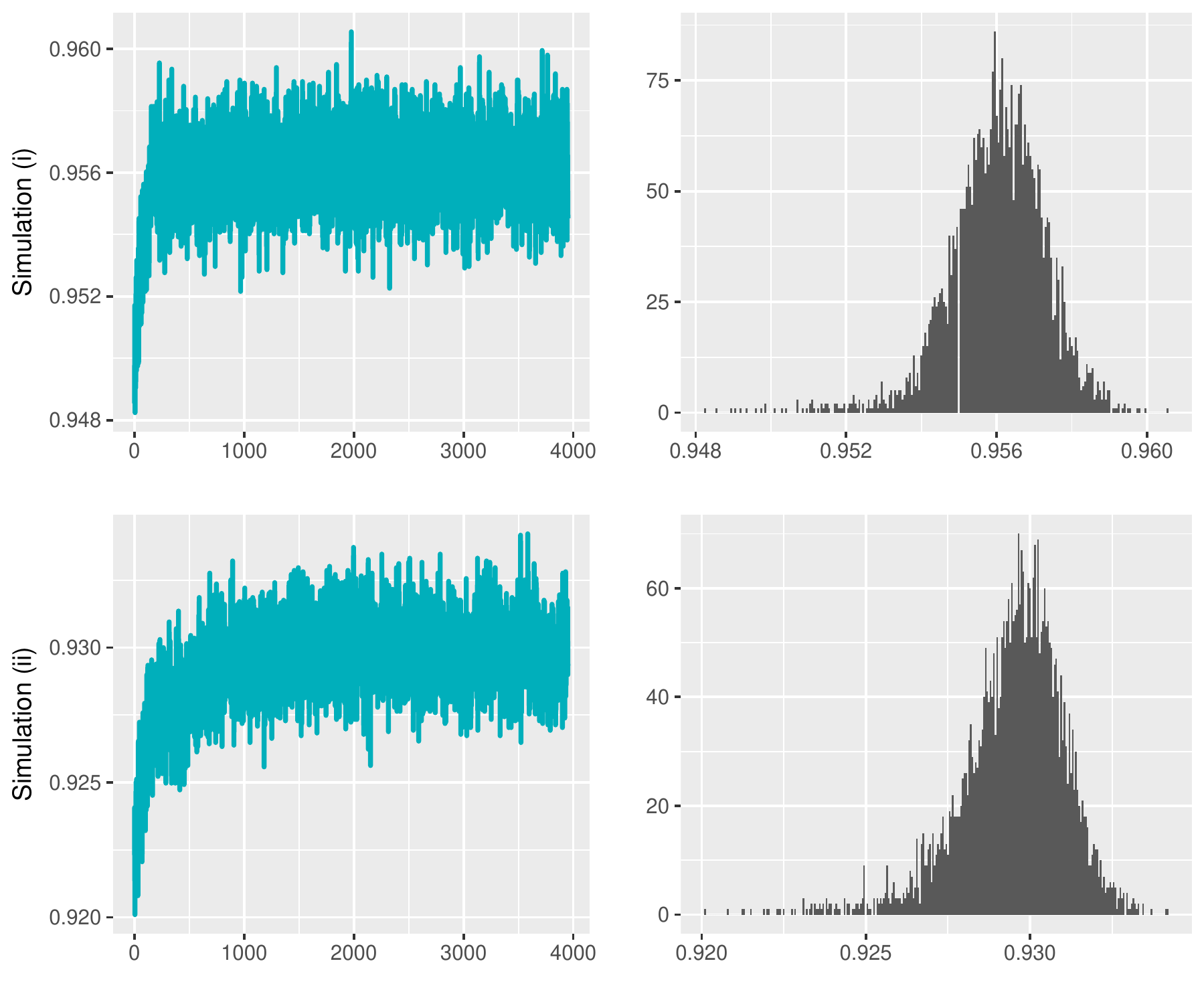} 
\label{MCMC and Selection plot}
\end{figure}

\subsection{Comparison With Existing Methods}

In this section, we compare the performance of BJNS with Glasso, JEM-G, GGL and JSEM, in two settings involving 6 networks. We also discuss a pure computational strategy to scale up BJNS to handle efficiently an even larger number of networks. 

\subsubsection{Comparison With Existing Methods}
We consider a scenario with $K = 6$ graphs each with $p = 200$ variables (see Figure \ref{Comparison 2}), where we first generate the adjacency matrices corresponding to three distinct $p$-dimensional networks, so that the adjacency matrices in each column of the plot in Figure \ref{Comparison 2} are common. Next, we replace the connectivity structure of the bottom right diagonal block of size $p/2$ by $p/2$ in each adjacency matrix with that of an other two distinct $p/2$-dimensional networks, so that graphical models in each column exhibit the same connectivity pattern except in the bottom right diagonal block of their adjacency matrices. The resulting true matrix decompositions of the networks is as follows:
\begin{equation}\label{true decomposition k=6}
\begin{split}
\mathbf{\Omega}^{1} &= \mathbf{\Psi}^{12} +  \mathbf{\Psi}^{135}, \quad \quad \mathbf{\Omega}^{3} = \mathbf{\Psi}^{34} +  \mathbf{\Psi}^{135}, \quad \quad
\mathbf{\Omega}^{5} = \mathbf{\Psi}^{56} +  \mathbf{\Psi}^{135},\\
\mathbf{\Omega}^{2} &= \mathbf{\Psi}^{12} +  \mathbf{\Psi}^{246}, \quad \quad 
 \mathbf{\Omega}^{4} = \mathbf{\Psi}^{34} +  \mathbf{\Psi}^{246}, \quad \quad
\mathbf{\Omega}^{6} = \mathbf{\Psi}^{56} +  \mathbf{\Psi}^{246}.
\end{split}
\end{equation}

Note that by replacing the connectivity structure among the second half of the nodes, the relationships between the first half and the second half of the nodes are also altered. In summary, these sparsity patterns illustrate the common components across different subsets of edges, as well as differences.
 The sparsity level for all networks is set to 92\%, while the proportion of common zeros (no edge present) across all six networks is about 60\%. Given the adjacency matrices, we then construct the inverse covariance matrices with the nonzero off-diagonal entries in each $\mathbf{\Omega}^k$ being uniformly generated from the $[-0.6, -0.4]\cup[0.4, 0.6]$ interval. To implement JESM, we supply the sparsity patterns defined according to the pattern in \ref{Comparison 2}. We also study the effect of mis-specification to JSEM in particular, by adding an additional $\rho = 4\%$ of edges to the networks. Adding $\rho = 4\%$ of edges results in having $60$\% of the information in the sparsity patterns being correct for JSEM. 

At each level of pattern mis-specification, we generated $n_k = 200, 300$ independent samples for each $k = 1,...,K$ and examined the finite sample performance of different methods in identifying the true graphs and estimating the precision matrices at the optimal choice of tuning parameters. Table \ref{comparison2:random sparsity} show the deviance measures between the estimated and the true precision matrices based on 50 replications for varying levels of sample size and pattern mis-specification. 

When $\rho =0$ (no mis-specifications), JSEM, which benefits from knowing 100\% of the information about which sets of edges can be fused across which subsets of networks, achieves a good balance between false positives and false negatives and yields the highest MCC score, as expected. BJNS is also very competitive and its overall performance is significantly better than all other unsupervised methods. Glasso and GGL tend to perform well in controlling false negatives and JEM-G is constantly among the best two models in terms of the SP score. With $\rho =0.04$ and sparsity level of 12\%, naturally the overall performance of all methods decrease. However, JSEM suffers more as it only has correctly specified information about shared patters for only 60\%  of the edges. On the other hand, BJNS proves robust and achieves a good balance between specificity and sensitivity, thus outperforming competing methods. In this setting, Glasso is also competitive due to the overall heterogeneity of the 6 networks.

\begin{figure}[h] 
\centering
\includegraphics[height=2.1in,width=4in,angle=0]{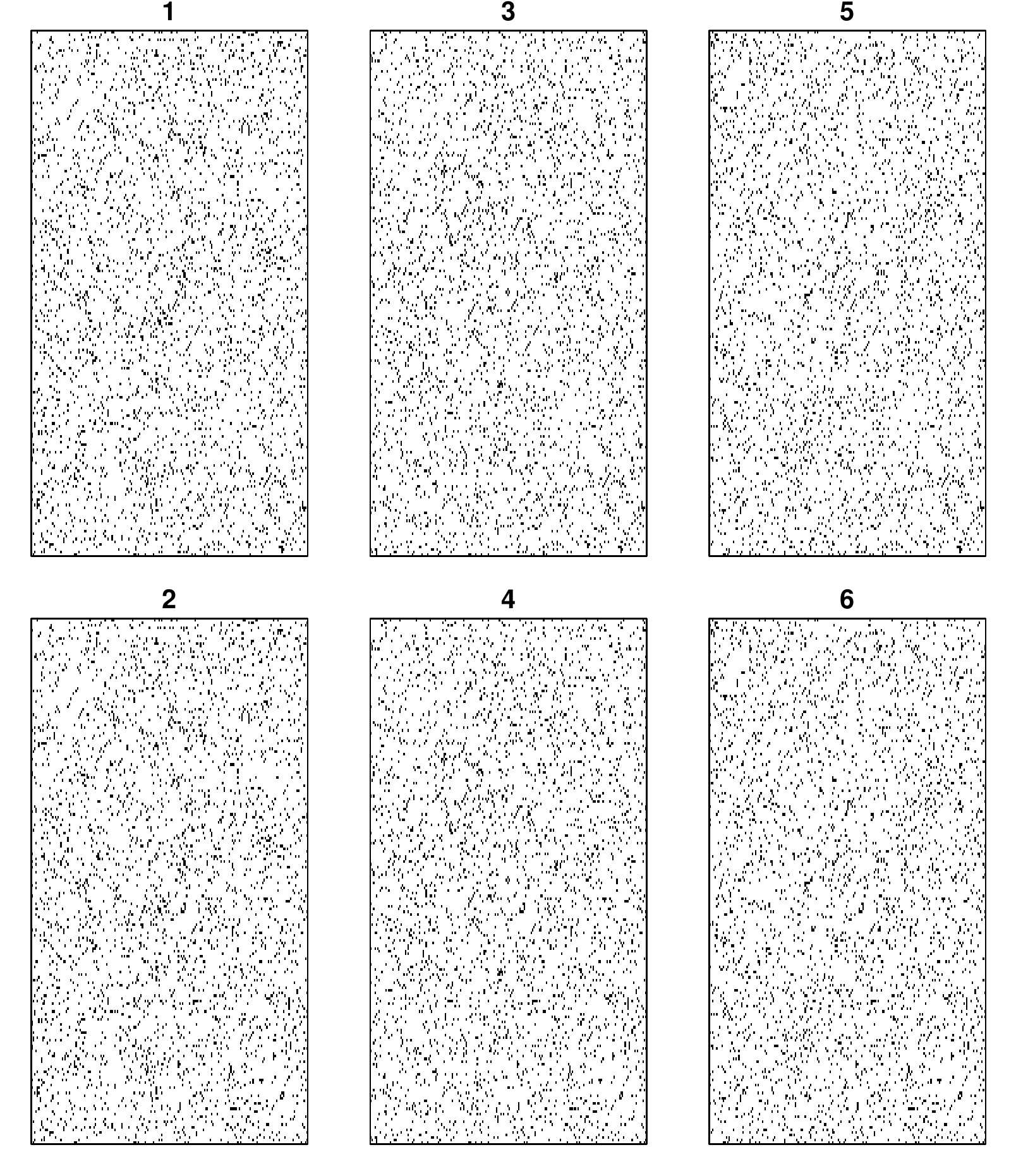} 
\caption{Image plots of the random sparse adjacency matrices from all graphical models. Graphs in the same row share the same sparsity pattern at the bottom right block, whereas graphs in the same column share the same pattern at remaining locations}
\label{Comparison 2}
\end{figure}

\begin{table}[h]
\centering 
\caption{Summary of model comparisons based on average of deviance measures across the 6 networks, for the case of random sparsity patterns}
\begin{tabular}{
l
S[table-format = 3]
S[table-format = 3]
S[table-format = 3]
S[table-format = 3]
S[table-format = 3]
S[table-format = 3]
S[table-format = 3]
S[table-format = 3]
S[table-format = 3]
S[table-format = 3]
S[table-format = 3]
}
\toprule
 & & {Glasso} & {JEM-G} & {GGL} & {JSEM} & {BJNS} \\
\cmidrule(lr){3-7}  
$\rho$ &&&{$n=200$}&&&\\
\midrule
{\multirow{3}{*}{0} } 
& {MC\%}& {47 (0.009)} & {50 (0.010)} & {47 (0.009)} & {\bf 61 (0.009)} & {\bf 57 (0.010)}  \\
& {SP\%}& {94 (0.003)} & {\bf 97 (0.001)} & {93 (0.003)} & {\bf 99 (0.001)} & {\bf 97 (0.001)}  \\
& {SE\%}& {\bf 60 (0.010)} & {48 (0.010)} & {\bf 61 (0.010)} & {46 (0.009)} & {59 (0.010)}  \\
\cmidrule(lr){1-2}  
\cmidrule(lr){3-7} 
{\multirow{3}{*}{0.04} } 
& {MC\%}& {\bf 40 (0.008)} & {35 (0.010)} & {35 (0.008)} & {32 (0.009)} & {\bf 40 (0.010)}  \\
& {SP\%}& {93 (0.003)} & {\bf 97 (0.001)} & {91 (0.003)} & {\bf 99 (0.001)} & {96 (0.002)}  \\
& {SE\%}& {\bf 46 (0.009)} & {30 (0.008)} & {\bf 46 (0.008)} & {18 (0.008)} & {\bf 39 (0.009)}  \\
\midrule
&&&{$n=300$}&&&\\
\midrule
{\multirow{3}{*}{0} } 
& {MC\%}& {54 (0.009)} & {60 (0.010)} & {54 (0.008)} & {\bf 73 (0.008)} & {\bf 70 (0.010)}  \\
& {SP\%}& {93 (0.003)} & {\bf 97 (0.001)} & {92 (0.002)} & {\bf 99 (0.001)} & {\bf 97 (0.001)}  \\
& {SE\%}& {72 (0.009)} & {61 (0.010)} & {\bf 75 (0.009)} & {63 (0.008)} & {\bf 77 (0.010)}  \\
\cmidrule(lr){1-2}  
\cmidrule(lr){3-7}  
{\multirow{3}{*}{0.04} } 
& {MC\%}& {\bf 47 (0.009)} & {43 (0.009)} & {42 (0.008)} & {42 (0.008)} & {\bf 51 (0.010)}  \\
& {SP\%}& {93 (0.003)} & {\bf 97 (0.001)} & {90 (0.003)} & {\bf 99 (0.001)} & {95 (0.002)}  \\
& {SE\%}& {\bf 58 (0.009)} & {39 (0.007)} & {\bf 58 (0.008)} & {29 (0.007)} & {\bf 53 (0.010)}  \\
\bottomrule
\end{tabular}
\label{comparison2:random sparsity}
\end{table}

\subsubsection{A Computational Strategy for Speeding Up BJNS for large $K$}

As presented in Algorithm \ref{BlockGibbsBCONCORD}, the Gibbs sampler updates all the $p(p-1)/2$ vectors $\boldsymbol{\theta}_{ij}$ based on their full conditional distributions. Although the conditional posterior distribution of $\boldsymbol{\theta}_{ij}$s is a mixture of univariate normal densities, 
all $2^K$ mixture probabilities $c_{l,ij}$ given in (\ref{mixture probs}) still need to be calculated, which in turn involves matrix-vector multiplications. Hence with increasing $K$, the computation complexity of the full decomposition  (\ref{eq21}) grows quickly.

Next, we discuss a strategy that starts by examining all {\em pairwise decompositions} to identify inactive pairwise components and the higher matrices $r$ of such pairwise components.
In the first step, ${K}\choose{2}$ pairwise joint models are considered; namely, for any pair $1\leq k_1 < k_2 \leq K$, we examine   
\begin{equation}\label{reduced model}
\begin{split}
\mathbf{\Omega}^{k_1} &= \mathbf{\Psi}^{k_1} +  \mathbf{\Psi}^{k_1k_2}, \\
\mathbf{\Omega}^{k_2} &= \mathbf{\Psi}^{k_2} +  \mathbf{\Psi}^{k_1k_2}.
\end{split}
\end{equation}

Subsequently, we remove all the pairwise matrices $\mathbf{\Psi}^{k_1k_2}$ that are considered  ``inactive", i.e have significantly fewer edges compared to other pairwise matrices. Once the ``inactive" pairwise components are identified, we then remove any higher order component $r$ that contains them. Next, we run the resulting reduced model and count the number of edges present in each component $r$ and calculate the number of edges present in the estimated matrices. We then further reduce the model if any matrix component seem to be inactive (has significantly smaller number of edges) and finally run BJNS one last time with the resulting reduced model.

The proposed purely computational strategy is illustrated on the setting given in Figure \ref{comparison2:random sparsity}, which involves $K=6$ groups with $p=200$ variables and $n_k= 200, 300$ samples per group. Note that the full decomposition would involve $r$-tuple interaction components, $r=2,\cdots,6$, which renders computations expensive, since each vector $\boldsymbol{\theta}_{ij}$ would have length $2^K -1 = 63$. However, as can be seen by the design of the simulation, most components $r$ are ``inactive". For instance, in the case of (\ref{true decomposition k=6}), only $5$ out of the $63$ components in the full model are non-zero (active). 

We start by first studying the ${6}\choose{2}$ pairwise interaction components. Figure \ref{fig:sub1} shows the number of edges in the pairwise matrices $\mathbf{\Psi}^{k_1k_2}$, $1 \leq k_1 < k_2 \leq 6$. 

\begin{figure}
\centering
\begin{subfigure}{.5\textwidth}
  \centering
  \includegraphics[height=2in,width=3in,angle=0]{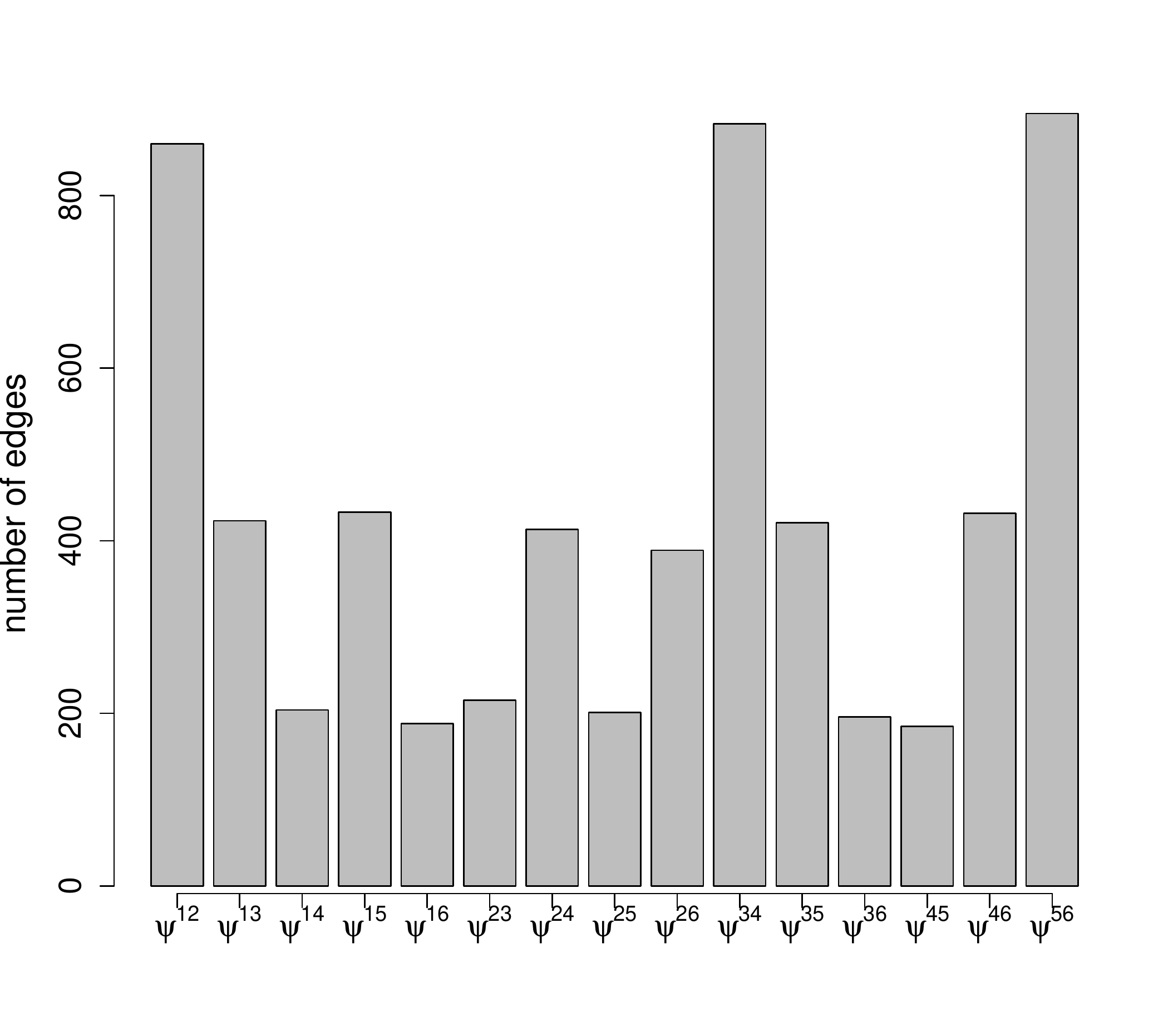}
  \caption{}
  \label{fig:sub1}
\end{subfigure}%
\begin{subfigure}{.5\textwidth}
  \centering
  \includegraphics[height=2in,width=3in,angle=0]{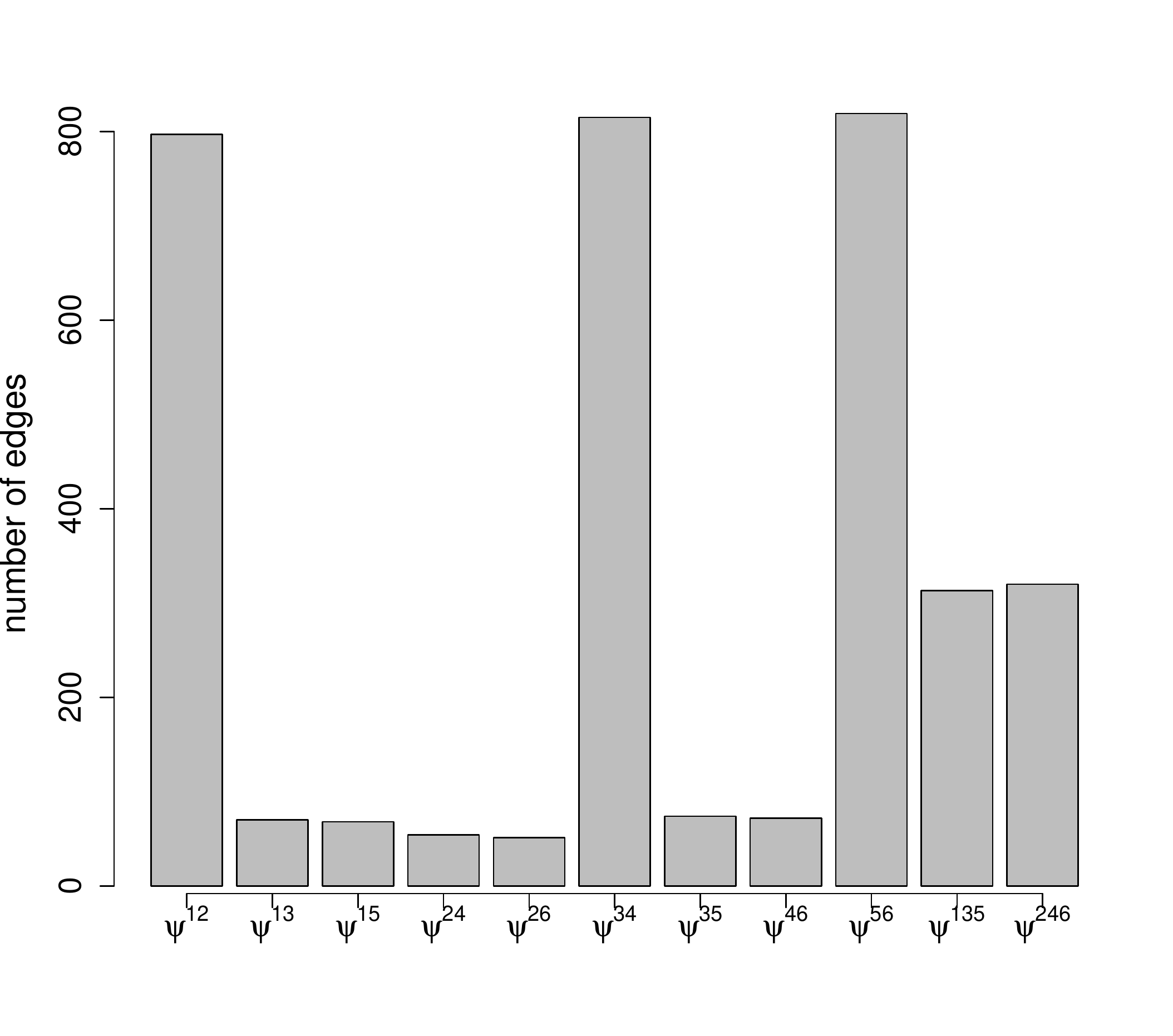}
  \caption{}
  \label{fig:sub2}
\end{subfigure}
\caption{Bar plots of the edge count of the pairwise joint components (a) and components of the reduced model (b)}
\label{fig:test}
\end{figure}

Based on this plot, one can see that matrices $\mathbf{\Psi}^{14}$, $\mathbf{\Psi}^{16}$, $\mathbf{\Psi}^{23}$, $\mathbf{\Psi}^{25}$, $\mathbf{\Psi}^{35}$, and $\mathbf{\Psi}^{45}$ have significantly smaller number of edges compared to other pairwise matrices. Hence, in the next step we remove all of these matrices and any higher order components, whose super scripts contain those of $\mathbf{\Psi}^{14}$, $\mathbf{\Psi}^{16}$, $\mathbf{\Psi}^{23}$, $\mathbf{\Psi}^{25}$, $\mathbf{\Psi}^{35}$, and $\mathbf{\Psi}^{45}$. 

Doing so results in a reduced model with components, $\mathbf{\Psi}^{12}$, $\mathbf{\Psi}^{13}$, $\mathbf{\Psi}^{15}$, $\mathbf{\Psi}^{24}$, $\mathbf{\Psi}^{26}$, $\mathbf{\Psi}^{34}$, $\mathbf{\Psi}^{35}$, $\mathbf{\Psi}^{46}$, $\mathbf{\Psi}^{56}$, $\mathbf{\Psi}^{135}$, and $\mathbf{\Psi}^{246}$. Thus, in the second step, we run BJNS with a decomposition that is based on the above components; the edge count for these components are shown in plot \ref{fig:sub2}. From which, it is clear that matrices $\mathbf{\Psi}^{13}$, $\mathbf{\Psi}^{15}$, $\mathbf{\Psi}^{24}$, $\mathbf{\Psi}^{26}$, $\mathbf{\Psi}^{35}$, and $\mathbf{\Psi}^{46}$ are redundant and should be removed from the model. This further model reduction achieves the true decomposition given in \ref{true decomposition k=6}. The results of BJNS running on the resulting reduced model can be read off from the Table \ref{compare reduced BJNS with JSEM}. As can be seen, BJNS outperforms JSEM, which has been supplied with complete information about the sparsity patterns.

\begin{table}[htp]
\centering 
\caption{Comparions between JSEM and BJNS when employing the step wise computational strategy; the results are based on 50 replications}
\begin{tabular}{
l
S[table-format = 3]
S[table-format = 3]
S[table-format = 3]
S[table-format = 3]
S[table-format = 3]
S[table-format = 3]
S[table-format = 3]
S[table-format = 3]
S[table-format = 3]
S[table-format = 3]
S[table-format = 3]
}
\toprule
 & {JSEM} & {BJNS}& {JSEM} & {BJNS} \\
\cmidrule(lr){2-3}  
\cmidrule(lr){4-5}
 & {$n=200$} & {}& {$n=300$} & {} \\
\cmidrule(lr){1-1}  
\cmidrule(lr){2-3} 
\cmidrule(lr){4-5}  
 {MC\%}&  {61 (0.009)} & {69 (0.011)} &  {73 (0.008)} & {83 (0.007)}  \\
 {SP\%}&  {99 (0.001)} & {99 (0.001)}&  {99 (0.001)} & {100 (0.001)}  \\
{SE\%}&   {46 (0.009)} & {56 (0.010)} &  {63 (0.008)} & {76 (0.010)}  \\ 
\bottomrule
\end{tabular}
\label{compare reduced BJNS with JSEM}
\end{table}

\subsubsection{Computational Cost of BJNS}
Lastly, as presented in Table \ref{computational cost of BJNS}, we investigate the computational cost associated with the above strategy, across varying values of $p$ and $n$. Each experiment was repeated 5 times and all computations were done sequentially using one processor (CPU). Note that around 60\% of the time in each experiment was spent on the first step which is investigating all the ${6}\choose{2}$ pairwise models. Since, the pairwise models are ran independently, one can use parallel computing and reduce the computational time in table \ref{computational cost of BJNS} potentially by 50\%. Finally, as can be seen in the last row of the table, the memory usage of BJNS is not necessarily large and that is due to the fact that the algorithm does not involve any matrix inversions or generation of samples from multivariate distributions. 

\begin{table}[htp]
\centering 
\caption{accuracy and cost of BJNS for varying values of $p$}
\begin{tabular}{
l
S[table-format = 3]
S[table-format = 3]
S[table-format = 3]
S[table-format = 3]
S[table-format = 3]
S[table-format = 3]
S[table-format = 3]
S[table-format = 3]
S[table-format = 3]
S[table-format = 3]
S[table-format = 3]
}
\toprule
{\multirow{4}{*}{}} 
 &{$p=200$}&{$p=500$}&{$p=700$}&{$p=1000$}\\
 &{$n=300$}&{$n=750$}&{$n=1050$}&{$n=1500$}\\

 \midrule
 {MC\%}&  {83} & {84} &  {85} & {85}   \\
 {SP\%}&  {100} & {100} &  {100} & {100}  \\
 {SE\%}&  {76} & {78} &  {78} & {78}   \\
 {hours} &  {2.9h} & {38.9h} &  {83.3h} & {214.5h}   \\
  {GigaBytes}&  {0.25gb} & {0.5gb} &  {0.6gb} & {0.9gb}  \\
\bottomrule
\end{tabular}
\label{computational cost of BJNS}
\end{table}

\section{An Application of BJNS to Metabolomics Data} \label{Omics data}
\noindent
In this section, we employ the proposed methodology to obtain networks across four groups of patients that participated in the Integrative Human Microbiome Project.
The data were downloaded from the Metabolomics Workbench \texttt{www.metabolomicsworkbench.org} (Study ID ST000923) and correspond to measurements of 428 primary and secondary metabolites and lipids from stool samples of  542 subjects, partitioned in the following groups: inflammatory bowel disease (IBD) patients (males $n_1=202$ and females $n_2=208$) and non-IBD controls (males $n_3=72$ and females $n_4=70$), Groups 1-4, respectively. Since there are two factors in the study design, the following model was fitted to the data.
\begin{equation*}
\begin{split}
\mathbf{\Omega}^1 &= \mathbf{\Psi}^1 + \mathbf{\Psi}^{12} + \mathbf{\Psi}^{13} + \mathbf{\Psi}^{1234}, \quad \quad \mathbf{\Omega}^3 = \mathbf{\Psi}^3 + \mathbf{\Psi}^{13} + \mathbf{\Psi}^{34} + \mathbf{\Psi}^{1234}, \\
\mathbf{\Omega}^2 &= \mathbf{\Psi}^2 + \mathbf{\Psi}^{12} + \mathbf{\Psi}^{24} + \mathbf{\Psi}^{1234}, \quad \quad  \mathbf{\Omega}^4 = \mathbf{\Psi}^4 + \mathbf{\Psi}^{24} + \mathbf{\Psi}^{34} + \mathbf{\Psi}^{1234}.
\end{split}
\end{equation*}

The results are shown in the next Table (set1: 289 lipids in red and set2: 139 metabolites in blue, set12: interaction edges between set1 and set2), for both the final estimates of the precision matrices, as well as the components in the proposed decomposition. It is interesting to note that a large number of edges are shared across all groups, indicating common patterns. Further, the component shared between male and female IBD patients has a fairly large number of edges, indicating that the disease status exhibits commonalities across both males and females. 
\begin{table}[H]
\centering
\caption{Number of edges in each matrix}
\begin{tabular}{
l
S[table-format = 3]
S[table-format = 3]
S[table-format = 3]
S[table-format = 3]
S[table-format = 3]
S[table-format = 3]
S[table-format = 3]
S[table-format = 3]
S[table-format = 3]
S[table-format = 3]
S[table-format = 3]
S[table-format = 3]
S[table-format = 3]
S[table-format = 3]
S[table-format = 3]
S[table-format = 3]
}

\toprule
{}& {$\mathbf{\Omega}^1$} & {$\mathbf{\Omega}^2$} & {$\mathbf{\Omega}^3$}  & {$\mathbf{\Omega}^4$} & {$\mathbf{\Psi}^1$} & {$\mathbf{\Psi}^2$} & {$\mathbf{\Psi}^3$}  & {$\mathbf{\Psi}^4$}  & {$\mathbf{\Psi}^{12}$}  & {$\mathbf{\Psi}^{13}$} &  {$\mathbf{\Psi}^{24}$} &  {$\mathbf{\Psi}^{34}$} & {$\mathbf{\Psi}^{1234}$} \\
 \midrule
 {set1} &{941}&{930} & {797}&{799} &{47}&{29}& {18}&{13}&{161}&{81}&{88}&{46}&{652}\\
 {set2} &{264}&{269}&{243}&{244} &{2}&{5}&{2}&{1}&{34}&{14}&{16}&{13}&{214}\\
 {set12}&{154}&{155} &{124}&{124} &{10}&{11}&{5}&{5}&{34}&{13}&{13}&{9}&{97}\\
\bottomrule
\end{tabular}
\label{data edge count table}
\end{table}

The next plot shows a network map of the common component shared across all groups. The primary and secondary metabolites are depicted in red, while the lipids in blue.
Not surprisingly, primary metabolites (those involved in cellular growth, development and reproduction) form a fairly strongly connected network. On the other hand,
the connectivity between lipids (whose functions include storing energy, signaling and acting as structural components of cell membranes) to the metabolites is not
particularly strong. On the other hand, different fairly strongly connected subnetworks amongst lipids are present, including dicylglycerols (DAG) with tricylglycerols (TAG)
that are main constituents of animal and vegetable fat (upper right corner of the plot) and various phospholipids (upper left corner of the plot). In general, the results reveal interesting patterns that can be used to understand progression of IBD disease.

\begin{figure}[H]
 \centering
 \caption{Network plot of the edges shared between the four groups ($\Psi^{1234}$); names of the metabolites and the lipids appear in red and blue, respectively.}
 \includegraphics[height=4in,width=5in,angle=0]{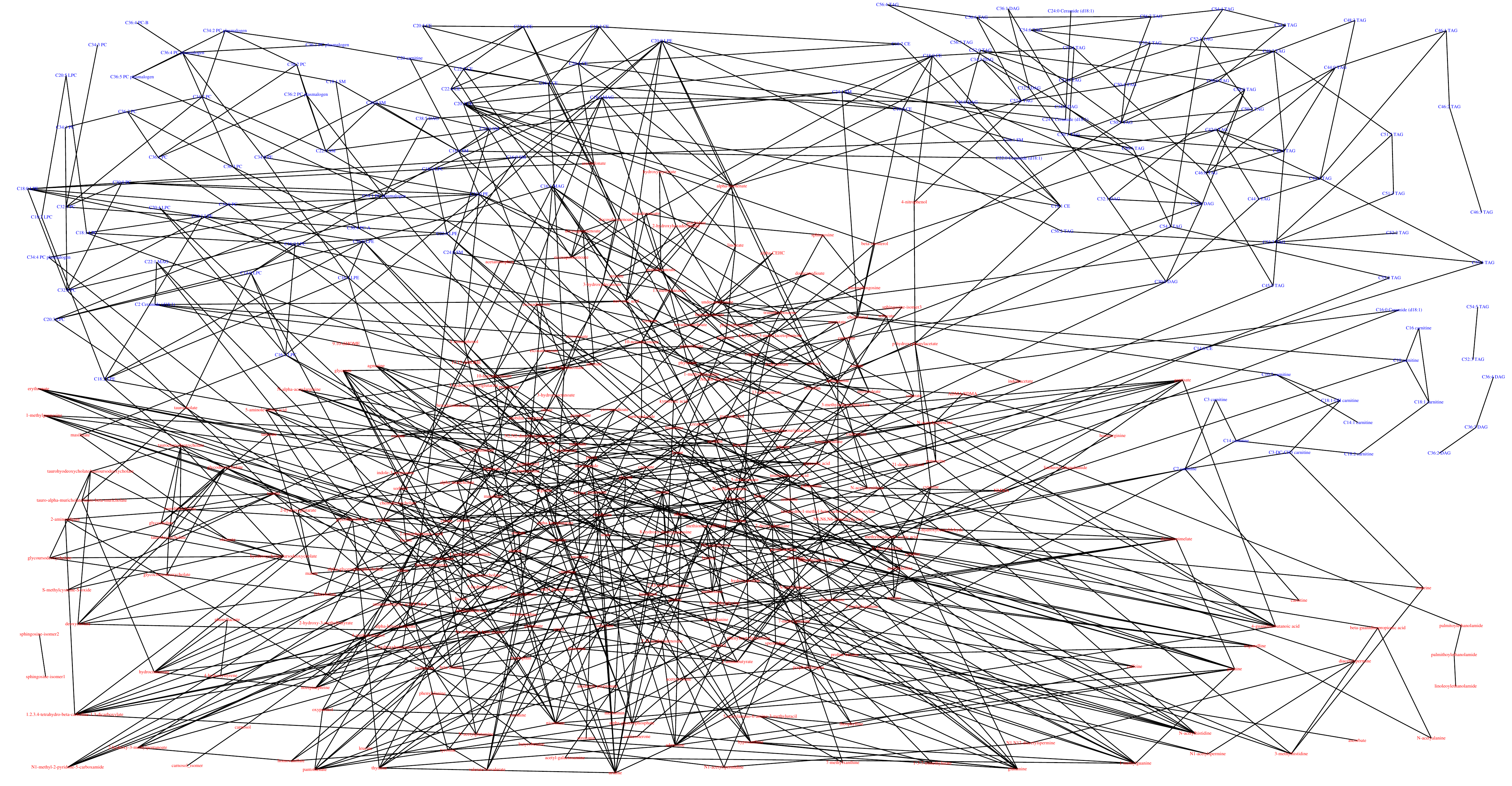} 
\end{figure}

\pagebreak

\addcontentsline{toc}{chapter}{References}
\bibliographystyle{plainnat}
\bibliography{reference}
\normalsize

\newpage
\begin{center}
\textbf{\large Supplemental Document for ``A Bayesian Approach to Joint Estimation of Multiple Graphical Models"}
\end{center}


\setcounter{equation}{0}
\setcounter{figure}{0}
\setcounter{Lemma}{0}
\setcounter{table}{0}
\setcounter{section}{0}
\setcounter{page}{1}
\renewcommand{\thesection}{S \arabic{section}}
\renewcommand{\theequation}{S \arabic{equation}}
\renewcommand{\thefigure}{S \arabic{figure}}
\renewcommand{\theLemma}{S \arabic{Lemma}}

\numberwithin{equation}{section}
\section{The Structures of $\mathbf{\Upsilon}$ and $\boldsymbol{a}$}
We first note that,
\begin{equation}\label{Theta'UpsilonTheta 1}
\begin{split}
\text{tr}\left[ \left( \sum\limits_{r \in \mathop \vartheta \nolimits_k } { \mathbf{\Psi}^{r} }\right)^2 \mathbf{S}^k \right] &=\frac{1}{n}\sum\limits_{k=1}^K\sum\limits_{i=1}^{n} \left[ \left( \sum\limits_{r \in \mathop \vartheta \nolimits_k } { \mathbf{\Psi}^{r} }\right) \mathbf{y}_{i:}^k \right]'\left[ \left( \sum\limits_{r \in \mathop \vartheta \nolimits_k } { \mathbf{\Psi}^{r} }\right) \mathbf{y}_{i:}^k \right] \\
& = \frac{1}{n} \sum\limits_{k=1}^K\sum\limits_{j=1}^p\sum\limits_{i=1}^n  \left[  \sum\limits_{r \in \mathop \vartheta \nolimits_k } { \mathbf{\Psi}^{r}_{j:} }\mathbf{y}_{i:}^k \right]^2 =\frac{1}{n}\sum\limits_{k=1}^K\sum\limits_{j=1}^p\sum\limits_{i=1}^n \sum\limits_{r \in \mathop \vartheta \nolimits_k } \left[ { \mathbf{\Psi}^{r}_{j:} }\mathbf{y}_{i:}^k \right]^2 \\
&+ \frac{2}{n}\sum\limits_{k=1}^K\sum\limits_{j=1}^p\sum\limits_{i=1}^n  \mathop{\sum_{r \in \mathop \vartheta \nolimits_k }\sum_{s \in \mathop \vartheta \nolimits_k }}_{r\neq s} \left[ { \mathbf{\Psi}^{r}_{j:} }\mathbf{y}_{i:}^k \right] \left[ { \mathbf{\Psi}^{s}_{j:} }\mathbf{y}_{i:}^k \right],
\end{split}
\end{equation}

where $\left\{ \mathbf{y}_{i:}^k\right\}_{i=1}^{n_k}$ denote $p$-dimensional observations for group $k$. 
Next, for any $1\leq j\leq p$, $1 \leq k \leq K$, and ${r \in \mathop \vartheta \nolimits_k }$
\begin{equation}\label{Theta'UpsilonTheta 2}
\begin{split}
\frac{1}{n}\sum\limits_{i=1}^n \sum\limits_{r \in \mathop \vartheta \nolimits_k }\left[ { \mathbf{\Psi}^{r}_{j:} }\mathbf{y}_{i:}^k \right]^2 &= \frac{1}{n}\sum\limits_{r \in \mathop \vartheta \nolimits_k }\sum\limits_{i=1}^n\left(\sum\limits_{l=1}^p {\psi}^{r}_{jl}  \mathbf{y}_{il}^k \right)^2 \\
&= \sum\limits_{r \in \mathop \vartheta \nolimits_k }\sum_{l=1}^p\left({\psi}^{r}_{jl} \right)^2s_{ll}^k +  2\sum\limits_{r \in \mathop \vartheta \nolimits_k }\mathop{\sum_{l=1}^p\sum_{m=1}^p}_{l \neq m}\left({\psi}^{r}_{jl} {\psi}^{r}_{jm}\right)s_{lm}^k.
\end{split}
\end{equation}
Similarly, for any $1\leq j\leq p$, $1 \leq k \leq K$, and ${(r\neq s) \in \mathop \vartheta \nolimits_k }$,

\begin{equation}\label{Theta'UpsilonTheta 3}
\begin{split}
\frac{1}{n}\sum\limits_{i=1}^n  \mathop{\sum_{r \in \mathop \vartheta \nolimits_k }\sum_{s \in \mathop \vartheta \nolimits_k }}_{r\neq s} \left[ { \mathbf{\Psi}^{r}_{j:} }\mathbf{y}_{i:}^k \right] \left[ { \mathbf{\Psi}^{s}_{j:} }\mathbf{y}_{i:}^k \right] &= \sum\limits_{r \in \mathop \vartheta \nolimits_k }\sum_{l=1}^p\left({\psi}^{r}_{jl} {\psi}^{s}_{jl} \right)s_{ll}^k\\
& +  2\mathop{\sum_{r \in \mathop \vartheta \nolimits_k }\sum_{s \in \mathop \vartheta \nolimits_k }}_{r\neq s} \mathop{\sum_{l=1}^p\sum_{m=1}^p}_{l \neq m}\left({\psi}^{r}_{jl} {\psi}^{s}_{jm}\right)s_{lm}^k.
\end{split}
\end{equation} 

Thus, by combining (\ref{Theta'UpsilonTheta 1}), (\ref{Theta'UpsilonTheta 2}), and (\ref{Theta'UpsilonTheta 3}), we have that  

\begin{equation}
\begin{split}
\text{tr}\left[ \left( \sum\limits_{r \in \mathop \vartheta \nolimits_k } { \mathbf{\Psi}^{r} }\right)^2 \mathbf{S}^k \right]  =\left( {\begin{array}{*{20}{c}}
{\mathbf{\Theta}'} &
{\mathbf{\Delta}'} \\
\end{array} } \right) \left( {\begin{array}{*{20}{c}}
{\mathbf{\Upsilon}} & {\mathbf{A}} \\
{\mathbf{A}'} & {\mathbf{D}} \\
\end{array} } \right) \left( {\begin{array}{*{20}{c}}
{\mathbf{\Theta}} \\
{\mathbf{\Delta}} \\
\end{array} } \right),
\end{split}
\end{equation}

The matrix $\mathbf{\Upsilon}$ is as follows

\begin{equation}\label{upsilon}
\mathbf{\Upsilon} = \left( {\begin{array}{*{20}{c}}
   {\mathbf{B}_1 + \mathbf{B}_2 + \mathbf{B}_3} 
& {\mathbf{B}_2 + \mathbf{B}_3} 
& {\mathbf{B}_1 + \mathbf{B}_3} 
& {\mathbf{B}_1 + \mathbf{B}_2} 
& {\mathbf{B}_3} 
& {\mathbf{B}_2} 
& {\mathbf{B}_1}\\ 
   {\mathbf{B}_2 + \mathbf{B}_3} 
& {\mathbf{B}_2 + \mathbf{B}_3} 
& {\mathbf{B}_3} 
& {\mathbf{B}_2} 
& {\mathbf{B}_3} 
& {\mathbf{B}_2} 
& {\mathbf{0}}\\ 
   {\mathbf{B}_1 + \mathbf{B}_3} 
& {\mathbf{B}_3 } 
& {\mathbf{B}_1 + \mathbf{B}_3} 
& {\mathbf{B}_1} 
& {\mathbf{B}_3} 
& {\mathbf{0}} 
& {\mathbf{B}_1}\\
   {\mathbf{B}_1 + \mathbf{B}_2} 
& {\mathbf{B}_2} 
& {\mathbf{B}_1} 
& {\mathbf{B}_1 + \mathbf{B}_2} 
& {\mathbf{0}} 
& {\mathbf{B}_2} 
& {\mathbf{B}_1}\\
   {\mathbf{B}_3} 
& {\mathbf{B}_3} 
& {\mathbf{B}_3} 
& {\mathbf{0}} 
& {\mathbf{B}_3} 
& {\mathbf{0}} 
& {\mathbf{0}}\\
{\mathbf{B}_2} 
& {\mathbf{B}_2} 
& {\mathbf{0}} 
& {\mathbf{B}_2} 
& {\mathbf{0}} 
& {\mathbf{B}_2} 
& {\mathbf{0}}\\
{\mathbf{B}_1} 
& {\mathbf{0}} 
& {\mathbf{B}_1} 
& {\mathbf{B}_1} 
& {\mathbf{0}} 
& {\mathbf{0}} 
& {\mathbf{B}_1}\\
\end{array} } \right), 
\end{equation}
where, $\mathbf{B}_k$s are $\frac{p(p-1)}{2} \times \frac{p(p-1)}{2}$ symmetric matrices. To understand the structure of the matrices $\mathbf{B}_k$s, we index it's rows and columns as $(12, 13, ..., p-1p)$. Then,
\begin{equation}\label{structure of Bs}
\mathbf{B}_{\left(ab, cd\right)}^k = \left\{
	\begin{array}{ll}
		s_{aa}^{k} + s_{bb}^{k}  & \mbox{if } a=b \ \& \ c=d , \\
				s_{ac}^{k}  & \mbox{if } b=d \ \& \ a \neq c,\\
						s_{bd}^{k}  & \mbox{if } a=c \ \& \ b\neq d,\\
								0  & \mbox{if } a \neq b \ \& \ c \neq d ,\\
	\end{array}
\right. \quad \text{for} \quad 1 \leq a < b \leq p, \ \text{and} \ 1 \leq c < d \leq p.
\end{equation}
For further illustration, when $p=5$, $\mathbf{B}^k$ is as follows,
\begingroup\makeatletter\def\f@size{8}\check@mathfonts\begin{equation*}
\left( {\begin{array}{*{20}{c}}
{s_{11}^{k} + s_{22}^{k}} & {s_{23}^{k}} & {s_{24}^{k}} & {s_{25}^{k}} & {s_{13}^{k}} & {s_{14}^{k}} & {s_{15}^{k}} & {0} & {0} & {0} \\ 
{s_{23}^{k}} & {s_{11}^{k}+s_{33}^{k}} & {s_{34}^{k}} & {s_{35}^{k}} & {s_{12}^{k}} & {0} & {0} & {s_{14}^{k}} & {s_{15}^{k}} & {0} \\ 
{s_{24}^{k}} & {s_{34}^{k}} & {s_{11}^{k} + s_{44}^{k}} & {s_{45}^{k}} & {0} & {s_{12}^{k}} & {0} & {s_{13}^{k}} & {0} & {s_{15}^{k}} \\ 
{s_{25}^{k}} & {s_{35}^{k}} & {s_{45}^{k}} & {s_{11}^{k}+s_{55}^{k}} & {0} & {0} & {s_{12}^{k}} & {0} & {s_{13}^{k}} & {s_{14}^{k}} \\ 
{s_{13}^{k}} & {s_{12}^{k}} & {0} & {0} & {s_{22}^{k}+s_{33}^{k}} & {s_{34}^{k}} & {s_{35}^{k}} & {s_{24}^{k}} & {s_{25}^{k}} & {0} \\ 
{s_{14}^{k}} & {0} & {s_{12}^{k}} & {0} & {s_{34}^{k}} & {s_{22}^{k}+s_{44}^{k}} & {s_{45}^{k}} & {s_{23}^{k}} & {0} & {s_{25}^{k}} \\ 
{s_{15}^{k}} & {0} & {0} & {s_{12}^{k}} & {s_{35}^{k}} & {s_{45}^{k}} & {s_{22}^{k}+s_{55}^{k}} & {0} & {s_{23}^{k}} & {s_{24}^{k}} \\ 
{0} & {s_{14}^{k}} & {s_{13}^{k}} & {0} & {s_{24}^{k}} & {s_{23}^{k}} & {0} & {s_{33}^{k}+s_{44}^{k}} & {s_{45}^{k}} & {s_{35}^{k}} \\ 
{0} & {s_{15}^{k}} & {0} & {s_{13}^{k}} & {s_{25}^{k}} & {0} & {s_{23}^{k}} & {s_{45}^{k}} & {s_{33}^{k}+s_{55}^{k}} & {s_{34}^{k}} \\ 
{0} & {0} & {s_{15}^{k}} & {s_{14}^{k}} & {0} & {s_{25}^{k}} & {s_{24}^{k}} & {s_{35}^{k}} & {s_{34}^{k}} & {s_{44}^{k}+s_{55}^{k}} \\ 
\end{array} } \right),
\end{equation*}
\endgroup
 
Also, the vector $\boldsymbol{a}$ for the special case of $K=3$, is given as
\begin{equation}\label{a}
\boldsymbol{a}= \left( {\begin{array}{*{20}{c}}
{\boldsymbol{a}^1 + \boldsymbol{a}^2 + \boldsymbol{a}^3} \\ 
{\boldsymbol{a}^2 + \boldsymbol{a}^3} \\ 
{\boldsymbol{a}^1 + \boldsymbol{a}^3} \\ 
{\boldsymbol{a}^1 + \boldsymbol{a}^2} \\ 
{\boldsymbol{a}^3 } \\ 
{\boldsymbol{a}^2} \\ 
{\boldsymbol{a}^1} \\ 
\end{array} } \right),
\end{equation}
with $\boldsymbol{a}_k$ being the following vector,
\begin{equation*}
\boldsymbol{a}^k=(s_{12}^{k}(\psi^k_{11}+\psi^k_{22}), ..., s_{1p}^{k}(\psi^k_{11}+\psi^k_{pp}),..., s_{p-1p}^{k}(\psi^k_{p-1p-1}+\psi^k_{pp}))', \quad \quad k=1,2,3.
\end{equation*}

\section{An Algorithm for Generating Samples from Posterior Coniditional Density of $\psi_{ii}^k$} \label{sampling from CPD of psi_ii}

Using the fact that it has a unique mode 
at
\begin{equation*}
\frac{-b_i^k + \sqrt{(b_i^k )^2 + 4n^2 \sigma_{ii}^k}}{2n\sigma_{ii}^k},
\end{equation*}

\noindent
one can use a discretization technique, as described in Algorithm \ref{generate samples from the dist of diagonals}, to 
generate samples from it. 
  
\begin{algorithm}[h]
\caption{Generating samples from density in (4.8)} \label{generate samples from the dist of 
diagonals}
\begin{algorithmic}[1]
\Procedure{Generating sample from density in (4.8)}{}\Comment{Input $n, \sigma_{ii}^k, b_i^k$ }
\State $\text{mode} \gets \frac{-b_i^k + \sqrt{(b_i^k )^2 + 4n^2 \sigma_{ii}^k}}{2n\sigma_{ii}^k}$
\State $S \gets \text{seq}(0, 6 \times\text{mode}, 0.001)$ \Comment{sequence from $0$ to $6 \times \text{mode}$ with 
increments $0.001$}
\For{$t \in S$}
\State $p_t \gets \exp \left\{n \log t -\frac{n}{2} s_{ii}^k t^2 - b_i^k  t\right\}$
\EndFor
\State Set $\text{sum} =\sum\limits_{t \in S} p_t$
\If{$\text{sum} \gets \infty$}
\For{$t \in S$}
\State $p_t \gets \frac{\exp \left\{n \log t -\frac{n}{2} s_{ii}^k t^2 - b_i^k  t\right\}}{\exp \left\{n \log \text{mode} -\frac{n}{2} s_{ii}^k \text{mode}^2 - b_i^k  \text{mode}\right\}}$
\EndFor
\State Set $\text{sum} \gets \sum\limits_{t \in S} p_t$
\EndIf
\State ${x} \gets \text{sample}\left(1, S, \text{probs} \propto \left\{p_t: t \in S\right\} \right)$ 
\State \textbf{return} $x$ \Comment{Output $x$}
\EndProcedure
\end{algorithmic}
\end{algorithm}
However, we have observed in the numerical work undertaken that the density (4.8) exhibits a high pick at its mode. As a result, one can simply approximate it using a degenerate density with a point mass at $\frac{-b_i^k + \sqrt{(b_i^k )^2 + 4n^2 \sigma_{ii}^k}}{2n\sigma_{ii}^k}$. This approximation allows faster implementation of the algorithm without much sacrificing on its accuracy.

\section{Proofs of Theorems 1 and 2 }\label{Proos of theorem 1 and 2}
By Assumption 3 and Hanson-Wright inequality from \cite{rudelson2013hanson}, there exists a $c>0$, independent of $n$ and $K$, such that

\begin{equation*}\label{PC1n}
P\left\{ \max\limits_{i,j,k}\| s_{ij}^k - \sigma_{ij}^k \| < c \sqrt{\frac{\log p}{n}}\right\} \geq 1 - \frac{1}{p^2},
\end{equation*}
and,
\begin{equation*}\label{PC2n}
P\left\{ \max\limits_{i,j,k}\| {\mathbf{\Omega}_{:i}^{k,0}}' {\mathbf{S}}^k_{:j}\| < c \sqrt{\frac{\log p}{n}}\right\} \geq 1 - \frac{1}{p^2}.
\end{equation*}

Define the events $C_{1,n}$, $C_{2,n}$ as 

\begin{equation}\label{C1n}
C_{1,n} = \left\{ \max\limits_{i,j,k}\| s_{ij}^k - \sigma_{ij}^k \| < c \sqrt{\frac{\log p}{n}}\right\},
\end{equation}

\begin{equation}\label{C2n}
C_{2,n} = \left\{ \max\limits_{i,j,k}\| {\mathbf{\Omega}_{:i}^{k,0}}' {\mathbf{S}}^k_{:j}\| < c \sqrt{\frac{\log p}{n}}\right\},
\end{equation}
for the next series of lemmas, we restrict ourself to the event $C_{1,n} \cap C_{2,n}$.

The next two lemmas prove important properties of the matrix $\mathbf{\Upsilon}$ and it's components $\mathbf{B}_k$, $k=1,..., K$.
\begin{Lemma}\label{Lemma 2}
For any $k=1, ..., K$, the following holds
\begin{equation}\label{•}
\text{eig}_{min}\left( \mathbf{S}^k\right) \le \text{eig}_{min}\left( \mathbf{B}^k\right) \le \text{eig}_{max}\left(\mathbf{B}^k\right) \le \text{eig}_{max}\left( \mathbf{S}^k\right).
\end{equation}
\end{Lemma}

\begin{proof}
Let $\boldsymbol{y} = \boldsymbol{y}\left( \mathbf{\Omega}^k\right)$ be a vectorized version of $\mathbf{\Omega}^k$ obtained by shifting the corresponding diagonal entry at the bottom of each column of $\mathbf{\Omega}^k$ and then stacking the columns on top of each other. Let $\mathbf{P}^i$ be the $p\times p$ permutation matrix such that $\mathbf{P}^i \boldsymbol{z} = \left( z_1, ..., z_{i-1}, z_{i+1}, ..., z_p, z_i\right)$ for every $\boldsymbol{z}\in \mathbb{R}^p$. It follows by the definition of $\boldsymbol{y}$ that 
\begin{equation*}
\boldsymbol{y} = \boldsymbol{y}\left( \mathbf{\Omega}^k\right) = \left( \left(\mathbf{P}^1\mathbf{\Omega}^k_{:1} \right)', \left(\mathbf{P}^2\mathbf{\Omega}^k_{:2} \right)', ..., \left(\mathbf{P}^p\mathbf{\Omega}^k_{:p} \right)'\right)'.
\end{equation*}
Let $\boldsymbol{x} \in \mathbb{R}^{\frac{p(p+1)}{2}}$ be the symmetric version of $\boldsymbol{y}$ obtained by removing all $\omega_{k,ij}$ with $i > j$. More precisely,
\begin{equation*}
\boldsymbol{x}= \left( \omega^k_{11}, \omega^k_{12}, \omega^k_{22}, ..., \omega^k_{1p},..., \omega^k_{pp}\right)'.
\end{equation*}

Let $\tilde{\mathbf{P}}$ be the $p^2 \times \frac{p(p+1)}{2}$ matrix such that every entry of $\tilde{\mathbf{P}}$ is either zero or one, exactly one entry in each row of $\tilde{\mathbf{P}}$ is equal to 1, and $\boldsymbol{y} = \tilde{\mathbf{P}}\boldsymbol{x}$.  

Now, define $\underline{\boldsymbol{\omega}}^k = \left( \omega^k_{12}, \omega^k_{13}, ..., \omega^k_{p-1p}\right)'$ and $\underline{\boldsymbol{\delta}}_{\mathbf{\Omega}^k} = \left( \omega^k_{11}, \omega^k_{22}, ..., \omega^k_{pp}\right)'$ and let $\tilde{\mathbf{Q}}$ be the $\frac{p(p+1)}{2} \times \frac{p(p+1)}{2}$ permutation matrix for which 
\begin{equation*}
\boldsymbol{x} = \mathbf{Q}\left( {\begin{array}{*{20}{c}}
{\underline{\boldsymbol{\omega}}^k} \\
{\underline{\boldsymbol{\delta}}_{\mathbf{\Omega}^k}} 
\end{array} } \right) .
\end{equation*}

Let $\tilde{\mathbf{\Sigma}}^k$ be a $p^2 \times p^2$ block diagonal matrix with $p$ diagonal blocks, the $i^{\text{th}}$ block is equal to $\tilde{\mathbf{\Sigma}}^{k,i}:=\mathbf{P}^i\mathbf{S}^k{\mathbf{P}^i}'$. It follows that 

\begin{equation*}
\begin{split}
\text{tr}\left[ \left(\mathbf{\Omega}^k\right)^2 \mathbf{S}^k\right] &= \sum\limits_{i=1}^{p} {\mathbf{\Omega}^k_{:i}}' \mathbf{S}^k \mathbf{\Omega}^k_{:i} = \sum\limits_{i=1}^{p} {\mathbf{\Omega}^k_{:i}}' {\mathbf{P}^i }'\mathbf{P}^i \mathbf{S}^k {\mathbf{P}^i }'\mathbf{P}^i\mathbf{\Omega}^k_{:i} = \sum\limits_{i=1}^{p} {\mathbf{\Omega}^k_{:i}}' {\mathbf{P}^i }'\left( \mathbf{P}^i \mathbf{S}^k {\mathbf{P}^i }'\right)\mathbf{P}^i\mathbf{\Omega}^k_{:i}\\
& = \boldsymbol{y}'\tilde{\mathbf{\Sigma}}^k\boldsymbol{y} = \boldsymbol{x}'\tilde{\mathbf{P}}'\tilde{\mathbf{\Sigma}}^k\tilde{\mathbf{P}}\boldsymbol{x}= \left( {\underline{\boldsymbol{\omega}}^k}', \underline{\boldsymbol{\delta}}_{{\mathbf{\Omega}^k}}' \right) \mathbf{Q}' \tilde{\mathbf{P}}'\tilde{\mathbf{\Sigma}}^k\tilde{\mathbf{P}} \mathbf{Q} \left( {\begin{array}{*{20}{c}}
{\underline{\boldsymbol{\omega}}^k} \\
{\underline{\boldsymbol{\delta}}_{\mathbf{\Omega}^k}} 
\end{array} } \right). 
\end{split}
\end{equation*}
There also exist appropriate matrices $\mathbf{A}^k$ and $\mathbf{D}^k$ such that 
\begin{equation*}
\text{tr}\left[ \left(\mathbf{\Omega}^k\right)^2 \mathbf{S}^k\right] = \left( {\underline{\boldsymbol{\omega}}^k}', \underline{\boldsymbol{\delta}}_{{\mathbf{\Omega}^k}}' \right) \left( {\begin{array}{*{20}{c}}
{\mathbf{B}^k} & {\mathbf{A}^k} \\
{\mathbf{A}^k} & {\mathbf{D}^k} \\
\end{array} } \right) \left( {\begin{array}{*{20}{c}}
{\underline{\boldsymbol{\omega}}}^k\\
{\underline{\boldsymbol{\delta}}_{\mathbf{\Omega}^k}} 
\end{array} } \right), 
\end{equation*}
therefore, we must have 
\begin{equation*}
 \mathbf{Q}' \tilde{\mathbf{P}}'\tilde{\mathbf{\Sigma}}^k\tilde{\mathbf{P}} \mathbf{Q} = \left( {\begin{array}{*{20}{c}}
{\mathbf{B}^k} & {\mathbf{A}^k} \\
{\mathbf{A}^k} & {\mathbf{D}^k} \\
\end{array} } \right). 
\end{equation*}
Now, since $\tilde{\mathbf{P}} \mathbf{Q}$ is orthogonal, we conclude that the eigenvalues of $\tilde{\mathbf{\Sigma}}^k$ and $\left( {\begin{array}{*{20}{c}}
{\mathbf{B}^k} & {\mathbf{A}^k} \\
{\mathbf{A}^k} & {\mathbf{D}^k} \\
\end{array} } \right) $ are the same. Moreover, the diagonal blocks of $\tilde{\mathbf{\Sigma}}^k$ all have the same eigenvalues as $\mathbf{S}^k$, and $\mathbf{B}^k$ can be regarded as a principal sub matrix of $\tilde{\mathbf{\Sigma}}^k$, hence, we have that
\begin{equation*}
\text{eig}_{min}\left( \mathbf{S}^k\right) = \text{eig}_{min}\left( \tilde{\mathbf{\Sigma}}^k\right) \le \text{eig}_{min}\left( \mathbf{B}^k\right) \le \text{eig}_{max}\left(\mathbf{B}^k\right) \le \text{eig}_{max}\left( \tilde{\mathbf{\Sigma}}^k\right) = \text{eig}_{max}\left( \mathbf{S}^k\right).
\end{equation*}
\end{proof}

\begin{Lemma}\label{Lemma 3}
Let ${\boldsymbol{\ell}} \in \mathcal{L}$ be any sparsity pattern/model with $d_{\boldsymbol{\ell}} < \frac{\varepsilon_0}{4c}\sqrt{\frac{n}{\log p}}$, then the sub matrix $\mathbf{\Upsilon}_{{\boldsymbol{\ell}}{\boldsymbol{\ell}}}$ of $\mathbf{\Upsilon}$, obtained by taking out all the rows and columns corresponding to the zero coordinates in $\mathbf{\Theta}\in \mathcal{M}_{{\boldsymbol{\ell}}}$, is positive definite. Specifically,
\begin{equation}
\frac{3K\varepsilon_0}{4} \leq \text{eig}_{\min}\left( \mathbf{\Upsilon}_{{\boldsymbol{\ell}} {\boldsymbol{\ell}}} \right) \leq \text{eig}_{\max}\left( \mathbf{\Upsilon}_{{\boldsymbol{\ell}} {\boldsymbol{\ell}}} \right) \leq \frac{3K}{2\varepsilon_0}, \quad \forall \boldsymbol{\ell} \in \mathcal{L}.
\end{equation}

\end{Lemma}
\begin{proof}
For ease of exposition, we show this result holds for the case of $K=3$. The proof for a general case will follow exactly from the same argument. Let $\boldsymbol{x}$ be a $d_{\boldsymbol{\ell}} \times 1$ vector in $\mathbb{R}^{d_{\boldsymbol{\ell}}}$ and partition $\boldsymbol{x}$ as 
\begin{equation*}
\boldsymbol{x}=\left( \boldsymbol{x}_{{\underline{\boldsymbol{\psi}}}^{{\boldsymbol{\ell}}, 123}}', \boldsymbol{x}_{{\underline{\boldsymbol{\psi}}}^{{\boldsymbol{\ell}}, 23}}', \boldsymbol{x}_{{\underline{\boldsymbol{\psi}}}^{{\boldsymbol{\ell}},13}}', \boldsymbol{x}_{{\underline{\boldsymbol{\psi}}}^{{\boldsymbol{\ell}},12}}',
\boldsymbol{x}_{{\underline{\boldsymbol{\psi}}}^{{\boldsymbol{\ell}},3}}',
\boldsymbol{x}_{{\underline{\boldsymbol{\psi}}}^{{\boldsymbol{\ell}},2}}',
\boldsymbol{x}_{{\underline{\boldsymbol{\psi}}}^{{\boldsymbol{\ell}},1}}'
\right),
\end{equation*}
then, by making similar partitions on each block of $\mathbf{\Upsilon}_{{\boldsymbol{\ell}}{\boldsymbol{\ell}}}$ (see $\mathbf{\Upsilon}$ in (\ref{upsilon})), we have that
\begin{equation*}
\begin{split}
\boldsymbol{x}'\mathbf{\Upsilon}_{{\boldsymbol{\ell}}{\boldsymbol{\ell}}}\boldsymbol{x} =& \left( \boldsymbol{x}_{{\underline{\boldsymbol{\psi}}}^{{\boldsymbol{\ell}},123}}', \boldsymbol{x}_{{\underline{\boldsymbol{\psi}}}^{{\boldsymbol{\ell}},13}}', \boldsymbol{x}_{{\underline{\boldsymbol{\psi}}}^{{\boldsymbol{\ell}},12}}',
\boldsymbol{x}_{{\underline{\boldsymbol{\psi}}}^{{\boldsymbol{\ell}},1}}'
\right)'\mathbf{B}^{1}_* \left( \boldsymbol{x}_{{\underline{\boldsymbol{\psi}}}^{{\boldsymbol{\ell}},123}}', \boldsymbol{x}_{{\underline{\boldsymbol{\psi}}}^{{\boldsymbol{\ell}},13}}', \boldsymbol{x}_{{\underline{\boldsymbol{\psi}}}^{{\boldsymbol{\ell}},12}}',
\boldsymbol{x}_{{\underline{\boldsymbol{\psi}}}^{{\boldsymbol{\ell}},1}}'
\right)\\
+&\left( \boldsymbol{x}_{{\underline{\boldsymbol{\psi}}}^{{\boldsymbol{\ell}},123}}', \boldsymbol{x}_{{\underline{\boldsymbol{\psi}}}^{{\boldsymbol{\ell}},23}}', \boldsymbol{x}_{{\underline{\boldsymbol{\psi}}}^{{\boldsymbol{\ell}},12}}',
\boldsymbol{x}_{{\underline{\boldsymbol{\psi}}}^{{\boldsymbol{\ell}},2}}'
\right)'\mathbf{B}^{2}_* \left( \boldsymbol{x}_{{\underline{\boldsymbol{\psi}}}^{{\boldsymbol{\ell}},123}}', \boldsymbol{x}_{{\underline{\boldsymbol{\psi}}}^{{\boldsymbol{\ell}},23}}', \boldsymbol{x}_{{\underline{\boldsymbol{\psi}}}^{{\boldsymbol{\ell}},12}}',
\boldsymbol{x}_{{\underline{\boldsymbol{\psi}}}^{{\boldsymbol{\ell}},2}}'
\right)\\
+&\left( \boldsymbol{x}_{{\underline{\boldsymbol{\psi}}}^{{\boldsymbol{\ell}},123}}', \boldsymbol{x}_{{\underline{\boldsymbol{\psi}}}^{{\boldsymbol{\ell}},23}}', \boldsymbol{x}_{{\underline{\boldsymbol{\psi}}}^{{\boldsymbol{\ell}},13}}',
\boldsymbol{x}_{{\underline{\boldsymbol{\psi}}}^{{\boldsymbol{\ell}},3}}'
\right)'\mathbf{B}^{3}_* \left( \boldsymbol{x}_{{\underline{\boldsymbol{\psi}}}^{{\boldsymbol{\ell}},123}}', \boldsymbol{x}_{{\underline{\boldsymbol{\psi}}}^{{\boldsymbol{\ell}},23}}', \boldsymbol{x}_{{\underline{\boldsymbol{\psi}}}^{{\boldsymbol{\ell}},13}}',
\boldsymbol{x}_{{\underline{\boldsymbol{\psi}}}^{{\boldsymbol{\ell}},3}}'
\right),
\end{split}
\end{equation*}
where,
\begin{equation*}
\mathbf{B}^{*}_1 = \left( {\begin{array}{*{20}{c}}
{\mathbf{B}^1_{{\underline{\boldsymbol{\psi}}}^{{\boldsymbol{\ell}},123}{\underline{\boldsymbol{\psi}}}^{{\boldsymbol{\ell}},123}}}& {\mathbf{B}^1_{{\underline{\boldsymbol{\psi}}}^{{\boldsymbol{\ell}},123}{\underline{\boldsymbol{\psi}}}^{{\boldsymbol{\ell}},13}}} &{\mathbf{B}^1_{{\underline{\boldsymbol{\psi}}}^{{\boldsymbol{\ell}},123}{\underline{\boldsymbol{\psi}}}^{{\boldsymbol{\ell}},12}}} &{\mathbf{B}^1_{{\underline{\boldsymbol{\psi}}}^{{\boldsymbol{\ell}},123}{\underline{\boldsymbol{\psi}}}^{{\boldsymbol{\ell}},1}}}\\
{\mathbf{B}^1_{{\underline{\boldsymbol{\psi}}}^{{\boldsymbol{\ell}},13}{\underline{\boldsymbol{\psi}}}^{{\boldsymbol{\ell}},123}}}& {\mathbf{B}^1_{{\underline{\boldsymbol{\psi}}}^{{\boldsymbol{\ell}},13}{\underline{\boldsymbol{\psi}}}^{{\boldsymbol{\ell}},13}}} &{\mathbf{B}^1_{{\underline{\boldsymbol{\psi}}}^{{\boldsymbol{\ell}},13}{\underline{\boldsymbol{\psi}}}^{{\boldsymbol{\ell}},12}}} &{\mathbf{B}^1_{{\underline{\boldsymbol{\psi}}}^{{\boldsymbol{\ell}},13}{\underline{\boldsymbol{\psi}}}^{{\boldsymbol{\ell}},1}}}\\
{\mathbf{B}^1_{{\underline{\boldsymbol{\psi}}}^{{\boldsymbol{\ell}},12}{\underline{\boldsymbol{\psi}}}^{{\boldsymbol{\ell}},123}}}& {\mathbf{B}^1_{{\underline{\boldsymbol{\psi}}}^{{\boldsymbol{\ell}},12}{\underline{\boldsymbol{\psi}}}^{{\boldsymbol{\ell}},13}}} &{\mathbf{B}^1_{{\underline{\boldsymbol{\psi}}}^{{\boldsymbol{\ell}},12}{\underline{\boldsymbol{\psi}}}^{{\boldsymbol{\ell}},12}}} &{\mathbf{B}^1_{{\underline{\boldsymbol{\psi}}}^{{\boldsymbol{\ell}},12}{\underline{\boldsymbol{\psi}}}^{{\boldsymbol{\ell}},1}}}\\
{\mathbf{B}^1_{{\underline{\boldsymbol{\psi}}}^{{\boldsymbol{\ell}},1}{\underline{\boldsymbol{\psi}}}^{{\boldsymbol{\ell}},123}}}& {\mathbf{B}^1_{{\underline{\boldsymbol{\psi}}}^{{\boldsymbol{\ell}},1}{\underline{\boldsymbol{\psi}}}^{{\boldsymbol{\ell}},13}}} &{\mathbf{B}^1_{{\underline{\boldsymbol{\psi}}}^{{\boldsymbol{\ell}},1}{\underline{\boldsymbol{\psi}}}^{{\boldsymbol{\ell}},12}}} &{\mathbf{B}^1_{{\underline{\boldsymbol{\psi}}}^{{\boldsymbol{\ell}},1}{\underline{\boldsymbol{\psi}}}^{{\boldsymbol{\ell}},1}}}\\
\end{array} } \right).
\end{equation*}
Let $\mathbf{B}^{k,0}_*$ denote the population version of $\mathbf{B}^{K}_*$. Since, we are restricted to $C_{1,n}\cap C_{2,n}$,\\ $\| \mathbf{B}^{K}_* - \mathbf{B}^{k,0}_* \| \leq c d_k \sqrt{\frac{\log p}{n}}$, hence 

\begin{equation*}
\begin{split}
\text{eig}_{min} \left( \mathbf{\Upsilon}\right)_{{\boldsymbol{\ell}} {\boldsymbol{\ell}}} &\geq \sum\limits_{k=1}^K \text{eig}_{\min} \left( \mathbf{B}_{*}^{k}\right)= \sum\limits_{k=1}^K\inf\limits_{|\boldsymbol{x}|=1} \boldsymbol{x}'\mathbf{B}^{k}_{*}\boldsymbol{x}\\
&\geq \sum\limits_{k=1}^K \left[ \inf\limits_{|\boldsymbol{x}|=1} \boldsymbol{x}'\mathbf{B}^{k,0}_{*}\boldsymbol{x} - \inf\limits_{|\boldsymbol{x}|=1} \boldsymbol{x}'\left( \mathbf{B}^{k}_{*} - \mathbf{B}^{k,0}_{*}\right)\boldsymbol{x} \right]\\
&\geq \sum\limits_{k=1}^K \left[ \inf\limits_{|\boldsymbol{x}|=1} \boldsymbol{x}'\mathbf{B}^{k,0}_{*}\boldsymbol{x} \right] - \sum\limits_{k=1}^K \| \mathbf{B}^{k}_{*} - \mathbf{B}^{k,0}_{*}\|_2 \\
&\geq \sum\limits_{k=1}^K \left[ \inf\limits_{|\boldsymbol{x}|=1} \boldsymbol{x}'\mathbf{B}^{k,0}_{*}\boldsymbol{x} \right] - K d_{\boldsymbol{\ell}} c \sqrt{\frac{\log p}{n}}
\end{split}
\end{equation*}
hence, by Lemma \ref{Lemma 3},
\begin{equation*}
\begin{split}
\text{eig}_{min} \left( \mathbf{\Upsilon}\right)_{{\boldsymbol{\ell}} {\boldsymbol{\ell}}} &\geq K\varepsilon_0 - K c d_{\boldsymbol{\ell}} \sqrt{\frac{\log p}{n}}\\
& \geq K\left( \varepsilon_0 - c\tau_{n} \sqrt{\frac{\log p}{n}}\right) = \frac{3K\varepsilon_0}{4}.
\end{split}
\end{equation*}
Similarly one can show that
\begin{equation*}
\text{eig}_{\max}\left( \mathbf{\Upsilon}_{{\boldsymbol{\ell}} {\boldsymbol{\ell}}} \right) \leq\frac{3K}{2\varepsilon_0}.
\end{equation*}
\end{proof}
By Lemma \ref{Lemma 3}, the value of the threshold $\tau_n$ which we used in building our hierarchical prior in (3.3) is given as $\tau_n = \frac{\varepsilon_0}{4c}\sqrt{\frac{n}{\log p}}$. Hence by Assumption 2, we can write $d_t \leq \tau_n$, for any sufficiently large $n$.

\begin{Lemma}\label{Lemma 4}
Let, $\mathbf{\Upsilon}$, and $\boldsymbol{a}$ be according to (\ref{upsilon}), (\ref{a}), and let $\mathbf{\Theta}^0$ be the true value of $\mathbf{\Theta}$ in (3.1). Then for large enough $n$, there exists a constant $c_0$ such that
\begin{equation}
\|\mathbf{\Upsilon} \mathbf{\Theta}^0 + \hat{\boldsymbol{a}} \|_{\max} \leq c_0 \sqrt{\frac{\log p}{n}}.
\end{equation}

\end{Lemma}
\begin{proof}
Note that by the triangular inequality,
\begin{equation}\label{Lemma 4 proof 1}
\|\mathbf{\Upsilon} \mathbf{\Theta}^0 + \hat{\boldsymbol{a}} \|_{\max} \leq \|\mathbf{\Upsilon} \mathbf{\Theta}^0 + \boldsymbol{a} \|_{\max} + \| \hat{\boldsymbol{a}} - \boldsymbol{a}\|_{\max},
\end{equation}
where, $\hat{\boldsymbol{a}}$ is the estimate of $\boldsymbol{a}$ provided by Assumption 1.

Now, in view of (\ref{upsilon}), (\ref{a}), and (3.1), one can easily check that
\begin{equation*}
\begin{split}
\mathbf{\Upsilon} \mathbf{\Theta}^0 + \boldsymbol{a} =& \left( {\begin{array}{*{20}{c}}
{\mathbf{B}^1 {\underline{\boldsymbol{\omega}}}^{1,0} + \boldsymbol{a}^1 + \mathbf{B}^2 {\underline{\boldsymbol{\omega}}}^{2,0} + \boldsymbol{a}^2 + \mathbf{B}^3 {\underline{\boldsymbol{\omega}}}^{3,0} + \boldsymbol{a}^3} \\ 
{\mathbf{B}^2 {\underline{\boldsymbol{\omega}}}^{2,0} + \boldsymbol{a}^2 + \mathbf{B}^3 {\underline{\boldsymbol{\omega}}}^{3,0} + \boldsymbol{a}^3 }\\ 
{\mathbf{B}^1{\underline{\boldsymbol{\omega}}}^{1,0} + \boldsymbol{a}^1 + \mathbf{B}^3 {\underline{\boldsymbol{\omega}}}^{3,0} + \boldsymbol{a}^3}\\ 
{\mathbf{B}^1 {\underline{\boldsymbol{\omega}}}^{1,0} + \boldsymbol{a}^1 + \mathbf{B}^2 {\underline{\boldsymbol{\omega}}}^{2,0} + \boldsymbol{a}^2} \\ 
{\mathbf{B}^3 {\underline{\boldsymbol{\omega}}}^{3,0} + \boldsymbol{a}^3} \\ 
{\mathbf{B}^2 {\underline{\boldsymbol{\omega}}}^{2,0} + \boldsymbol{a}^2} \\ 
{\mathbf{B}^1 {\underline{\boldsymbol{\omega}}}^{1,0} + \boldsymbol{a}^1} \\ 
\end{array} } \right),
\end{split}
\end{equation*}
where ${\underline{\boldsymbol{\omega}}}^{1,0} = \left( {\underline{\boldsymbol{\psi}}}^{1,0} + {\underline{\boldsymbol{\psi}}}^{12,0} + {\underline{\boldsymbol{\psi}}}^{13,0} + {\underline{\boldsymbol{\psi}}}^{123,0}\right)$, ${\underline{\boldsymbol{\omega}}}_{2,0} = \left( {\underline{\boldsymbol{\psi}}}^{2,0} + {\underline{\boldsymbol{\psi}}}^{12,0} + {\underline{\boldsymbol{\psi}}}^{23,0} + {\underline{\boldsymbol{\psi}}}^{123,0}\right)$, and ${\underline{\boldsymbol{\omega}}}^{3,0} = \left( {\underline{\boldsymbol{\psi}}}^{3,0} + {\underline{\boldsymbol{\psi}}}^{13,0} + {\underline{\boldsymbol{\psi}}}^{23,0} + {\underline{\boldsymbol{\psi}}}^{123,0}\right)$. Furthermore,
\begin{equation*}
\mathbf{B}^k {\underline{\boldsymbol{\omega}}}^{k,0} + \boldsymbol{a}^k = \left( {\begin{array}{*{20}{c}}
{{\mathbf{\Omega}^{k,0}_{:1}}' {\mathbf{S}}^k_{:2} + {\mathbf{\Omega}^{k,0}_{:2}}' {\mathbf{S}}^k_{:1}}\\
{{\mathbf{\Omega}^{k,0}_{:1}}' {\mathbf{S}}^k_{:3} + {\mathbf{\Omega}^{k,0}_{:3}}' {\mathbf{S}}^k_{:1}}\\
{\vdots}\\
{{\mathbf{\Omega}^{k,0}_{:p-1}}' {\mathbf{S}}^k_{:p} + {\mathbf{\Omega}^{k,0}_{:p}}' {\mathbf{S}}^k_{:p-1}}
\end{array} } \right), \quad \quad k=1,2,3.
\end{equation*}
Now, we rewrite $\mathbf{\Upsilon} \mathbf{\Theta}^0 + \boldsymbol{a}$ as,
\begin{equation*}
\mathbf{\Upsilon} \mathbf{\Theta}^0 + \boldsymbol{a} = \left( {\begin{array}{*{20}{c}} 
{\mathbf{I}_{\frac{p(p-1)}{2}}} & {\mathbf{I}_{\frac{p(p-1)}{2}}} & {\mathbf{I}_{\frac{p(p-1)}{2}}} \\
{\mathbf{I}_{\frac{p(p-1)}{2}}} & {\mathbf{I}_{\frac{p(p-1)}{2}}} & {} \\
{\mathbf{I}_{\frac{p(p-1)}{2}}} & {\mathbf{0}} & {\mathbf{I}_{\frac{p(p-1)}{2}}} \\
{\mathbf{0}} & {\mathbf{I}_{\frac{p(p-1)}{2}}} & {\mathbf{I}_{\frac{p(p-1)}{2}}} \\
{\mathbf{I}_{\frac{p(p-1)}{2}}} & {\mathbf{0}} & {\mathbf{0}} \\
{\mathbf{0}} & {\mathbf{I}_{\frac{p(p-1)}{2}}} & {\mathbf{0}} \\
{\mathbf{0}} & {\mathbf{0}} & {\mathbf{I}_{\frac{p(p-1)}{2}}} \\
\end{array} } \right) \left( \begin{array}{*{20}{c}} {\mathbf{B}^3 {\underline{\boldsymbol{\omega}}}^{3,0} + \boldsymbol{a}^3}\\
{\mathbf{B}^2 {\underline{\boldsymbol{\omega}}}^{2,0} + \boldsymbol{a}^2}\\
{\mathbf{B}^1 {\underline{\boldsymbol{\omega}}}^{1,0} + \boldsymbol{a}^1} \end{array}\right).
\end{equation*}
The norm of the matrix in the right hand side of the above equation is equal to $\sqrt{K(2^K-1)}$, hence, by restricting to the event $C_{1,n} \cap C_{2,n}$, we have that

\begin{equation}\label{Lemma 4 proof 2}
\begin{split}
\| \mathbf{\Upsilon} \mathbf{\Theta}^0 + \boldsymbol{a} \|_{\max} &\le \sqrt{\sum\limits_{k=1}^{K} \| \mathbf{B}^k {\underline{\boldsymbol{\omega}}}^{k,0} + \boldsymbol{a}^k\|_{\max}^2}\\
& \le \sqrt{K \mathop{\max_{1 \le k\le K}}_{1 \le i < j \le p} \left( {\mathbf{\Omega}^{k,0}_{:i}}' {\mathbf{S}}^k_{:j}\right)^2}\\ 
& \le \sqrt{K}\mathop{\max_{1 \le k\le K}}_{1 \le i < j \le p} | {\mathbf{\Omega}^{k,0}_{:i}}' \left( {\mathbf{S}}^k_{:j} - \mathbf{\Sigma}^k_{:j}\right) | \\
& \le 2c \sqrt{K} \sqrt{\frac{\log p}{n}}.
\end{split}
\end{equation}
Moreover, by (\ref{a}), and Assumption 1, it is easy to see that
\begin{equation*}
\begin{split}
\| \hat{\boldsymbol{a}} - \boldsymbol{a}\|_{\max} \leq & 2 K C \| \mathbf{\Upsilon}\|_{\max} \sqrt{\frac{\log p}{n}}
\end{split}
\end{equation*}
also, $\|\mathbf{\Upsilon}\|_{\max} \leq \max\limits_{\boldsymbol{\ell}\in \mathcal{L}} \|\mathbf{\Upsilon}_{\boldsymbol{\ell}\boldsymbol{\ell}}\|_{\max}$ hence, by applying Lemma \ref{Lemma 3}, we have that
\begin{equation}\label{Lemma 4 proof 3}
\begin{split}
\| \hat{\boldsymbol{a}} - \boldsymbol{a}\|_{\max} & \leq  \frac{3CK^2}{\varepsilon_0}\sqrt{\frac{\log p}{n}}
\end{split}
\end{equation}
Therefore, by combining (\ref{Lemma 4 proof 1}), (\ref{Lemma 4 proof 2}), and (\ref{Lemma 4 proof 3}), 

\begin{equation*}
\|\mathbf{\Upsilon} \mathbf{\Theta}^0 + \hat{\boldsymbol{a}} \|_{\max} \leq \left( 2c \sqrt{K} + \frac{3CK^2}{\varepsilon_0}\right) \sqrt{\frac{\log p}{n}},
\end{equation*}

thus, the conclusion follows by letting $c_0 = 2c \sqrt{K} + \frac{3CK^2}{\varepsilon_0}$.
\end{proof}

For simplicity in writing, we denote the ratio of the posterior probabilities of any sparsity pattern/model $\boldsymbol{{\boldsymbol{\ell}}}$ and the true sparsity pattern/model $\boldsymbol{t}$,  by $PR\left( \boldsymbol{\ell}, \boldsymbol{t}\right)$, i.e.
\begin{equation}
PR\left( \boldsymbol{\ell}, \boldsymbol{t}\right) = \frac{P \left\{ \boldsymbol{{\boldsymbol{\ell}}} | \hat{\mathbf{\Delta}}, \mathcal{Y}\right\}}{P \left\{ \boldsymbol{t} | \hat{\mathbf{\Delta}}, \mathcal{Y}\right\}}, \quad \quad \text{for any sparsity pattern}\quad  \boldsymbol{\ell} \neq \boldsymbol{t}.
\end{equation}

\begin{Lemma} \label{Lemma 1 PR}
Under Assumption 1, the ratio of the posterior probabilities of any sparsity pattern/model $\boldsymbol{{\boldsymbol{\ell}}}$ and the true sparsity pattern/model $\boldsymbol{t}$ satisfies:

\begin{equation}\label{PR}
\begin{split}
PR(\boldsymbol{\ell}, \boldsymbol{t}) &= \frac{P \left\{ \boldsymbol{{\boldsymbol{\ell}}} | \hat{\mathbf{\Delta}}, \mathcal{Y} \right\}}{P \left\{ \boldsymbol{t} | \hat{\mathbf{\Delta}}, \mathcal{Y} \right\}} = \frac{\left[ q_1^{d_{\boldsymbol{{\boldsymbol{\ell}}}}}(1-q_1)^{\binom{p}{2} - d_{\boldsymbol{{\boldsymbol{\ell}}}}} I_{\left\{d_{\boldsymbol{\ell}}\leq \tau\right\}} + q_1^{d_{\boldsymbol{\ell}}}(1-q_1)^{\binom{p}{2} - d_{\boldsymbol{\ell}}} I_{\left\{d_{\boldsymbol{\ell}}\leq \tau\right\}}\right]}{\left[ q_1^{d_{\boldsymbol{t}}}(1-q_1)^{\binom{p}{2} - d_{\boldsymbol{t}}} I_{\left\{d_{\boldsymbol{t}}\leq \tau\right\}} + q_1^{d_{\boldsymbol{t}}}(1-q_1)^{\binom{p}{2} - d_{\boldsymbol{t}}} I_{\left\{d_{\boldsymbol{t}}\leq \tau\right\}}\right]}\\
&\times \frac{|\mathbf{\Lambda}_{\boldsymbol{\ell}\boldsymbol{\ell}}|^{\frac{1}{2}}}{|\mathbf{\Lambda}_{\boldsymbol{\ell}\boldsymbol{t}}|^{\frac{1}{2}}} \frac{| \left( n\mathbf{\Upsilon} + \mathbf{\Lambda}\right)_{tt}|^{\frac{1}{2}}}{| \left( n\mathbf{\Upsilon} + \mathbf{\Lambda}\right)_{{\boldsymbol{\ell}}{\boldsymbol{\ell}}}|^{\frac{1}{2}}} \frac{\exp\left\{ \frac{n^2}{2}\hat{\boldsymbol{a}}_{\boldsymbol{\ell}}\left( n\mathbf{\Upsilon} + \mathbf{\Lambda}\right)_{{\boldsymbol{\ell}}{\boldsymbol{\ell}}}^{-1}\hat{\boldsymbol{a}}_{\boldsymbol{\ell}}\right\}}{\exp\left\{ \frac{n^2}{2}\hat{\boldsymbol{a}}_t\left( n\mathbf{\Upsilon} + \mathbf{\Lambda}\right)_{tt}^{-1}\hat{\boldsymbol{a}}_t\right\}}.
\end{split}
\end{equation}
\end{Lemma}
\begin{proof} 
We note that
\begin{equation*}
\begin{split}
P& \left\{ \boldsymbol{\ell} | \hat{\mathbf{\Delta}}, \mathcal{Y}\right\} = P \left\{ \mathbf{\Theta}\in \mathcal{M}_{\boldsymbol{\ell}}| \hat{\mathbf{\Delta}}, \mathcal{Y} \right\} = \int_{\mathcal{M}_{\boldsymbol{\ell}}} \pi \left(\mathbf{\Theta} | \hat{\mathbf{\Delta}}, \mathcal{Y}\right) d\mathbf{\Theta}, 
\end{split}
\end{equation*}
hence, in view of (4.4),
\begin{equation*}
\begin{split}
P& \left\{ \boldsymbol{\ell} | \hat{\mathbf{\Delta}}, \mathcal{Y} \right\} =C_0 \left[ q_1^{d_{\boldsymbol{\ell}}}(1-q_1)^{\binom{p}{2} - d_{\boldsymbol{\ell}}} I_{\left\{d_{\boldsymbol{\ell}}\leq \tau\right\}} + q_1^{d_{\boldsymbol{\ell}}}(1-q_1)^{\binom{p}{2} - d_{\boldsymbol{\ell}}} I_{\left\{d_{\boldsymbol{\ell}}\leq \tau\right\}}\right]\\
& \frac{|\mathbf{\Lambda}_{\boldsymbol{\ell}\boldsymbol{\ell}}|^{\frac{1}{2}}}{| \left( n\mathbf{\Upsilon} + \mathbf{\Lambda}\right)_{{\boldsymbol{\ell}}{\boldsymbol{\ell}}}|^{\frac{1}{2}}} \exp\left\{ \frac{n^2}{2}\hat{\boldsymbol{a}}_{\boldsymbol{\ell}}\left( n\mathbf{\Upsilon} + \mathbf{\Lambda}\right)_{{\boldsymbol{\ell}}{\boldsymbol{\ell}}}^{-1}\hat{\boldsymbol{a}}_{\boldsymbol{\ell}}\right\},\\
\end{split}
\end{equation*}
where the last equality is achieved using the properties of the multivariate normal distribution. 
\end{proof}
In the next series of lemmas, we will show that for any sparsity pattern $\boldsymbol{\ell} \in \mathcal{L}$, the posterior probability ratio $PR(\boldsymbol{\ell}, \boldsymbol{t})$ is approaching zero, as $n$ goes to $\infty$. Specifically, we consider four cases of underfitted (${\boldsymbol{\ell}} \subset \boldsymbol{t}$), overfitted (${\boldsymbol{t}} \subset \boldsymbol{\ell}$ with $d_{\boldsymbol{\ell}} < \tau_n$), unrealistically overfitted (${\boldsymbol{t}} \subset \boldsymbol{\ell}$ with $d_{\boldsymbol{\ell}} > \tau_n$), and non-inclusive ( ${\boldsymbol{t}} \not \subseteq \boldsymbol{\ell}$ and  ${\boldsymbol{\ell}} \not \subseteq \boldsymbol{t}$ ).

\begin{Lemma} \label{Lemma 5}
Suppose ${\boldsymbol{\ell}} \subset \boldsymbol{t}$ then, under Assumptions 1-6,
\begin{equation}
PR({\boldsymbol{\ell}}, \boldsymbol{t}) \to 0, \quad \quad \quad \text{as} \quad n \to \infty . 
\end{equation}
\end{Lemma}

\begin{proof}
By Assumption 2, $d_{\boldsymbol{t}} < \tau$, hence $d_{\boldsymbol{\ell}} < d_{\boldsymbol{t}} < \tau$. Now,

\begin{equation*}
\begin{split}
PR\left({\boldsymbol{\ell}}, \boldsymbol{t}\right) =& \frac{\| \mathbf{\Lambda}_{{\boldsymbol{\ell}} {\boldsymbol{\ell}}} \|^{\frac{1}{2}}}{\| \mathbf{\Lambda}_{\boldsymbol{t}\boldsymbol{t}} \|^{\frac{1}{2}}} \left( \frac{q_1}{1-q_1}\right)^{d_{\boldsymbol{\ell}} - d_{\boldsymbol{t}}} \frac{\|\left( n\mathbf{\Upsilon} + \mathbf{\Lambda}\right)_{\boldsymbol{t}\boldsymbol{t}} \|^{\frac{1}{2}}}{\|\left( n\mathbf{\Upsilon} + \mathbf{\Lambda}\right)_{{\boldsymbol{\ell}} {\boldsymbol{\ell}}} \|^{\frac{1}{2}}} \frac{\exp\left\{ \frac{n^2}{2} {\hat{\boldsymbol{a}}_{{\boldsymbol{\ell}}}}' \left( n\mathbf{\Upsilon} + \mathbf{\Lambda}\right)^{-1}_{{\boldsymbol{\ell}} {\boldsymbol{\ell}}} \hat{\boldsymbol{a}}_{\boldsymbol{\ell}} \right\}}{\exp\left\{ \frac{n^2}{2} {\hat{\boldsymbol{a}}_{\boldsymbol{t}}}' \left( n\mathbf{\Upsilon} + \mathbf{\Lambda}\right)^{-1}_{\boldsymbol{t}\boldsymbol{t}} \hat{\boldsymbol{a}}_{\boldsymbol{t}}\right\}}\\
= & \frac{\| \mathbf{\Lambda}_{{\boldsymbol{\ell}} {\boldsymbol{\ell}}} \|^{\frac{1}{2}}}{\| \mathbf{\Lambda}_{\boldsymbol{t}\boldsymbol{t}} \|^{\frac{1}{2}}} \left( \frac{q_1}{1-q_1}\right)^{d_{\boldsymbol{\ell}} - d_{\boldsymbol{t}}} \frac{\|\left( n\mathbf{\Upsilon} + \mathbf{\Lambda}\right)_{\boldsymbol{t}\boldsymbol{t}} \|^{\frac{1}{2}}}{\|\left( n\mathbf{\Upsilon} + \mathbf{\Lambda}\right)_{{\boldsymbol{\ell}} {\boldsymbol{\ell}}} \|^{\frac{1}{2}}}\\
&\exp\left\{- \frac{n^2}{2}\left[\hat{\boldsymbol{a}}_{{\boldsymbol{\ell}}^c} - n \mathbf{\Upsilon}_{{\boldsymbol{\ell}}^c{\boldsymbol{\ell}}}\left( n\mathbf{\Upsilon} + \mathbf{\Lambda}\right)_{{\boldsymbol{\ell}}{\boldsymbol{\ell}}}^{-1}\hat{\boldsymbol{a}}_{\boldsymbol{\ell}}\right]'\left( n\mathbf{\Upsilon} + \mathbf{\Lambda}\right)_{\boldsymbol{t}|{\boldsymbol{\ell}}}^{-1} \left[\hat{\boldsymbol{a}}_{{\boldsymbol{\ell}}^c} - n \mathbf{\Upsilon}_{{\boldsymbol{\ell}}^c{\boldsymbol{\ell}}}\left( n\mathbf{\Upsilon} + \mathbf{\Lambda}\right)_{{\boldsymbol{\ell}}{\boldsymbol{\ell}}}^{-1}\hat{\boldsymbol{a}}_{\boldsymbol{\ell}}\right]\right\},
\end{split}
\end{equation*}
that is,
\begin{equation*}
\begin{split}
PR\left({\boldsymbol{\ell}}, \boldsymbol{t}\right) \leq& \frac{\| \mathbf{\Lambda}_{{\boldsymbol{\ell}} {\boldsymbol{\ell}}} \|^{\frac{1}{2}}}{\| \mathbf{\Lambda}_{\boldsymbol{t}\boldsymbol{t}} \|^{\frac{1}{2}}} \left( \frac{q_1}{1-q_1}\right)^{d_{\boldsymbol{\ell}} - d_{\boldsymbol{t}}} \frac{\|\left( n\mathbf{\Upsilon} + \mathbf{\Lambda}\right)_{\boldsymbol{t}\boldsymbol{t}} \|^{\frac{1}{2}}}{\|\left( n\mathbf{\Upsilon} + \mathbf{\Lambda}\right)_{{\boldsymbol{\ell}} {\boldsymbol{\ell}}} \|^{\frac{1}{2}}} \exp \left\{ - \frac{n^2 \| \hat{\boldsymbol{a}}_{{\boldsymbol{\ell}}^c} - n \mathbf{\Upsilon}_{{\boldsymbol{\ell}}^c{\boldsymbol{\ell}}}\left( n\mathbf{\Upsilon} + \mathbf{\Lambda}\right)_{{\boldsymbol{\ell}}{\boldsymbol{\ell}}}^{-1}\hat{\boldsymbol{a}}_{\boldsymbol{\ell}}\|^2}{2 \text{eig}_{max}\left(n\mathbf{\Upsilon} + \mathbf{\Lambda} \right)_{\boldsymbol{t}\boldsymbol{t}}}\right\},
\end{split}
\end{equation*}

Now, by the triangular inequality,
\begin{equation}\label{Lemma 5.1}
\begin{split}
\| \hat{\boldsymbol{a}}_{{\boldsymbol{\ell}}^c} - n \mathbf{\Upsilon}_{{\boldsymbol{\ell}}^c{\boldsymbol{\ell}}}&\left( n\mathbf{\Upsilon} + \mathbf{\Lambda}\right)_{{\boldsymbol{\ell}}{\boldsymbol{\ell}}}^{-1}\hat{\boldsymbol{a}}_{\boldsymbol{\ell}}\| \\
\geq &\| \boldsymbol{a}_{{\boldsymbol{\ell}}^c} - n \mathbf{\Upsilon}_{{\boldsymbol{\ell}}^c{\boldsymbol{\ell}}}\left( n\mathbf{\Upsilon} + \mathbf{\Lambda}\right)_{{\boldsymbol{\ell}}{\boldsymbol{\ell}}}^{-1}\boldsymbol{a}_{\boldsymbol{\ell}}\|\\
&-\|\left( \hat{\boldsymbol{a}}_{{\boldsymbol{\ell}}^c} - \boldsymbol{a}_{{\boldsymbol{\ell}}^c}\right) - n \mathbf{\Upsilon}_{{\boldsymbol{\ell}}^c{\boldsymbol{\ell}}}\left( n\mathbf{\Upsilon} + \mathbf{\Lambda}\right)_{{\boldsymbol{\ell}}{\boldsymbol{\ell}}}^{-1}\left( \hat{\boldsymbol{a}}_{{\boldsymbol{\ell}}} - \boldsymbol{a}_{{\boldsymbol{\ell}}} \right) \|\\
=&\| \left( \pm \mathbf{\Upsilon} \mathbf{\Theta}^0 + \boldsymbol{a}\right)_{{\boldsymbol{\ell}}^c} - n \mathbf{\Upsilon}_{{\boldsymbol{\ell}}^c{\boldsymbol{\ell}}}\left( n\mathbf{\Upsilon} + \mathbf{\Lambda}\right)_{{\boldsymbol{\ell}}{\boldsymbol{\ell}}}^{-1}\left( \pm \mathbf{\Upsilon} \mathbf{\Theta}^0 + \boldsymbol{a} \right)_{{\boldsymbol{\ell}}}\| \\
&-\|\left( \hat{\boldsymbol{a}}_{{\boldsymbol{\ell}}^c} - \boldsymbol{a}_{{\boldsymbol{\ell}}^c}\right) - n \mathbf{\Upsilon}_{{\boldsymbol{\ell}}^c{\boldsymbol{\ell}}}\left( n\mathbf{\Upsilon} + \mathbf{\Lambda}\right)_{{\boldsymbol{\ell}}{\boldsymbol{\ell}}}^{-1}\left( \hat{\boldsymbol{a}}_{{\boldsymbol{\ell}}} - \boldsymbol{a}_{{\boldsymbol{\ell}}} \right) \|\\
\ge & \| \left( \mathbf{\Upsilon} \mathbf{\Theta}^0\right)_{{\boldsymbol{\ell}}^c} - n \mathbf{\Upsilon}_{{\boldsymbol{\ell}}^c{\boldsymbol{\ell}}}\left( n\mathbf{\Upsilon} + \mathbf{\Lambda}\right)_{{\boldsymbol{\ell}}{\boldsymbol{\ell}}}^{-1}\left( \mathbf{\Upsilon} \mathbf{\Theta}^0 \right)_{{\boldsymbol{\ell}}}\| \\
&- \| \left( \mathbf{\Upsilon} \mathbf{\Theta}^0 + \boldsymbol{a}\right)_{{\boldsymbol{\ell}}^c} - n \mathbf{\Upsilon}_{{\boldsymbol{\ell}}^c{\boldsymbol{\ell}}}\left( n\mathbf{\Upsilon} + \mathbf{\Lambda}\right)_{{\boldsymbol{\ell}}{\boldsymbol{\ell}}}^{-1}\left( \mathbf{\Upsilon} \mathbf{\Theta}^0 + \boldsymbol{a}\right)_{{\boldsymbol{\ell}}}\| \\
&-\|\left( \hat{\boldsymbol{a}}_{{\boldsymbol{\ell}}^c} - \boldsymbol{a}_{{\boldsymbol{\ell}}^c}\right) - n \mathbf{\Upsilon}_{{\boldsymbol{\ell}}^c{\boldsymbol{\ell}}}\left( n\mathbf{\Upsilon} + \mathbf{\Lambda}\right)_{{\boldsymbol{\ell}}{\boldsymbol{\ell}}}^{-1}\left( \hat{\boldsymbol{a}}_{{\boldsymbol{\ell}}} - \boldsymbol{a}_{{\boldsymbol{\ell}}} \right) \|.
\end{split}
\end{equation}
Now, by appropriately partitioning $\mathbf{\Upsilon}$, we can write $\left( \mathbf{\Upsilon} \mathbf{\Theta}^0\right)_{{\boldsymbol{\ell}}^c} = \mathbf{\Upsilon}_{{\boldsymbol{\ell}}^c{\boldsymbol{\ell}}} \mathbf{\Theta}^0_{{\boldsymbol{\ell}}} + \mathbf{\Upsilon}_{{\boldsymbol{\ell}}^c{\boldsymbol{\ell}}^c} \mathbf{\Theta}^0_{{\boldsymbol{\ell}}^c}$ and $\left( \mathbf{\Upsilon} \mathbf{\Theta}^0\right)_{{\boldsymbol{\ell}}} = \mathbf{\Upsilon}_{{\boldsymbol{\ell}}{\boldsymbol{\ell}}} \mathbf{\Theta}^0_{{\boldsymbol{\ell}}} + \mathbf{\Upsilon}_{{\boldsymbol{\ell}}{\boldsymbol{\ell}}^c} \mathbf{\Theta}^0_{{\boldsymbol{\ell}}^c}$. Hence, for large enough $n$,

\begin{equation}\label{Lemma 5.2}
\begin{split}
\| \left(\mathbf{\Upsilon} \mathbf{\Theta}^0\right)_{{\boldsymbol{\ell}}^c} - n \mathbf{\Upsilon}_{{\boldsymbol{\ell}}^c{\boldsymbol{\ell}}}\left( n\mathbf{\Upsilon} + \mathbf{\Lambda}\right)_{{\boldsymbol{\ell}}{\boldsymbol{\ell}}}^{-1}\left( \mathbf{\Upsilon} \mathbf{\Theta}^0 \right)_{{\boldsymbol{\ell}}}\| &= \| \frac{1}{n}\left( n\mathbf{\Upsilon} + \mathbf{\Lambda}\right)_{\boldsymbol{t}|{\boldsymbol{\ell}}}\mathbf{\Theta}_{{\boldsymbol{\ell}}^c}^0 - \mathbf{\Upsilon}_{{\boldsymbol{\ell}}^c{\boldsymbol{\ell}}}\left( n\mathbf{\Upsilon} + \mathbf{\Lambda}\right)_{{\boldsymbol{\ell}}{\boldsymbol{\ell}}}^{-1}\mathbf{\Lambda}_{{\boldsymbol{\ell}}{\boldsymbol{\ell}}}\mathbf{\Theta}_{{\boldsymbol{\ell}}}^0\| \\
&\geq \| \frac{1}{n}\left( n\mathbf{\Upsilon} + \mathbf{\Lambda}\right)_{\boldsymbol{t}|{\boldsymbol{\ell}}}\mathbf{\Theta}_{{\boldsymbol{\ell}}^c}^0 \| - \| \mathbf{\Upsilon}_{{\boldsymbol{\ell}}^c{\boldsymbol{\ell}}}\left( n\mathbf{\Upsilon} + \mathbf{\Lambda}\right)_{{\boldsymbol{\ell}}{\boldsymbol{\ell}}}^{-1}\mathbf{\Lambda}_{{\boldsymbol{\ell}}{\boldsymbol{\ell}}}\mathbf{\Theta}_{{\boldsymbol{\ell}}}^0\| \\
&\geq \| \frac{1}{n}\left( n\mathbf{\Upsilon} + \mathbf{\Lambda}\right)_{\boldsymbol{t}|{\boldsymbol{\ell}}}\mathbf{\Theta}_{{\boldsymbol{\ell}}^c}^0 \| - \frac{\text{eig}_{\min}\left( \mathbf{\Upsilon}_{{\boldsymbol{\ell}}^c{\boldsymbol{\ell}}}\right) \| \mathbf{\Lambda}_{{\boldsymbol{\ell}}{\boldsymbol{\ell}}}\mathbf{\Theta}_{{\boldsymbol{\ell}}}^0\|}{\text{eig}_{\min}\left( n\mathbf{\Upsilon} + \mathbf{\Lambda}\right)_{{\boldsymbol{\ell}} {\boldsymbol{\ell}}}} \\
&\geq \| \frac{1}{n}\left( n\mathbf{\Upsilon} + \mathbf{\Lambda}\right)_{\boldsymbol{t}|{\boldsymbol{\ell}}}\mathbf{\Theta}_{{\boldsymbol{\ell}}^c}^0 \| - \frac{2 \| \mathbf{\Lambda}_{{\boldsymbol{\ell}}{\boldsymbol{\ell}}}\mathbf{\Theta}_{{\boldsymbol{\ell}}}^0\|}{n\varepsilon_0^2} \\
& \geq \frac{1}{2} \| \frac{1}{n}\left( n\mathbf{\Upsilon} + \mathbf{\Lambda}\right)_{\boldsymbol{t}|{\boldsymbol{\ell}}}\mathbf{\Theta}_{{\boldsymbol{\ell}}^c}^0 \| \\
&\ge \frac{1}{2} \frac{1}{n} \text{eig}_{\text{min}}\left( n\mathbf{\Upsilon} + \mathbf{\Lambda}\right)_{\boldsymbol{t}\boldsymbol{t}}s_n\sqrt{(d_{\boldsymbol{t}} - d_{\boldsymbol{\ell}})} \\
&\ge \frac{1}{2} \frac{1}{n} n\text{eig}_{\text{min}}\left(\mathbf{\Upsilon} \right)_{tt}s_n\sqrt{(d_{\boldsymbol{t}} - d_{\boldsymbol{\ell}})} \\
&\ge \frac{1}{2} \varepsilon_0 s_n \sqrt{(d_{\boldsymbol{t}} - d_{\boldsymbol{\ell}})}
\end{split}
\end{equation} 
Moving onto the second term in the right hand side of (\ref{Lemma 5.1}),
\begin{equation}\label{Lemma 5.3}
\begin{split}
& \| \left( \mathbf{\Upsilon} \mathbf{\Theta}^0 + \boldsymbol{a}\right)_{{\boldsymbol{\ell}}^c} - n \mathbf{\Upsilon}_{{\boldsymbol{\ell}}^c{\boldsymbol{\ell}}}\left( n\mathbf{\Upsilon} + \mathbf{\Lambda}\right)_{{\boldsymbol{\ell}}{\boldsymbol{\ell}}}^{-1}\left( \mathbf{\Upsilon} \mathbf{\Theta}^0 + \boldsymbol{a}\right)_{{\boldsymbol{\ell}}}\| \\
& \le \| \left( \mathbf{\Upsilon} \mathbf{\Theta}^0 + \boldsymbol{a}\right)_{{\boldsymbol{\ell}}^c} \| + \| n \mathbf{\Upsilon}_{{\boldsymbol{\ell}}^c{\boldsymbol{\ell}}}\left( n\mathbf{\Upsilon} + \mathbf{\Lambda}\right)_{{\boldsymbol{\ell}}{\boldsymbol{\ell}}}^{-1}\left( \mathbf{\Upsilon} \mathbf{\Theta}^0 + \boldsymbol{a}\right)_{{\boldsymbol{\ell}}}\|\\
& \le \| \left( \mathbf{\Upsilon} \mathbf{\Theta}^0 + \boldsymbol{a}\right)_{{\boldsymbol{\ell}}^c} \| + \frac{n\text{eig}_{\max} \left(\mathbf{\Upsilon}_{\boldsymbol{\ell}^c \boldsymbol{\ell}}\right) \|\left( \mathbf{\Upsilon} \mathbf{\Theta}^0 + \boldsymbol{a}\right)_{{\boldsymbol{\ell}}}\| }{\text{eig}_{\min} \left( n\mathbf{\Upsilon} + \mathbf{\Lambda}\right)_{{\boldsymbol{\ell}} {\boldsymbol{\ell}}} }\\
& \le \| \left( \mathbf{\Upsilon} \mathbf{\Theta}^0 + \boldsymbol{a}\right)_{{\boldsymbol{\ell}}^c} \| + \frac{2\|\left( \mathbf{\Upsilon} \mathbf{\Theta}^0 + \boldsymbol{a}\right)_{{\boldsymbol{\ell}}}\|}{\varepsilon_0^2} \\
& \leq c_0\sqrt{\frac{\log p}{n}} \left(\sqrt{d_{\boldsymbol{t}} - d_{\boldsymbol{\ell}}} + \frac{2\sqrt{d_{\boldsymbol{\ell}}}}{\varepsilon_0^2} \right),
\end{split}
\end{equation}
where the last equality was achieved by Lemma \ref{Lemma 4}. 
Further, regarding the third term in right hand side of (\ref{Lemma 5.1}) we can write 
\begin{equation}\label{Lemma 5.4}
\begin{split}
\|\left( \hat{\boldsymbol{a}}_{{\boldsymbol{\ell}}^c} - \boldsymbol{a}_{{\boldsymbol{\ell}}^c}\right) - n \mathbf{\Upsilon}_{{\boldsymbol{\ell}}^c{\boldsymbol{\ell}}}\left( n\mathbf{\Upsilon} + \mathbf{\Lambda}\right)_{{\boldsymbol{\ell}}{\boldsymbol{\ell}}}^{-1}\left( \hat{\boldsymbol{a}}_{{\boldsymbol{\ell}}} - \boldsymbol{a}_{{\boldsymbol{\ell}}} \right) \| & \leq \| \hat{\boldsymbol{a}}_{{\boldsymbol{\ell}}^c} - \boldsymbol{a}_{{\boldsymbol{\ell}}^c} \| + \frac{n\text{eig}_{\max} \left(\mathbf{\Upsilon}_{\boldsymbol{\ell}^c \boldsymbol{\ell}}\right) \|\left( \hat{\boldsymbol{a}}_{{\boldsymbol{\ell}}} - \boldsymbol{a}_{{\boldsymbol{\ell}}} \right) \| }{\text{eig}_{\min} \left( n\mathbf{\Upsilon} + \mathbf{\Lambda}\right)_{{\boldsymbol{\ell}} {\boldsymbol{\ell}}} }\\
& \leq \frac{3CK^2}{\varepsilon_0} \sqrt{\frac{\log p}{n}}  \left(\sqrt{d_{\boldsymbol{t}} - d_{\boldsymbol{\ell}}} + \frac{2\sqrt{d_{\boldsymbol{\ell}}}}{\varepsilon_0^2} \right),
\end{split}
\end{equation}
hence, by combining (\ref{Lemma 5.1}), (\ref{Lemma 5.2}), (\ref{Lemma 5.3}), and (\ref{Lemma 5.4}), for sufficiently large $n$, we have that

\begin{equation*}
\begin{split}
\| \hat{\boldsymbol{a}}_{{\boldsymbol{\ell}}^c} - n \mathbf{\Upsilon}_{{\boldsymbol{\ell}}^c{\boldsymbol{\ell}}}\left( n\mathbf{\Upsilon} + \mathbf{\Lambda}\right)_{{\boldsymbol{\ell}}{\boldsymbol{\ell}}}^{-1}\hat{\boldsymbol{a}}_{\boldsymbol{\ell}}\| \geq & \frac{1}{2} \varepsilon_0 s_n \sqrt{(d_{\boldsymbol{t}} - d_{\boldsymbol{\ell}})} \\
&- c_0\sqrt{\frac{\log p}{n}}  \left(\sqrt{d_{\boldsymbol{t}} - d_{\boldsymbol{\ell}}} + \frac{2\sqrt{d_{\boldsymbol{\ell}}}}{\varepsilon_0^2} \right) \\
&- \frac{3CK^2}{\varepsilon_0} \sqrt{\frac{\log p}{n}}  \left(\sqrt{d_{\boldsymbol{t}} - d_{\boldsymbol{\ell}}} + \frac{2\sqrt{d_{\boldsymbol{\ell}}}}{\varepsilon_0^2} \right)\\
\geq & \frac{1}{2} \varepsilon_0 s_n \sqrt{(d_{\boldsymbol{t}} - d_{\boldsymbol{\ell}})} \\
& - \left(c_0 +  \frac{3CK^2}{\varepsilon_0} \right) \sqrt{\frac{\log p}{n}}  \left(\sqrt{d_{\boldsymbol{t}} - d_{\boldsymbol{\ell}}} + \frac{2\sqrt{d_{\boldsymbol{\ell}}}}{\varepsilon_0^2} \right)\\
\geq & \frac{1}{2} \varepsilon_0 s_n - \left(c_0 +  \frac{3CK^2}{\varepsilon_0} \right) \sqrt{\frac{\log p}{n}}  \left(\frac{2\sqrt{d_{\boldsymbol{t}}}}{\varepsilon_0^2} \right), \\
\end{split}
\end{equation*}
in view of Assumption 5, $\frac{ \frac{1}{2} \varepsilon_0 s_n }{ \left(c_0 +  \frac{3CK^2}{\varepsilon_0} \right) \sqrt{\frac{\log p}{n}}  \left(\frac{2\sqrt{d_{\boldsymbol{t}}}}{\varepsilon_0^2} \right)} \to \infty$, as $n \to \infty$, hence, for all large $n$, we can write,
\begin{equation*}
\| \hat{\boldsymbol{a}}_{{\boldsymbol{\ell}}^c} - n \mathbf{\Upsilon}_{{\boldsymbol{\ell}}^c{\boldsymbol{\ell}}}\left( n\mathbf{\Upsilon} + \mathbf{\Lambda}\right)_{{\boldsymbol{\ell}}{\boldsymbol{\ell}}}^{-1}\hat{\boldsymbol{a}}_{\boldsymbol{\ell}}\| \geq \frac{1}{4} \varepsilon_0 s_n 
\end{equation*}
Now, once again by Lemma \ref{Lemma 4}
\begin{equation*}
\begin{split}
PR({\boldsymbol{\ell}}, \boldsymbol{t}) & \leq \frac{\| \mathbf{\Lambda}_{{\boldsymbol{\ell}} {\boldsymbol{\ell}}}\|^{\frac{1}{2}}}{\| \mathbf{\Lambda}_{\boldsymbol{t}\boldsymbol{t}}\|^{\frac{1}{2}}} (2q_1)^{d_{\boldsymbol{\ell}} - d_{\boldsymbol{t}}} n^{\frac{d_{\boldsymbol{t}} - d_{\boldsymbol{\ell}} }{2}} \exp \left\{-\frac{n^2 \frac{1}{64} \varepsilon_0^2s_n^2}{6Kn\varepsilon_0^{-1}} \right\} \\
& = \frac{\| \mathbf{\Lambda}_{{\boldsymbol{\ell}} {\boldsymbol{\ell}}}\|^{\frac{1}{2}}}{\| \mathbf{\Lambda}_{\boldsymbol{t}\boldsymbol{t}}\|^{\frac{1}{2}}} 2^{d_{{\boldsymbol{\ell}}} - d_{\boldsymbol{t}}} \left( \frac{\sqrt{n}}{q_1} \exp \left\{ -\frac{n\varepsilon_0^3 s_n^2}{384K}\right\} \right)^{d_{\boldsymbol{t}} - d_{\boldsymbol{\ell}} }\\
&= \frac{\| \mathbf{\Lambda}_{{\boldsymbol{\ell}} {\boldsymbol{\ell}}}\|^{\frac{1}{2}}}{\| \mathbf{\Lambda}_{\boldsymbol{t}\boldsymbol{t}}\|^{\frac{1}{2}}} 2^{d_{{\boldsymbol{\ell}}} - d_{\boldsymbol{t}}} \left( \frac{\sqrt{n}}{q_1} \exp \left\{ -2a_1 n s_n^2\right\} \right)^{d_{\boldsymbol{t}} - d_{\boldsymbol{\ell}} }
\end{split}
\end{equation*}
by Assumption 5, for all large $n$,
\begin{equation*}
\begin{split}
PR({\boldsymbol{\ell}}, \boldsymbol{t}) & \leq \frac{\| \mathbf{\Lambda}_{{\boldsymbol{\ell}} {\boldsymbol{\ell}}}\|^{\frac{1}{2}}}{\| \mathbf{\Lambda}_{\boldsymbol{t}\boldsymbol{t}}\|^{\frac{1}{2}}} 2^{d_{{\boldsymbol{\ell}}} - d_{\boldsymbol{t}}} \left( \frac{\sqrt{n}}{q_1} \exp \left\{ -\log n -2a_2d_{\boldsymbol{t}\log p}\right\} \right)^{d_{\boldsymbol{t}} - d_{\boldsymbol{\ell}} }\\
&= \frac{\| \mathbf{\Lambda}_{{\boldsymbol{\ell}} {\boldsymbol{\ell}}}\|^{\frac{1}{2}}}{\| \mathbf{\Lambda}_{\boldsymbol{t}\boldsymbol{t}}\|^{\frac{1}{2}}} 2^{d_{{\boldsymbol{\ell}}} - d_{\boldsymbol{t}}} \left( \frac{p^{-2a_2 d_{\boldsymbol{t}}}}{\sqrt{n}q_1}  \right)^{d_{\boldsymbol{t}} - d_{\boldsymbol{\ell}} }\\
&= \frac{\| \mathbf{\Lambda}_{{\boldsymbol{\ell}} {\boldsymbol{\ell}}}\|^{\frac{1}{2}}}{\| \mathbf{\Lambda}_{\boldsymbol{t}\boldsymbol{t}}\|^{\frac{1}{2}}} 2^{d_{{\boldsymbol{\ell}}} - d_{\boldsymbol{t}}} \left( \frac{p^{-a_2 d_{\boldsymbol{t}}}}{\sqrt{n}}  \right)^{d_{\boldsymbol{t}} - d_{\boldsymbol{\ell}} } \to 0 \quad \text{as} \quad n \to \infty.
\end{split}
\end{equation*}
\end{proof}

\begin{Lemma} \label{Lemma 6}
Suppose ${\boldsymbol{\ell}} \supset \boldsymbol{t}$, and $d_{\boldsymbol{\ell}} < \tau_n$ then, under Assumptions 1-6,
\begin{equation*}
PR({\boldsymbol{\ell}}, \boldsymbol{t}) \to 0, \quad \quad \quad \text{as} \quad n \to \infty . 
\end{equation*}
\end{Lemma}

\begin{proof}
In this case,
\begin{equation*}
\begin{split}
PR(\boldsymbol{\ell}, \boldsymbol{t}) &\leq \frac{\|\mathbf{\Lambda}_{{\boldsymbol{\ell}} {\boldsymbol{\ell}}} \|^{\frac{1}{2}}}{\| \mathbf{\Lambda}_{\boldsymbol{t}\boldsymbol{t}}\|^{\frac{1}{2}}} \frac{\left(2q_1\right)^{d_{\boldsymbol{\ell}} - d_{\boldsymbol{t}}}}{\| \left( n\mathbf{\Upsilon} + \mathbf{\Lambda}\right)_{{\boldsymbol{\ell}} | \boldsymbol{t}}\|^{\frac{1}{2}}}\\
&\exp \left\{ \frac{1}{2} \left[ \left( n\mathbf{\Upsilon + \mathbf{\Lambda}}\right)_{{\boldsymbol{\ell}} {\boldsymbol{\ell}}} \mathbf{\Theta}_{{\boldsymbol{\ell}}}^0 + n{\hat{\boldsymbol{a}}}_{\boldsymbol{\ell}} \right]' \left(n\mathbf{\Upsilon} + \mathbf{\Lambda} \right)_{{\boldsymbol{\ell}} {\boldsymbol{\ell}}}^{-1} \left[ \left( n\mathbf{\Upsilon + \mathbf{\Lambda}}\right)_{{\boldsymbol{\ell}} {\boldsymbol{\ell}}} \mathbf{\Theta}_{{\boldsymbol{\ell}}}^0 + n{\hat{\boldsymbol{a}}}_{\boldsymbol{\ell}} \right] \right\}.
\end{split}
\end{equation*}
Now, we note that
\begin{equation*}
\left( n\mathbf{\Upsilon + \mathbf{\Lambda}}\right)_{{\boldsymbol{\ell}} {\boldsymbol{\ell}}} \mathbf{\Theta}_{{\boldsymbol{\ell}}}^0 + n{\hat{\boldsymbol{a}}}_{\boldsymbol{\ell}} = n\left( \mathbf{\Upsilon} \mathbf{\Theta}^0 + {{\boldsymbol{a}}}\right)_{\boldsymbol{\ell}} + \mathbf{\Lambda}_{{\boldsymbol{\ell}} {\boldsymbol{\ell}}} \mathbf{\Theta}_{{\boldsymbol{\ell}}}^0, + n \left( {\hat{\boldsymbol{a}}} - {{\boldsymbol{a}}}\right)_{\boldsymbol{\ell}}
\end{equation*}
each entry of the above vector in absolute value is smaller than 
\begin{equation*}
nc\sqrt{\frac{\log p}{n}} + \frac{\| \mathbf{\Lambda} \|_{\max}}{\varepsilon}_0 + \frac{3CK^2}{\varepsilon_0} \sqrt{\frac{\log p}{n}}\leq 2nc \sqrt{\frac{\log p}{n}},
\end{equation*}
hence, by Lemma \ref{Lemma 3},
\begin{equation*}
\begin{split}
&\left[ \left( n\mathbf{\Upsilon + \mathbf{\Lambda}}\right)_{{\boldsymbol{\ell}} {\boldsymbol{\ell}}} \mathbf{\Theta}_{{\boldsymbol{\ell}}}^0 + n{\hat{\boldsymbol{a}}}_{\boldsymbol{\ell}} \right]' \left(n\mathbf{\Upsilon} + \mathbf{\Lambda} \right)_{{\boldsymbol{\ell}} {\boldsymbol{\ell}}}^{-1} \left[ \left( n\mathbf{\Upsilon + \mathbf{\Lambda}}\right)_{{\boldsymbol{\ell}} {\boldsymbol{\ell}}} \mathbf{\Theta}_{{\boldsymbol{\ell}}}^0 + n{\hat{\boldsymbol{a}}}_{\boldsymbol{\ell}} \right] \\
& \leq \frac{1}{nK\varepsilon_0} d_{\boldsymbol{\ell}} \frac{4n^2 c^2 \log p}{n} = \frac{4c^2 d_{\boldsymbol{\ell}} \log p}{K\varepsilon_0}.
\end{split}
\end{equation*}
Hence,
\begin{equation}\label{Lemma S6 proof 1}
\begin{split}
PR(\boldsymbol{\ell}, \boldsymbol{t}) &\leq \frac{\|\mathbf{\Lambda}_{{\boldsymbol{\ell}} {\boldsymbol{\ell}}} \|^{\frac{1}{2}}}{\| \mathbf{\Lambda}_{\boldsymbol{t}\boldsymbol{t}}\|^{\frac{1}{2}}} \frac{\left(2q_1\right)^{d_{\boldsymbol{\ell}} - d_t}}{\left(\frac{nK\varepsilon_0}{2}\right)^{\frac{d_{\boldsymbol{\ell}}}{2}}} \exp \left\{\frac{2c^2}{K\varepsilon_0}d_{\boldsymbol{\ell}} \log p \right\} = \frac{2^{d_{\boldsymbol{\ell}}}\|\mathbf{\Lambda}_{{\boldsymbol{\ell}} {\boldsymbol{\ell}}} \|^{\frac{1}{2}}}{\| \mathbf{\Lambda}_{\boldsymbol{t}\boldsymbol{t}}\|^{\frac{1}{2}}} \frac{q_1^{d_{\boldsymbol{\ell}} - d_t} p^{\frac{2c^2 d_{\boldsymbol{\ell}}}{K\varepsilon_0}}}{\left(\frac{nK\varepsilon_0}{2}\right)^{\frac{d_{\boldsymbol{\ell}}}{2}}}.
\end{split} 
\end{equation}
Note that by Assumption 6, for all large $n$, $q_1 = p^{-\frac{8c^2d_{\boldsymbol{t}}}{K\varepsilon_0}} \leq p^{-\frac{8c^2}{K\varepsilon_0}}$. Therefore, if $ d_{\boldsymbol{\ell}} > 2d_{\boldsymbol{t}}$,
\begin{equation}\label{Lemma S6 proof 2}
\begin{split}
q_1^{d_{\boldsymbol{\ell}} - d_t} p^{\frac{2c^2 d_{\boldsymbol{\ell}}}{K\varepsilon_0}} \leq q_1^{\frac{d_{\boldsymbol{\ell}}}{2} } p^{\frac{2c^2 d_{\boldsymbol{\ell}}}{K\varepsilon_0}} \leq \left( p^{\frac{-8c^2 }{K\varepsilon_0}} p^{\frac{4c^2 }{K\varepsilon_0}} \right)^{\frac{d_{\boldsymbol{\ell}}}{2}} \leq \left( p^{\frac{-4c^2 }{K\varepsilon_0}}  \right)^{\frac{d_{\boldsymbol{\ell}}}{2}} \leq \left( p^{\frac{-4c^2 }{K\varepsilon_0}}  \right)^{\frac{d_{\boldsymbol{\ell}} - d_{\boldsymbol{t}}}{2}}
\end{split}
\end{equation}
and, in the case of $d_{\boldsymbol{t}} < d_{\boldsymbol{\ell}} < 2d_{\boldsymbol{t}}$,
\begin{equation} \label{Lemma S6 proof 3}
\begin{split}
q_1^{d_{\boldsymbol{\ell}} - d_t} p^{\frac{2c^2 d_{\boldsymbol{\ell}}}{K\varepsilon_0}} \leq q_1 p^{\frac{2c^2 d_{\boldsymbol{\ell}}}{K\varepsilon_0}} \leq p^{-\frac{8c^2d_{\boldsymbol{t}}}{K\varepsilon_0}} p^{\frac{2c^2 d_{\boldsymbol{\ell}}}{K\varepsilon_0}} \leq p^{-\frac{8c^2d_{\boldsymbol{t}}}{K\varepsilon_0}} p^{\frac{4c^2 d_{\boldsymbol{t}}}{K\varepsilon_0}} \leq p^{-\frac{4c^2d_{\boldsymbol{t}}}{K\varepsilon_0}} \leq \left( p^{\frac{-4c^2 }{K\varepsilon_0}}  \right)^{\frac{d_{\boldsymbol{\ell}} - d_{\boldsymbol{t}}}{2}}
\end{split}
\end{equation}
thus, by (\ref{Lemma S6 proof 1}), (\ref{Lemma S6 proof 1}), and (\ref{Lemma S6 proof 3}) we have that
\begin{equation*}
\begin{split}
PR(\boldsymbol{\ell}, \boldsymbol{t}) &\leq \frac{2^{d_{\boldsymbol{\ell}}}\|\mathbf{\Lambda}_{{\boldsymbol{\ell}} {\boldsymbol{\ell}}} \|^{\frac{1}{2}}}{\| \mathbf{\Lambda}_{\boldsymbol{t}\boldsymbol{t}}\|^{\frac{1}{2}}} \frac{q_1^{d_{\boldsymbol{\ell}} - d_t} p^{\frac{2c^2 d_{\boldsymbol{\ell}}}{K\varepsilon_0}}}{\left(\frac{nK\varepsilon_0}{2}\right)^{\frac{d_{\boldsymbol{\ell}}}{2}}} \leq \frac{2^{\frac{3d_{\boldsymbol{\ell}}}{2}}\|\mathbf{\Lambda}_{{\boldsymbol{\ell}} {\boldsymbol{\ell}}} \|^{\frac{1}{2}}}{\| \mathbf{\Lambda}_{\boldsymbol{t}\boldsymbol{t}}\|^{\frac{1}{2}}} \left( \frac{ p^{\frac{-4c^2 }{K\varepsilon_0}}}{nK\varepsilon_0} \right)^{\frac{d_{\boldsymbol{\ell}} - d_{\boldsymbol{t}}}{2}} \to 0 \quad \text{as} \quad n \to 0.
\end{split} 
\end{equation*}
\end{proof}

\begin{Lemma} \label{Lemma 7}
Suppose ${\boldsymbol{\ell}} \supset \boldsymbol{t}$, and $d_{\boldsymbol{\ell}} > \tau_n$ then, under Assumptions 1-6,
\begin{equation*}
PR({\boldsymbol{\ell}}, \boldsymbol{t}) \to 0, \quad \quad \quad \text{as} \quad n \to \infty . 
\end{equation*}
\end{Lemma}

\begin{proof}
When ${\boldsymbol{\ell}} \supset t$, 
\begin{equation*}
\begin{split}
PR\left({\boldsymbol{\ell}}, \boldsymbol{t}\right) &\leq \frac{\| \mathbf{\Lambda}_{{\boldsymbol{\ell}} {\boldsymbol{\ell}}}\|^{\frac{1}{2}}}{\| \mathbf{\Lambda}_{tt}\|^{\frac{1}{2}}} \frac{q_2^{d_{\boldsymbol{\ell}}}\left(1-q_2\right)^{\binom{p}{2} - d_{\boldsymbol{\ell}}}}{q_1^{d_{\boldsymbol{t}}}\left(1-q_1\right)^{\binom{p}{2} - d_{\boldsymbol{t}}}} \frac{1}{\| \left(n \mathbf{\Upsilon} + \mathbf{\Lambda}\right)_{{\boldsymbol{\ell}} | \boldsymbol{t}}\|^{\frac{1}{2}}} \\
&\exp \left\{ \frac{1}{2} \left[ \left( n\mathbf{\Upsilon + \mathbf{\Lambda}}\right)_{{\boldsymbol{\ell}} {\boldsymbol{\ell}}} \mathbf{\Theta}_{{\boldsymbol{\ell}}}^0 + n{\hat{\boldsymbol{a}}}_{\boldsymbol{\ell}} \right]' \left(n\mathbf{\Upsilon} + \mathbf{\Lambda} \right)_{{\boldsymbol{\ell}} {\boldsymbol{\ell}}}^{-1} \left[ \left( n\mathbf{\Upsilon + \mathbf{\Lambda}}\right)_{{\boldsymbol{\ell}} {\boldsymbol{\ell}}} \mathbf{\Theta}_{{\boldsymbol{\ell}}}^0 + n{\hat{\boldsymbol{a}}}_{\boldsymbol{\ell}} \right] \right\},
\end{split}
\end{equation*}
similar to the argument in Lemma \ref{Lemma 6}, each entry of the vector $ \left( n\mathbf{\Upsilon + \mathbf{\Lambda}}\right)_{{\boldsymbol{\ell}} {\boldsymbol{\ell}}} \mathbf{\Theta}_{{\boldsymbol{\ell}}}^0 + n{\hat{\boldsymbol{a}}}_{\boldsymbol{\ell}}$, in absolute value, is smaller than $2nc \sqrt{\frac{\log p}{n}}$. Now, since $\mathbf{\Upsilon}$ is non-negative definite (note that in the case of $d_{\boldsymbol{\ell}} > \tau_n$, $\mathbf{\Upsilon}$ is not necessarily positive definite) we have that $\text{eig}_{\min}\left( n\mathbf{\Upsilon} + \mathbf{\Lambda}\right)_{{\boldsymbol{\ell}}{\boldsymbol{\ell}}} \geq \text{eig}_{\min}\left(\mathbf{\Lambda} \right) = \| \mathbf{\Lambda}\|_{\min}$, hence for large enough $n$
\begin{equation*}
\begin{split}
PR\left({\boldsymbol{\ell}}, \boldsymbol{t} \right) &\leq \frac{\| \mathbf{\Lambda}_{{\boldsymbol{\ell}} {\boldsymbol{\ell}}}\|^{\frac{1}{2}}}{\| \mathbf{\Lambda}_{\boldsymbol{t}\boldsymbol{t}}\|^{\frac{1}{2}}} \left( \frac{q_2}{1-q_2}\right)^{d_{\boldsymbol{\ell}} - d_{\boldsymbol{t}}} \frac{q_2^{d_{\boldsymbol{t}}}\left(1-q_2\right)^{\binom{p}{2} - d_t}}{q_1^{d_{\boldsymbol{t}}}\left(1-q_1\right)^{\binom{p}{2} - d_{\boldsymbol{t}}}} \exp \left\{ \frac{4 c^2 n^2 d_{\boldsymbol{\ell}} \log p}{n \|\mathbf{\Lambda} \|_{\min}}\right\},
\end{split}
\end{equation*}
now, since the function $q^{d_{\boldsymbol{t}}}(1-q)^{\binom{p}{2} - d_{\boldsymbol{t}}}$ is globally maximized at $\hat{q}=\frac{d_{\boldsymbol{t}}}{\binom{p}{2}}$ and $q_2 < q_1 < \hat{q}$,
\begin{equation*}
\frac{q_2^{d_{\boldsymbol{t}}}\left(1-q_2\right)^{\binom{p}{2} - d_{\boldsymbol{t}}}}{q_1^{d_{\boldsymbol{t}}}\left(1-q_1\right)^{\binom{p}{2} - d_{\boldsymbol{t}}}} \leq 1,
\end{equation*}
hence,
\begin{equation*}
\begin{split}
PR\left({\boldsymbol{\ell}}, \boldsymbol{t} \right) &\leq \frac{\| \mathbf{\Lambda}_{{\boldsymbol{\ell}} {\boldsymbol{\ell}}}\|^{\frac{1}{2}}}{\| \mathbf{\Lambda}_{\boldsymbol{t}\boldsymbol{t}}\|^{\frac{1}{2}}} \left( 2q_2\right)^{d_{\boldsymbol{\ell}} - d_{\boldsymbol{t}}}  \exp \left\{ \frac{4 c^2 n d_{\boldsymbol{\ell}} \log p}{ \|\mathbf{\Lambda} \|_{\min}}\right\},
\end{split}
\end{equation*}
since $d_{\boldsymbol{\ell}} > \tau$, and by Assumption 2 $d_{\boldsymbol{t}} < \frac{\tau}{2}$ , we have that,
\begin{equation*}
d_{\boldsymbol{\ell}} - d_{\boldsymbol{t}} \geq \frac{d_{\boldsymbol{\ell}}}{2},
\end{equation*}
hence,
\begin{equation*}
\begin{split}
PR\left({\boldsymbol{\ell}}, \boldsymbol{t} \right) & \leq \frac{\| \mathbf{\Lambda}_{{\boldsymbol{\ell}} {\boldsymbol{\ell}}}\|^{\frac{1}{2}}}{\| \mathbf{\Lambda}_{\boldsymbol{t}\boldsymbol{t}}\|^{\frac{1}{2}}} \left( 2q_2\right)^{\frac{d_{\boldsymbol{\ell}}}{2}}  \exp \left\{ \frac{4 c^2 n d_{\boldsymbol{\ell}} \log p}{ \|\mathbf{\Lambda} \|_{\min}}\right\}\\
& \leq \frac{2^{\frac{d_{\boldsymbol{\ell}}}{2}}\| \mathbf{\Lambda}_{{\boldsymbol{\ell}} {\boldsymbol{\ell}}}\|^{\frac{1}{2}}}{\| \mathbf{\Lambda}_{\boldsymbol{t}\boldsymbol{t}}\|^{\frac{1}{2}}} \left[ q_2  \exp \left\{ \frac{8 c^2 n \log p }{\|\mathbf{\Lambda} \|_{\min}}\right\}\right]^{\frac{d_{\boldsymbol{\ell}}}{2}} \\
& \leq \frac{2^{\frac{d_{\boldsymbol{\ell}}}{2}}\| \mathbf{\Lambda}_{{\boldsymbol{\ell}} {\boldsymbol{\ell}}}\|^{\frac{1}{2}}}{\| \mathbf{\Lambda}_{\boldsymbol{t}\boldsymbol{t}}\|^{\frac{1}{2}}} \left[  p^{-\frac{16nc^2}{\|\mathbf{\Lambda}\|_{\min}}} p^{\frac{8nc^2}{\|\mathbf{\Lambda}\|_{\min}}} \right]^{\frac{d_{\boldsymbol{\ell}}}{2}}\\
& \leq \frac{2^{\frac{d_{\boldsymbol{\ell}}}{2}}\| \mathbf{\Lambda}_{{\boldsymbol{\ell}} {\boldsymbol{\ell}}}\|^{\frac{1}{2}}}{\| \mathbf{\Lambda}_{\boldsymbol{t}\boldsymbol{t}}\|^{\frac{1}{2}}} \left[  p^{-\frac{8nc^2}{\|\mathbf{\Lambda}\|_{\min}}} \right]^{\frac{d_{\boldsymbol{\ell}}}{2}} \\
&\leq \frac{2^{\frac{d_{\boldsymbol{\ell}}}{2}}\| \mathbf{\Lambda}_{{\boldsymbol{\ell}} {\boldsymbol{\ell}}}\|^{\frac{1}{2}}}{\| \mathbf{\Lambda}_{\boldsymbol{t}\boldsymbol{t}}\|^{\frac{1}{2}}} \left[  p^{-\frac{8nc^2}{\|\mathbf{\Lambda}\|_{\min}}} \right]^{\frac{d_{\boldsymbol{\ell}}- d_{\boldsymbol{t}}}{2}}\to 0, \quad \quad \text{as} \quad n \to \infty.
\end{split}
\end{equation*}

\end{proof}
Now, let,
\begin{equation}\label{f_n}
\begin{split}
f_n = &\max\limits_{\boldsymbol{\ell}\in \mathcal{L}}\left\{ 2\frac{\| \mathbf{\Lambda}_{{\boldsymbol{\ell}} {\boldsymbol{\ell}}}\|^{\frac{1}{2}}}{\| \mathbf{\Lambda}_{\boldsymbol{t}\boldsymbol{t}}\|^{\frac{1}{2}}}, 2^{\frac{3d_{\boldsymbol{\ell}}}{2}}\frac{\| \mathbf{\Lambda}_{{\boldsymbol{\ell}} {\boldsymbol{\ell}}}\|^{\frac{1}{2}}}{\| \mathbf{\Lambda}_{\boldsymbol{t}\boldsymbol{t}}\|^{\frac{1}{2}}}, 2^{3\frac{3d_{\boldsymbol{\ell}}}{2}}\frac{\| \mathbf{\Lambda}_{{\boldsymbol{\ell}} {\boldsymbol{\ell}}}\|^{\frac{1}{2}}}{\| \mathbf{\Lambda}_{\boldsymbol{t}\boldsymbol{t}}\|^{\frac{1}{2}}} \right\} \\
\times &\max \left\{  \left( \frac{p^{- \frac{8c^2 d_{\boldsymbol{t}}}{K\varepsilon_0}}}{\sqrt{n}}  \right), \left(  p^{-\frac{8nc^2}{\|\mathbf{\Lambda}\|_{\min}}} \right)^{\frac{1}{2}},  \left( \frac{ p^{\frac{-4c^2 }{K\varepsilon_0}}}{nK\varepsilon_0} \right)^{\frac{1}{2}}\right\}.
\end{split}
\end{equation}
\begin{Lemma}\label{Lemma 8}
Let $\boldsymbol{\ell} \in \mathcal{L}$ such that $\boldsymbol{\ell} \not\subset \boldsymbol{t}$, $\boldsymbol{t} \not\subset \boldsymbol{\ell}$, and $\boldsymbol{\ell} \neq \boldsymbol{t}$, then under Assumptions 1-6, for sufficiently large n, 
\begin{equation*}
PR({\boldsymbol{\ell}}, \boldsymbol{t}) \to 0, \quad \quad \quad \text{as} \quad n \to \infty . 
\end{equation*}
\end{Lemma}
\begin{proof}
Denote the shared part between $\boldsymbol{\ell}$ and $\boldsymbol{t}$ by $\boldsymbol{h} = \boldsymbol{\ell} \cap \boldsymbol{t}$. Then,
\begin{equation*}
PR({\boldsymbol{\ell}}, \boldsymbol{t}) = PR({\boldsymbol{\ell}},\boldsymbol{h})\times PR(\boldsymbol{h}, {\boldsymbol{t}}),
\end{equation*}
Since, $\boldsymbol{h} \subset \boldsymbol{\ell}$ and $\boldsymbol{h} \subset \boldsymbol{t}$, by Lemma \ref{Lemma 5}, \ref{Lemma 6}, and \ref{Lemma 7}, we have that 
\begin{equation*}
\begin{split}
PR({\boldsymbol{\ell}}, \boldsymbol{t}) &\leq f_n^{d_{\boldsymbol{\ell}} - d_{\boldsymbol{h}}} f_n^{d_{\boldsymbol{t}} - d_{\boldsymbol{h}}} = f_{n}^{d_{\boldsymbol{\ell}} + d_{\boldsymbol{t}} - 2 d_{\boldsymbol{h}}} \to 0, \quad \text{as} \quad n \to \infty.
\end{split}
\end{equation*}
\end{proof}

\begin{Corollary}
Let $\boldsymbol{\ell} \in \mathcal{L}$ such that $\boldsymbol{\ell} \not\subset \boldsymbol{t}$, $\boldsymbol{t} \not\subset \boldsymbol{\ell}$, and $\boldsymbol{\ell} \neq \boldsymbol{t}$. Denote the total number of disagreements by $D(\boldsymbol{\ell}, \boldsymbol{t})$. Then, under Assumptions 1-6, for sufficiently large n,
\begin{equation*}
PR(\boldsymbol{\ell}, \boldsymbol{t}) \leq f_n^{D(\boldsymbol{\ell}, \boldsymbol{t})}, \quad \forall \boldsymbol{\ell} \in \mathcal{L}.
\end{equation*}
\end{Corollary}
\begin{proof}
The proof is straightforward application of Lemmas \ref{Lemma 5}, \ref{Lemma 6}, \ref{Lemma 7} and \ref{Lemma 8}.
\end{proof}

\begin{proof}[Proof of Theorem 1]
\begin{equation}
\begin{split}
\frac{1 - P \left\{ \mathbf{\Theta}\in\mathcal{M}_{\boldsymbol{t}} | \hat{\mathbf{\Delta}}, \mathcal{Y}\right\}}{P \left\{ \mathbf{\Theta}\in\mathcal{M}_{\boldsymbol{t}} | \hat{\mathbf{\Delta}}, \mathcal{Y}\right\}}& = \sum\limits_{\boldsymbol{\ell}\neq \boldsymbol{t}} PR(\boldsymbol{\ell}, \boldsymbol{t}) \\
&= \sum\limits_{\boldsymbol{\ell}\neq \boldsymbol{t}} \sum\limits_{j=1}^{\binom{p}{2}} PR(\boldsymbol{\ell}, \boldsymbol{t})I_{\{D(\boldsymbol{\ell}, \boldsymbol{t})=j\}}\\
&\leq \sum\limits_{j=1}^{\binom{p}{2}} \binom{\binom{p}{2}}{j} f_n^{j}\\
& \leq  \sum\limits_{j=1}^{\binom{p}{2}} {\binom{p}{2}}^j f_n^{j}\\
& \leq \sum\limits_{j=1}^{p^2} \left(p^2 f_n\right)^j \\
& \leq \frac{p^2 f_n}{1 - p^2 f_n} \to 0, \quad \text{as} \quad n \to \infty.
\end{split}
\end{equation}
The last two inequalities follow from the fact that $p^2 f_n < 1$ and $p^2 f_n \to 0$. Which follows from (\ref{f_n}) and choice of $\varepsilon_0 = \frac{c^2}{K\tilde{\varepsilon_0}}$. 
\end{proof}

\begin{proof}[Proof of Theorem 2]
For simplicity in notation, let $\epsilon_n = \sqrt{\frac{d_{\boldsymbol{t}} \log p}{n}}$.  First note that for any constant $G$,

\begin{equation*}
\begin{split}
\mathbb{E}_0& \left[ P  \left(\| \mathbf{\Theta} - \mathbf{\Theta}^0\|_2 > G \epsilon_n |  \hat{\mathbf{\Delta}}, \mathcal{Y} \right)\right] \\
&= \sum\limits_{\boldsymbol{\ell} \in \mathcal{L}} \mathbb{E}_0 \left[ P \left(\| \mathbf{\Theta} - \mathbf{\Theta}^0\|_2 > G \epsilon_n | \boldsymbol{\ell}, \hat{\mathbf{\Delta}}, \mathcal{Y} \right) P\left( \boldsymbol{\ell}| \hat{\mathbf{\Delta}}, \mathcal{Y} \right)\right] \\
& \leq \mathbb{E}_0 \left[ P \left(\| \mathbf{\Theta} - \mathbf{\Theta}^0\|_2 > G \epsilon_n | \boldsymbol{t}, \hat{\mathbf{\Delta}}, \mathcal{Y} \right)\right] + \mathbb{E}_0 \left[ \sum\limits_{\boldsymbol{\ell} \neq \boldsymbol{t}} P \left(\boldsymbol{\ell}| \hat{\mathbf{\Delta}}, \mathcal{Y}\right)\right].
\end{split}
\end{equation*}
By Theorem 1, it is sufficient to prove $\mathbb{E}_0 \left[ P \left(\| \mathbf{\Theta} - \mathbf{\Theta}^0\|_2 > G \epsilon_n | \boldsymbol{t}, \hat{\mathbf{\Delta}}, \mathcal{Y} \right)\right] \quad \to 0$ as $n \to \infty$. First note that
\begin{equation*}
\begin{split}
& P\left(\| \mathbf{\Theta} - \mathbf{\Theta}^0\|_2 > G \epsilon_n | \boldsymbol{t}, \hat{\mathbf{\Delta}}, \mathcal{Y} \right) = P \left(\| \mathbf{\Theta}_{\boldsymbol{t}} - \mathbf{\Theta}^0_{\boldsymbol{t}} \|_2 > G \epsilon_n | \boldsymbol{t}, \hat{\mathbf{\Delta}}, \mathcal{Y}\right),
\end{split}
\end{equation*}
now, from (4.4), it is easy to see that
\begin{equation} \label{Proof of theorem 2, 1}
\begin{split}
\left(\mathbf{\Theta}_{\boldsymbol{t}}  | \boldsymbol{t}, \hat{\mathbf{\Delta}}, \mathcal{Y}\right)\sim MVN\left[\mathfrak{M} , \left(n \mathbf{\Upsilon} + \mathbf{\Lambda}\right)_{\boldsymbol{t}\boldsymbol{t}}^{-1}\right],
\end{split}
\end{equation}
with $\mathfrak{M} = - n \left(n \mathbf{\Upsilon} + \mathbf{\Lambda}\right)_{\boldsymbol{t}\boldsymbol{t}}^{-1} \boldsymbol{a}_{\boldsymbol{t}}$. Hence,
\begin{equation}\label{Proof of theorem 2, 2}
\begin{split}
\mathbb{E}_{0} \left[P \left(\| \mathbf{\Theta} - \mathbf{\Theta}^0\|_2 > G \epsilon_n | \hat{\mathbf{\Delta}}, \mathcal{Y} \right) \right] &\leq \mathbb{E}_{0} \left[P \left(\| \mathbf{\Theta}_{\boldsymbol{t}} - \mathfrak{M}_{\boldsymbol{t}} \|_2 > \frac{G \epsilon_n }{2}| \boldsymbol{t}, \hat{\mathbf{\Delta}}, \mathcal{Y}\right) \right]\\
& + \mathbb{E}_{0} \left[P \left(\| \mathfrak{M}_{\boldsymbol{t}} - \mathbf{\Theta}^0_{\boldsymbol{t}} \|_2 > \frac{G \epsilon_n }{2}| \boldsymbol{t}, \hat{\mathbf{\Delta}}, \mathcal{Y}\right) \right].
\end{split}
\end{equation}
Now,
\begin{equation*}
\begin{split}
 P \left(\| \mathbf{\Theta}_{\boldsymbol{t}} - \mathfrak{M}_{\boldsymbol{t}} \|_2 > \frac{G \epsilon_n }{2}| \boldsymbol{t}, \hat{\mathbf{\Delta}}, \mathcal{Y}\right) &\leq  P \left( \max\limits_{1 \leq l \leq d_{d_{\boldsymbol{t}}}} | \left( \mathbf{\Theta}_{\boldsymbol{t}} - \mathfrak{M}_{\boldsymbol{t}} \right)_l | > \frac{G \epsilon_n }{2\sqrt{d_{\boldsymbol{t}}}}| \boldsymbol{t}, \hat{\mathbf{\Delta}}, \mathcal{Y} \right) \\
& \leq \sum\limits_{l=1}^{d_{\boldsymbol{t}}} P\left(  | \left( \mathbf{\Theta}_{\boldsymbol{t}} - \mathfrak{M}_{\boldsymbol{t}} \right)_l | > \frac{G \epsilon_n }{2\sqrt{d_{\boldsymbol{t}}}}| \boldsymbol{t}, \hat{\mathbf{\Delta}}, \mathcal{Y} \right),
\end{split}
\end{equation*}
denoting the $l^{\text{th}}$ element of $\left( \mathbf{\Theta}_{\boldsymbol{t}} - \mathfrak{M}_{\boldsymbol{t}} \right)$ by $\left( \mathbf{\Theta}_{\boldsymbol{t}} - \mathfrak{M}_{\boldsymbol{t}} \right)_l$ and in view of (\ref{Proof of theorem 2, 1}),

$$\left( \mathbf{\Theta}_{\boldsymbol{t}} - \mathfrak{M}_{\boldsymbol{t}} \right)_l \sim N(0, \rho_{l}),$$ where by Lemma \ref{Lemma 3}, $\rho_l \leq \text{eig}_{max}\left(n \mathbf{\Upsilon} + \mathbf{\Lambda}\right)_{\boldsymbol{t}\boldsymbol{t}}^{-1} = \frac{1}{\text{eig}_{min}\left(n \mathbf{\Upsilon}_{\boldsymbol{t}\boldsymbol{t}}\right)} \leq \frac{2}{nK\varepsilon_0}$. Therefore,

\begin{equation}\label{Proof of theorem 2, 3}
\begin{split}
 P \left(\| \mathbf{\Theta}_{\boldsymbol{t}} - \mathfrak{M}_{\boldsymbol{t}} \|_2 > \frac{G \epsilon_n }{2}| \boldsymbol{t}, \hat{\mathbf{\Delta}}, \mathcal{Y} \right) & \leq \sum\limits_{l=1}^{d_{\boldsymbol{t}}} P \left(  | \left( \mathbf{\Theta}_{\boldsymbol{t}} - \mathfrak{M}_{\boldsymbol{t}} \right)_l | > \frac{G \epsilon_n }{2\sqrt{d_{\boldsymbol{t}}}}| \boldsymbol{t}, \hat{\mathbf{\Delta}}, \mathcal{Y} \right) \\
& \leq d_{\boldsymbol{t}} P \left(  | \left( \mathbf{\Theta}_{\boldsymbol{t}} - \mathfrak{M}_{\boldsymbol{t}} \right)_1 | > \frac{G \epsilon_n }{2\sqrt{d_{\boldsymbol{t}}}}| \boldsymbol{t}, \hat{\mathbf{\Delta}}, \mathcal{Y} \right)\\
& \leq d_{\boldsymbol{t}} P\left(  |Z| > \frac{\sqrt{nK\varepsilon_0}G\epsilon_n}{2\sqrt{2d_{\boldsymbol{t}}}} \right)\\
& \leq d_{{\boldsymbol{t}}} e^{-\frac{nK\varepsilon_0 G^2 \epsilon_n^2}{8d_{\boldsymbol{t}}}},
\end{split}
\end{equation}
where the last inequality was achieved using the Mills ratio inequality. Moving onto the second term in (\ref{Proof of theorem 2, 1}), and recalling that we are restricted to the even $C_{1,n} \cap C_{2,n}$, for all large $n$ we have that
\begin{equation*}
\begin{split}
\| \mathfrak{M}_{\boldsymbol{t}} - \mathbf{\Theta}^0_{\boldsymbol{t}} \|_2 &= \|n \left(n \mathbf{\Upsilon} + \mathbf{\Lambda}\right)_{\boldsymbol{t}\boldsymbol{t}}^{-1} \boldsymbol{a}_{\boldsymbol{t}} + \mathbf{\Theta}^0_{\boldsymbol{t}} \|_2 \\
& = \|\left(\mathbf{\Upsilon} + \frac{1}{n}\mathbf{\Lambda}\right)_{\boldsymbol{t}\boldsymbol{t}}^{-1} \boldsymbol{a}_{\boldsymbol{t}} + \mathbf{\Theta}^0_{\boldsymbol{t}} \|_2 \\
& \leq \frac{1}{\text{eig}_{\min}\left(\mathbf{\Upsilon} + \frac{1}{n}\mathbf{\Lambda}\right)_{\boldsymbol{t}\boldsymbol{t}}}\| \boldsymbol{a}_{\boldsymbol{t}} + \left(\mathbf{\Upsilon} + \frac{1}{n}\mathbf{\Lambda}\right)_{tt}\mathbf{\Theta}^0_{\boldsymbol{t}}\|_2\\
&\leq \frac{4}{3K\varepsilon_0}\|\mathbf{\Upsilon}_{\boldsymbol{t}\boldsymbol{t}}\mathbf{\Theta}^0_{\boldsymbol{t}} + \boldsymbol{a}_{\boldsymbol{t}} + \frac{1}{n}\mathbf{\Lambda}_{\boldsymbol{t}\boldsymbol{t}}\mathbf{\Theta}^0_{\boldsymbol{t}}\|_2 \\ 
& \leq \frac{4}{3K\varepsilon_0} \left( \sqrt{d_{\boldsymbol{t}}}\|\mathbf{\Upsilon}_{\boldsymbol{t}\boldsymbol{t}}\mathbf{\Theta}^0_{\boldsymbol{t}} + \boldsymbol{a}_{\boldsymbol{t}}\|_{\max} + \|\frac{1}{n}\mathbf{\Lambda}_{\boldsymbol{t}\boldsymbol{t}}\mathbf{\Theta}^0_{\boldsymbol{t}}\|_2 \right) \\
& \leq \frac{4}{3K\varepsilon_0} \left(c_0 \sqrt{\frac{d_{\boldsymbol{t}} \log p}{n}} + \frac{\sqrt{d_{\boldsymbol{t}} }R_n\|\mathbf{\Lambda}\|_{\max}}{n} \right)\\
& \leq \frac{4}{3K\varepsilon_0} \left(c_0 \sqrt{\frac{d_{\boldsymbol{t}} \log p}{n}} + \frac{\sqrt{d_{\boldsymbol{t}} }\sqrt{n \log p}\|\mathbf{\Lambda}\|_{\max}}{n} \right)\\
& = \frac{4(c_0 + \|\mathbf{\Lambda}\|_{\max})}{3K\varepsilon_0} \sqrt{\frac{d_{\boldsymbol{t}} \log p}{n}},
\end{split}
\end{equation*}
where we used Lemma \ref{Lemma 3} to get the second inequality, used Lemma \ref{Lemma 4} with $c_0 = 2c \sqrt{K} + \frac{3CK^2}{\varepsilon_0}$ to get the fourth inequality and the last inequality is a direct application of the assumption. Hence, by taking $G = \frac{4(c_0 + \|\mathbf{\Lambda}\|_{\max})}{3K\varepsilon_0}$, we have that
\begin{equation}\label{Proof of theorem 2, 4}
\begin{split}
P \left(\| \mathfrak{M}_{\boldsymbol{t}} - \mathbf{\Theta}^0_{\boldsymbol{t}} \|_2 > \frac{G \epsilon_n }{2}| \boldsymbol{t}, \hat{\mathbf{\Delta}}, \mathcal{Y}\right) &= P\left(\| \mathfrak{M}_{\boldsymbol{t}} - \mathbf{\Theta}^0_{\boldsymbol{t}} \|_2 > \frac{G \epsilon_n }{2} \right)\\
&\leq 1- P\left(C_{1,n}\cap C_{2,n} \right) \to 0.
\end{split}
\end{equation}
Thus, by combining (\ref{Proof of theorem 2, 1}), (\ref{Proof of theorem 2, 2}), (\ref{Proof of theorem 2, 3}), and (\ref{Proof of theorem 2, 4}), we have that
\begin{equation}
\begin{split}
\mathbb{E}_{0} & \left[P\left(\| \mathbf{\Theta} - \mathbf{\Theta}^0\|_2 > G \epsilon_n | \hat{\mathbf{\Delta}}, \mathcal{Y} \right) \right] \\
&\leq  d_{{\boldsymbol{t}}} e^{-\frac{nK\varepsilon_0 G^2 \epsilon_n^2}{8d_{\boldsymbol{t}}}} + 0 = d_{{\boldsymbol{t}}} e^{-\frac{32c_0^2 \log p}{9 K \varepsilon_0}} \to 0 \quad \text{as} \quad n \to \infty,
\end{split}
\end{equation}
which completes the proof.
\end{proof}

\newpage
\newpage
\small

\end{document}